\documentclass[a4paper,UKenglish,cleveref, autoref]{lipics-v2021}

\title{Tight algorithms for connectivity problems parameterized by clique-width}
\titlerunning{Tight algorithms for connectivity problems parameterized by clique-width}

\author{Falko Hegerfeld}{Humboldt-Universit\"at zu Berlin, Germany}{hegerfeld@informatik.hu-berlin.de}{https://orcid.org/0000-0003-2125-5048}{Partially supported by DFG Emmy Noether-grant (KR 4286/1).}
\author{Stefan Kratsch}{Humboldt-Universit\"at zu Berlin, Germany}{kratsch@informatik.hu-berlin.de}{https://orcid.org/0000-0002-0193-7239}{}

\authorrunning{F. Hegerfeld and S. Kratsch}

\Copyright{Falko Hegerfeld and Stefan Kratsch}

\ccsdesc[500]{Mathematics of computing~Paths and connectivity problems}
\ccsdesc[300]{Theory of computation~Parameterized complexity and exact algorithms}
\ccsdesc[100]{Mathematics of computing~Combinatorial algorithms}

\keywords{Parameterized Complexity, Connectivity, Clique-width, Cut\&Count, Lower Bound}
\category{} 

\relatedversion{} 

\supplement{}

\nolinenumbers 

\hideLIPIcs  

\EventEditors{John Q. Open and Joan R. Access}
\EventNoEds{2}
\EventLongTitle{42nd Conference on Very Important Topics (CVIT 2016)}
\EventShortTitle{CVIT 2016}
\EventAcronym{CVIT}
\EventYear{2016}
\EventDate{December 24--27, 2016}
\EventLocation{Little Whinging, United Kingdom}
\EventLogo{}
\SeriesVolume{42}
\ArticleNo{23}

\usepackage{xspace}
\usepackage[linesnumbered,ruled,vlined]{algorithm2e}

\SetCommentSty{mycommstyle}
\SetKw{True}{true}
\SetKw{False}{false}

\theoremstyle{definition}
\newtheorem{dfn}{Definition}[section]

\theoremstyle{plain}
\newtheorem{thm}[dfn]{Theorem}
\newtheorem{cor}[dfn]{Corollary}
\newtheorem{lem}[dfn]{Lemma}

\newtheorem{cnj}[dfn]{Conjecture}

\crefname{dfn}{Definition}{Definitions}
\crefname{thm}{Theorem}{Theorems}
\crefname{cor}{Corollary}{Corollaries}
\crefname{lem}{Lemma}{Lemmata}
\crefname{prop}{Proposition}{Propositions}
\crefname{rem}{Remark}{Remarks}
\crefname{algorithm}{Algorithm}{Algorithms}
\crefname{algocf}{Algorithm}{Algorithms}
\crefname{obs}{Observation}{Observations}
\crefname{cnj}{Conjecture}{Conjectures}
\crefname{table}{Table}{Tables}

\crefformat{equation}{(#2#1#3)}

\usepackage{mathtools}

\newcommand{\Oh}{\mathcal{O}} 

\newcommand{\SETH}{SETH\xspace}
\newcommand{\CNFSETH}{CNF-SETH\xspace}

\newcommand{\FF}{\mathbb{F}}
\newcommand{\ZZ}{\mathbb{Z}}
\newcommand{\NN}{\mathbb{N}}

\newcommand{\eps}{\varepsilon}

\newcommand{\DS}{\textsc{Dominating Set}\xspace}

\newcommand{\OCT}{\textsc{Odd Cycle Transversal}\xspace}

\newcommand{\SAT}{\textsc{Satisfiability}\xspace}
\newcommand{\ST}{\textsc{Steiner Tree}\xspace}
\newcommand{\CVC}{\textsc{Connected Vertex Cover}\xspace}
\newcommand{\COCT}{\textsc{Connected Odd Cycle Transversal}\xspace}

\newcommand{\CDS}{\textsc{Connected Dominating Set}\xspace}
\newcommand{\FVS}{\textsc{Feedback Vertex Set}\xspace}

\newcommand{\VC}{\textsc{Vertex Cover}\xspace}
\newcommand{\IS}{\textsc{Independent Set}\xspace}
\newcommand{\HC}{\textsc{Hamiltonian Cycle}\xspace}

\DeclareMathOperator{\lcw}{lin-cw} 
\DeclareMathOperator{\cw}{cw} 
\DeclareMathOperator{\tw}{tw} 
\DeclareMathOperator{\poly}{poly} 
\DeclareMathOperator{\modtw}{mod-tw} 
\DeclareMathOperator{\modpw}{mod-pw} 

\newcommand{\sep}{:}

\newcommand{\symdiff}{\mathbin{\triangle}}

\newcommand{\cexpr}{\mu}
\newcommand{\tree}{T}
\newcommand{\augtree}{\hat{T}}

\newcommand{\intro}[1]{{#1}}
\newcommand{\union}{\oplus}
\newcommand{\relab}[2]{\rho_{#1 \rightarrow #2}}
\newcommand{\join}[2]{\eta_{#1,#2}}
\newcommand{\deadop}[1]{\bot_{#1}}
\newcommand{\lfct}{\mathtt{lab}}

\newcommand{\livelabels}[1]{{L^{live}_{#1}}}
\newcommand{\deadlabels}[1]{{L^{dead}_{#1}}}
\newcommand{\dead}{D}

\newcommand{\labels}{L}
\newcommand{\nlabels}{k}

\newcommand{\formula}{\sigma}
\newcommand{\tassign}{\tau}
\newcommand{\grpsize}{p}
\newcommand{\vgrpsize}{\beta}
\newcommand{\ngrps}{t}
\newcommand{\nvars}{n}
\newcommand{\nclss}{m}
\newcommand{\clss}{q}
\newcommand{\embedding}{\kappa}
\newcommand{\statemap}{\mathbf{state}}
\newcommand{\states}{\mathbf{States}}
\newcommand{\stateset}{\mathbf{S}}
\newcommand{\atoms}{\mathbf{Atoms}}
\newcommand{\bstates}{\Omega}
\newcommand{\rvertex}{\hat{r}}

\newcommand{\state}{\mathbf{s}}
\newcommand{\bolda}{\mathbf{a}}
\newcommand{\zero}{\mathbf{0}}
\newcommand{\one}{\mathbf{1}}
\newcommand{\two}{\mathbf{2}}

\newcommand{\conn}{\mathtt{conn}}
\newcommand{\sol}{\mathtt{sol}}
\newcommand{\dom}{\mathtt{dom}}

\newcommand{\powerset}[1]{\mathcal{P}(#1)}
\newcommand{\tfa}{\text{ for all }}

\newcommand{\algo}{\mathcal{A}}

\newcommand{\cfct}{\mathbf{c}}
\newcommand{\ctarget}{{\overline{c}}}
\newcommand{\budget}{{\overline{b}}}

\newcommand{\wmax}{N}
\newcommand{\wfct}{\mathbf{w}}
\newcommand{\wtarget}{{\overline{w}}}
\newcommand{\family}{\mathcal{F}}
\newcommand{\PP}{\mathbb{P}}

\newcommand{\rsols}{\mathcal{R}}
\newcommand{\sols}{\mathcal{S}}

\newcommand{\cuts}{\mathcal{C}}

\newcommand{\cc}{\mathtt{cc}}

\newcommand{\fixedvertex}{{v_*}}

\newcommand{\dpsols}[2]{\mathcal{A}_{#1}^{#2}}
\newcommand{\dppoly}[2]{A_{#1}^{#2}}
\DeclareMathOperator{\feas}{feas}

\DeclareMathOperator{\merge}{merge}
\newcommand{\restrict}[1]{{\big|_{#1}}}

\newcommand{\universe}{U}

\newcommand{\lattice}{\mathcal{L}}
\newcommand{\idxset}{I}

\usepackage{subcaption}
\usepackage{tikzit}

\tikzstyle{filled}=[fill=black, draw=black, shape=circle, inner sep=1pt, minimum size=5]
\tikzstyle{big empty}=[fill=white, draw=black, shape=rectangle, inner sep=2pt, minimum size=20pt]
\tikzstyle{small empty}=[fill=white, draw=black, shape=circle, inner sep=0pt, minimum size=20pt, line width=0.75pt]
\tikzstyle{rect}=[fill=black, draw=black, shape=rectangle, inner sep=1pt, minimum size=5pt]
\tikzstyle{big empty red}=[fill=white, draw={rgb,255: red,171; green,0; blue,60}, shape=rectangle, inner sep=2pt, minimum size=20pt]
\tikzstyle{mid empty}=[fill=white, draw=black, shape=rectangle, inner sep=2pt, minimum size=20pt]
\tikzstyle{mid empty red}=[fill=white, draw={rgb,255: red,171; green,0; blue,60}, shape=rectangle, inner sep=2pt, minimum size=20pt]
\tikzstyle{small empty red}=[fill={rgb,255: red,227; green,159; blue,178}, draw={rgb,255: red,171; green,0; blue,60}, shape=circle, inner sep=0pt, minimum size=20pt, line width=0.75pt]
\tikzstyle{red rect}=[fill={rgb,255: red,171; green,0; blue,60}, draw={rgb,255: red,171; green,0; blue,60}, shape=rectangle, minimum size=5pt]
\tikzstyle{big circle}=[fill=white, draw=black, shape=circle, line width=1pt, minimum size=25pt]
\tikzstyle{filled red}=[fill={rgb,255: red,171; green,0; blue,60}, draw={rgb,255: red,171; green,0; blue,60}, shape=circle, inner sep=1pt, minimum size=5pt]
\tikzstyle{state 0}=[fill=white, draw=black, shape=circle, line width=0.75pt, minimum size=10pt]
\tikzstyle{state L}=[fill={rgb,255: red,25; green,61; blue,182}, draw=black, shape=circle, line width=0.75pt, minimum size=10pt]
\tikzstyle{state R}=[fill={rgb,255: red,191; green,0; blue,64}, draw=black, shape=circle, line width=0.75pt, minimum size=10pt]
\tikzstyle{state unknown}=[fill=black, draw=black, shape=circle, line width=0.75pt, minimum size=10pt]
\tikzstyle{tiny empty}=[fill=white, draw=black, shape=circle, inner sep=0pt, minimum size=12pt, line width=0.75pt]
\tikzstyle{tiny filled}=[fill=black, draw=black, shape=circle, inner sep=0pt, minimum size=12pt, line width=0.75pt]
\tikzstyle{wide rectangle}=[fill=white, draw=black, shape=rectangle, inner sep=2 pt, minimum height=20pt, minimum width=30pt, line width=0.75pt]
\tikzstyle{empty 30pt}=[fill=white, draw=black, shape=circle, inner sep=0pt, minimum size=30pt, line width=0.75pt]
\tikzstyle{empty 20pt}=[fill=white, draw=black, shape=circle, inner sep=0pt, minimum size=20pt, line width=0.75pt]
\tikzstyle{filled 10pt}=[fill=black, draw=black, shape=circle, minimum size=5pt, text=white, inner sep=1pt]
\tikzstyle{tiny rectangle}=[fill=white, draw=black, shape=rectangle, minimum width=5pt, minimum height=5pt, inner sep=0pt]
\tikzstyle{dash circle}=[fill=white, draw=black, shape=circle, dashed, inner sep=0.5pt]
\tikzstyle{tall}=[fill=white, draw=black, shape=rectangle, minimum width=1cm, minimum height=1.25cm, inner sep=0pt, line width=0.5pt]
\tikzstyle{filled R}=[fill={rgb,255: red,191; green,0; blue,64}, draw=black, shape=circle, minimum size=5pt, inner sep=1pt]
\tikzstyle{filled B}=[fill={rgb,255: red,0; green,0; blue,208}, draw=black, shape=circle, minimum size=5pt, inner sep=1pt]
\tikzstyle{filled G}=[fill={rgb,255: red,0; green,186; blue,0}, draw=black, shape=circle, minimum size=5pt, inner sep=1pt]
\tikzstyle{filled W}=[fill=white, draw=black, shape=circle, minimum size=5pt, inner sep=1pt]
\tikzstyle{big none}=[fill=none, draw=none, minimum size=5pt, inner sep=1pt]
\tikzstyle{state G}=[fill={rgb,255: red,0; green,186; blue,0}, draw=black, shape=circle, line  width=0.75pt, minimum size=10pt]
\tikzstyle{nano}=[fill=black, draw=black, shape=circle, minimum size=3pt, inner sep=0pt]
\tikzstyle{big empty}=[fill=white, draw=black, shape=circle, inner sep=0pt, minimum size=30pt]
\tikzstyle{nano 2pt}=[fill=black, draw=black, shape=circle, minimum size=2pt, inner sep=0pt]
\tikzstyle{state U}=[fill={rgb,255: red,180; green,180; blue,180}, draw=black, shape=circle, line width=0.75pt, minimum size=10pt]

\tikzstyle{thick}=[-, line width=1pt]
\tikzstyle{thick arrow}=[->, line width=1pt]
\tikzstyle{packing}=[-, line width=1pt, draw={rgb,255: red,171; green,0; blue,60}]
\tikzstyle{dotted}=[-, line width=1pt, dashed]
\tikzstyle{dashed cyan thick}=[-, dashed, line width=1pt, draw={rgb,255: red,2; green,200; blue,200}]
\tikzstyle{white}=[-, draw=white]
\tikzstyle{del}=[-, double, line width=1pt]
\tikzstyle{thin}=[-, line width=0.5pt]
\tikzstyle{red dashed}=[-, line width=1pt, draw={rgb,255: red,191; green,0; blue,64}, dashed]
\tikzstyle{dark green}=[-, line width=1pt, draw={rgb,255: red,44; green,127; blue,35}]
\tikzstyle{thin green}=[-, line width=0.5pt, draw={rgb,255: red,27; green,125; blue,52}]
\tikzstyle{thin red}=[-, line width=0.5pt, draw={rgb,255: red,181; green,0; blue,3}]
\tikzstyle{thin dashed}=[-, line width=0.5pt, dashed]
\tikzstyle{green fill}=[-, line width=0.5pt, draw={rgb,255: red,27; green,125; blue,52}, fill={rgb,255: red,162; green,255; blue,155}]
\tikzstyle{blue fill}=[-, line width=0.5pt, draw={rgb,255: red,18; green,0; blue,157}, fill={rgb,255: red,116; green,123; blue,255}]
\tikzstyle{thin arrow}=[->]
\tikzstyle{fill red}=[-, fill={rgb,255: red,255; green,0; blue,4}, fill opacity=0.2, draw=none]
\tikzstyle{bold green}=[-, fill=none, draw={rgb,255: red,0; green,182; blue,0}, line width=2pt]
\tikzstyle{bold red}=[-, draw={rgb,255: red,225; green,0; blue,0}, line width=2pt]
\tikzstyle{bold cyan dashed}=[-, draw={rgb,255: red,2; green,200; blue,200}, line width=1.5pt, dashed]

\begin{document}
  
\maketitle

\begin{abstract}
  The complexity of problems involving global constraints is usually much more difficult to understand than the complexity of problems only involving local constraints. In the realm of graph problems, connectivity constraints are a natural form of global constraints. We study connectivity problems from a fine-grained parameterized perspective. In a breakthrough result, Cygan et al.~(TALG 2022) first obtained algorithms with single-exponential running time $\alpha^{\tw} n^{\Oh(1)}$ for connectivity problems parameterized by treewidth by introducing the cut-and-count-technique, which reduces many connectivity problems to locally checkable counting problems. Furthermore, the obtained bases $\alpha$ were shown to be optimal under the Strong Exponential-Time Hypothesis (SETH).
  
  However, since only sparse graphs may admit small treewidth, we lack knowledge of the fine-grained complexity of connectivity problems with respect to dense structure. The most popular graph parameter to measure dense structure is arguably clique-width, which intuitively measures how easily a graph can be constructed by repeatedly adding bicliques. Bergougnoux and Kant\'e (TCS 2019) have shown, using the rank-based approach, that also parameterized by clique-width many connectivity problems admit single-exponential algorithms. Unfortunately, the obtained running times are far from optimal under SETH.
  
  We show how to obtain optimal running times parameterized by clique-width for two benchmark connectivity problems, namely \textsc{Connected Vertex Cover} and \textsc{Connected Dominating Set}. These are the first tight results for connectivity problems with respect to clique-width and these results are obtained by developing new algorithms based on the cut-and-count-technique and novel lower bound constructions. Precisely, we show that there exist one-sided error Monte-Carlo algorithms that given a $k$-clique-expression solve
  \begin{itemize}
    \item \textsc{Connected Vertex Cover} in time $6^k n^{\Oh(1)}$, and
    \item \textsc{Connected Dominating Set} in time $5^k n^{\Oh(1)}$.
  \end{itemize}
  Both results are shown to be tight under SETH.
\end{abstract}
\newpage

\tableofcontents

\section{Introduction}
\label{sec:intro}

One way to cope with the NP-hardness of a problem is the theory of parameterized complexity, where we seek to solve structured instances faster than worst-case instances; an additional parameter quantifies how structured an instance is. Ideally, we obtain fixed-parameter tractable algorithms with running time $\Oh^*(f(k))$\footnote{The $\Oh^*$-notation hides polynomial factors in the input size.}, where $k$ is the parameter and $f$ some computable function. Having established the existence of such an algorithm, the next natural step is to take a fine-grained perspective and to determine the smallest possible function $f$ for this problem-parameter-combination, which allows us to quantify the precise impact of the considered structure on problem complexity.

We study connectivity problems from a fine-grained parameterized perspective. This line of work starts with the breakthrough result of Cygan et al.~\cite{CyganNPPRW22} which for the first time obtained algorithms with running time $\Oh^*(\alpha^{\tw})$, for some constant \emph{base} $\alpha > 1$, for connectivity problems parameterized by \emph{treewidth} ($\tw$) by introducing the cut-and-count-technique, which reduces connectivity problems to locally checkable counting problems. In addition, the obtained bases $\alpha$ were proven to be optimal assuming the Strong Exponential-Time Hypothesis (SETH) \cite{CyganNPPRW11}.

As only sparse graphs may have small treewidth, we lack knowledge of the precise complexity of connectivity problems with respect to dense structure. In the regime of dense graphs, \emph{clique-width} is one of the most popular parameters. Bergougnoux~\cite{Bergougnoux19} has applied cut-and-count to several width-parameters based on structured neighborhoods with clique-width among these. Moreover, Bergougnoux and Kant\'e~\cite{BergougnouxK19a}, building upon the rank-based approach of Bodlaender et al.~\cite{BodlaenderCKN15}, obtain single-exponential running times $\Oh^*(\alpha^{\cw})$ for a large class of connectivity problems parameterized by clique-width. As both articles are aimed at obtaining a breadth of single-exponential algorithms for a large class of problems, the \textsc{Connected (Co-)$(\sigma, \rho)$-Dominating Set} problems, the obtained bases for particular problems are far from being optimal. For example, the former article implies an $\Oh^*(128^{\cw})$-time algorithm for \CDS and the latter article yields an $\Oh^*((27 \cdot 2^{\omega + 1})^{\cw})$-time algorithm for \CVC and an  $\Oh^*((8 \cdot 2^{\omega + 1})^{\cw})$-time algorithm for \CDS, where $\omega$ is the matrix multiplication exponent, see e.g.\ Alman and Vassilevska W.~\cite{AlmanW21}. Even if $\omega = 2$, this only yields the very large bases $216$ and $64$ respectively. 

We show that the running times for \CVC and \CDS parameterized by clique-width can be considerably optimized by providing novel algorithms. These faster algorithms again rely on the cut-and-count-technique and are fine-tuned by precisely analyzing which cut-and-count states are necessary to consider. Moreover, we use further techniques such as fast subset convolution, inclusion-exclusion states, and distinguishing between live and dead labels to obtain the improved running times.

\begin{thm}\label{thm:cw_algos}
 There are one-sided error Monte-Carlo algorithms that, given a $\nlabels$-expression\footnote{A $\nlabels$-expression witnesses\ that the clique-width of $G$ is at most $\nlabels$.} for a graph $G$, can solve
	\begin{itemize}
	 \item \CVC in time $\Oh^*(6^\nlabels)$,
	 \item \CDS in time $\Oh^*(5^\nlabels)$.
	\end{itemize} 
\end{thm}

We show that these algorithms are essentially the correct ones for these problem-parameter-combinations by proving that the obtained running times are optimal under SETH. To prove these lower bounds, we follow the by now standard construction principle of Lokshtanov et al.~\cite{LokshtanovMS18} for lower bounds relative to width-parameters. To apply this principle for clique-width, we closely investigate the problem behavior across \emph{joins}, i.e., the edge-structures via which clique-width is defined, and the results of this investigation strongly guide us in designing appropriate gadgets. Precisely, we obtain the following tight lower bounds:

\begin{thm}\label{thm:cw_lbs}
 Assuming SETH, the following statements hold for all $\eps > 0$:
 \begin{itemize}
	 \item \CVC cannot be solved in time $\Oh^*((6 - \eps)^{\cw})$.
	 \item \CDS cannot be solved in time $\Oh^*((5 - \eps)^{\cw})$.
 \end{itemize}
\end{thm}

This work is part of a larger research program to determine the optimal running times for various connectivity problems relative to several width-parameters ranging from restrictive to more and more general ones, hence yielding a \emph{fine-grained} understanding of the \emph{price of generality} in this setting. We summarize the known results in \cref{table:conn_time_overview}. The cut-and-count-technique by Cygan et al.~\cite{CyganNPPRW11arxiv, CyganNPPRW22} together with their lower bounds settle the complexity relative to treewidth (and pathwidth) for many connectivity problems. Bojikian et al.~\cite{BojikianCHK23} consider the even more restrictive \emph{cutwidth} and combine cut-and-count with the rank-based approach to improve upon the treewidth-algorithms or provide more economical lower bound constructions of low cutwidth when no improved algorithm exists. This work and a companion paper~\cite{HegerfeldK23mtw} present the first fine-grained parameterized results for connectivity problems parameterized by a dense width-parameter. The companion paper~\cite{HegerfeldK23mtw} considers the parameter \emph{modular-treewidth} which lifts treewidth into the dense regime by combining tree decompositions with modular decompositions and thus serves as a natural intermediate step between treewidth and clique-width. The algorithmic results on modular-treewidth are obtained by either reducing directly to the treewidth-case if the base in the running time remains the same or by applying the cut-and-count-technique and the modular structure to essentially reduce to a more involved problem parameterized by treewidth; in the latter case, new lower bound constructions are provided that follow similar high-level principles as here, but have to adhere to different design restrictions. Cygan et al.~\cite{CyganNPPRW22} have observed that imposing a connectivity constraint increases the base by at most 1 in the sparse setting, e.g., \VC has optimal base 2, see Lokshtanov et al.~\cite{LokshtanovMS18}, and \CVC has optimal base 3 parameterized by treewidth. In the dense setting, the impact of the connectivity constraint can vary more, e.g., parameterized by clique-width the optimal bases of \VC and \DS are 2 and 4~\cite{IwataY15, KatsikarelisLP19}, respectively, which increase to 6 and 5, respectively, when adding the connectivity constraint.
	\newcommand{\sexp}[1]{$\Oh^*(#1^k)$}%
	\begin{table}%
		\centering
		\begin{tabular}{l|cccc}%
			Parameters & cutwidth & treewidth & modular-tw & clique-width\\%
			\hline%
			\CVC & \sexp{2} & \sexp{3} & \sexp{5} & \sexp{6} \\%
			\CDS & \sexp{3} & \sexp{4}	& \sexp{4} & \sexp{5} \\%
			\ST  & \sexp{3} & \sexp{3} & \sexp{3} & ? \\%
			\FVS & \sexp{2} & \sexp{3} & \sexp{5} & ? \\%
			\hline%
			References & \cite{BojikianCHK23} & \cite{CyganNPPRW11arxiv, CyganNPPRW22} & \cite{HegerfeldK23mtw} & here%
		\end{tabular}%
		\caption{Optimal running-times of several connectivity problems with respect to various width-parameters listed in increasing generality. The results in the last column are obtained in this paper. Between modular-treewidth and clique-width, we only have the relationship $\cw(G) \leq 3 \cdot 2^{\modtw(G)-1}$, but the same results are also tight for the more restrictive modular-pathwidth, where we have $\cw(G) \leq \modpw(G) + 2$ by Hegerfeld and Kratsch~\cite{HegerfeldK23mtw}. The ``?'' marks problem-parameter combinations, where an algorithm with single-exponential running time is known by Bergougnoux and Kant\'e~\cite{BergougnouxK19a}, but a gap between the lower bound and algorithm remains.}\label{table:conn_time_overview}%
		\vspace*{-0.8cm}
	\end{table}%
\subparagraph*{Further Related Work.}
Beyond these tight results, the cut-and-count-technique has also been applied to the parameters branchwidth~\cite{PinoBR16} and treedepth~\cite{HegerfeldK20, NederlofPSW20}. Due to its reliance on the isolation lemma, the cut-and-count-technique yields randomized algorithms. The rank-based approach of Bodlaender et al.~\cite{BodlaenderCKN15} and the matroid-based techniques by Fomin et al.~\cite{FominLPS16, FominLPS17} deal with this shortcoming at the cost of a higher running time and the rank-based approach can also help for problems without connectivity constraints. By combining the rank-based approach with other techniques to avoid having to resort to Gaussian elimination, optimal running times can be obtained in some cases such as for \HC parameterized by pathwidth~\cite{CurticapeanLN18, CyganKN18}, or coloring problems parameterized by cutwidth~\cite{GroenlandMNS22, JansenN19}. There are further applications of the rank-based approach to connectivity problems relative to dense width-parameters, such as rankwidth~\cite{BergougnouxK21} and mim-width~\cite{BergougnouxDJ23}. We also refer to the survey of Nederlof on rank-based methods~\cite{Nederlof20}.

Moving away from connectivity problems, we survey some more of the literature obtaining tight fine-grained parameterized algorithms for dense parameters. Iwata and Yoshida show that for any $\eps > 0$ \VC can be solved in time $\Oh^*((2-\eps)^{\tw})$ if and only if \VC can be solved in time $\Oh^*((2-\eps)^{\cw})$~\cite{IwataY15}; as the bases differ for treewidth and clique-width in our case, it seems difficult to transfer their techniques to our setting. Lampis~\cite{Lampis20} obtains the tight running time of $\Oh^*((2^q - 2)^{\cw})$ for $q$-Coloring and a tight result for $q$-Coloring parameterized by a more restrictive variant of modular-treewidth. Generalizing to homomorphism problems, Ganian et al.~\cite{GanianHKOS22} obtain tight results for parameterization by clique-width, where the obtained base depends on a special measure of the target graph. Katsikarelis et al.~\cite{KatsikarelisLP19} obtain tight results for \textsc{$(k,r)$-Center} parameterized by cliquewidth and, in particular, the tight running time $\Oh^*(4^{\cw})$ for \textsc{Dominating Set}. Jacob et al.~\cite{JacobBDP21} and Hegerfeld and Kratsch~\cite{HegerfeldK22} show that the running time $\Oh^*(4^{\cw})$ is tight for \OCT, where the latter article also considers a generalization to more colors and contains tight results for parameters that are not width-parameters.

\section{Technical Overview}

In this section, we outline the techniques used to prove \cref{thm:cw_algos} and \cref{thm:cw_lbs}.

\subsection{Algorithmic Techniques}

\subparagraph*{Cut and Count.} 
The cut-and-count-technique by Cygan et al.~\cite{CyganNPPRW22} allows us to reduce the connectivity constraint to a locally checkable counting problem. A \emph{consistent cut} of a graph $G = (V, E)$ is an ordered partition of the vertices $V$ into two parts $(V_L, V_R)$ such that no edge in $E$ crosses between the two cut sides $V_L$ and $V_R$. The key property of consistent cuts is that $G$ admits precisely $2^{\cc(G)}$ distinct consistent cuts, where $\cc(G)$ is the number of connected components of $G$. By fixing a vertex $\fixedvertex$ and only considering consistent cuts $(V_L, V_R)$ with $\fixedvertex \in V_L$, this number reduces to $2^{\cc(G)-1}$, so that $G$ admits an odd number of such consistent cuts if and only if $G$ is connected. Hence, if we count pairs $(X, (X_L, X_R))$, where $\fixedvertex \in X \subseteq V$ is a partial solution and $(X_L, X_R)$ a consistent cut of $G[X]$ with $\fixedvertex \in X_L$, modulo two, then only connected solutions survive. When multiple connected solutions exist, this can lead to unwanted cancellations modulo two, but this issue can be avoided at the cost of randomization by using the isolation lemma~\cite{MulmuleyVV87}.

\subparagraph*{Lifting Vertex States to Label States.}
For dynamic programming along clique-expressions, we have to characterize the relevant interactions of a partial solution with the \emph{labels} which govern which \emph{joins} can be constructed by the expression. In the considered problems, a single vertex $v$ can take a constant number of different states with respect to a partial solution which we capture with a problem-dependent set $\bstates$; e.g., for \CVC, we have $\bstates = \{\zero, \one_L, \one_R\}$, representing $v \notin X$ (state $\zero$), $v \in X_L$ (state $\one_L$), and $v \in X_R$ (state $\one_R$), respectively. A clique-expression repeatedly adds \emph{joins} between pairs of vertex sets, say $A$ and $B$, i.e., all possible edges between $A$ and $B$ are added, and the algorithm must check whether a partial solution remains feasible after adding a join and possibly update some states. A priori, each choice of vertex states in a label could yield different behaviors for partial solutions. However, the crucial observation for the considered problems is that the precise multiplicity of a vertex state in $A$ or in $B$ is irrelevant for a join, rather it suffices to distinguish which vertex states appear on each side and which do not. Therefore, the relevant label states are captured by the subsets of $\bstates$. The next two techniques will allow us to reduce the number of considered states further.

\subparagraph*{Nice Clique-Expressions.}
For both algorithms, we refine and augment standard clique-expressions to distinguish between \emph{live} and \emph{dead} labels. When performing dynamic programming along a clique-expression, we consider the induced subgraphs defined by subexpressions of the given clique-expression. At a subexpression, we say that a label $\ell$ is \emph{live} if in the remaining expression the vertices with label $\ell$ receive further edges that are not present in the current subexpression, otherwise we say that $\ell$ is \emph{dead}. First, we observe that we do not need to track the states of a partial solution at the dead labels, as they only have trivial interactions with the other states in the remaining expression. Hence, we only need to consider the states that can be attained at live labels which allows us to reduce the number of considered states for \CVC. To simplify the description of the algorithms and avoid handling of edge cases, we transform the clique-expressions so that no degenerate cases occur and add a dead-operation $\deadop{\ell}$ which handles label $\ell$ turning from live to dead. The dead-operation is similar to \emph{forget vertex nodes} in \emph{nice tree decompositions}~\cite{Kloks94}. Distinguishing live and dead labels has been used before~\cite{GanianHKOS22, KatsikarelisLP19, Lampis20} to obtain improved running times, but handling the label types explicitly via an additional operation is new to the best of the authors' knowledge.

\subparagraph*{Inclusion-Exclusion States.}
For \CDS, we transform to a different set of vertex states, called \emph{inclusion-exclusion states}, which have proven helpful for domination problems before~\cite{HegerfeldK20, NederlofRD14, PilipczukW18, RooijBR09}. With these states we do not track whether a vertex is \emph{undominated} or \emph{dominated} by a partial solution, but rather \emph{allow} a vertex to be dominated or \emph{forbid} it. A solution to the original problem can usually be recovered by an inclusion-exclusion argument, however when lifting to label states this argument does not directly transfer. We show that the argument can be adapted for the label states when working modulo two, whereas for vertex states the argument is known to also work for non-modular counting. The advantage of the inclusion-exclusion states is that at join-operations we do not have to update vertex states from undominated to dominated, thus simplifying the algorithm and also allowing us to collapse several label states into a single one. The dead-operations of nice clique-expressions serve as suitable timepoints in the algorithm to apply the adapted inclusion-exclusion argument.

\subparagraph*{Fast Convolutions.}
To quickly compute the dynamic programming recurrences, we utilize algorithms for \emph{fast subset convolution}. In \cref{sec:fast_convolutions}, we tailor the techniques developed by Björklund et al.~\cite{BjorklundHKK10} on trimmed subset convolutions to obtain a fast algorithm for the union-recurrence appearing in the \CVC algorithm. For \CDS, the lattice-based results of Björklund et al.~\cite{BjorklundHKKNP16} provide the necessary means to compute the union-recurrence quickly. In both cases, we obtain fast convolution algorithms applicable in more general settings, so these results could be of independent interest.

\subsection{Lower Bound Approach}

\subparagraph*{Grid-like Construction.}
Both lower bounds are based on the high-level construction principle already present in the SETH-lower bounds of Lokshtanov et al.~\cite{LokshtanovMS18} for parameterization by path/treewidth. The resulting graphs can be interpreted as a \emph{grid/matrix of blocks}, where each block spans several rows and columns. Each row is a long \emph{path-like gadget} simulating a constant number of variables of the \SAT instance and contributes one unit of clique-width. The more variables we are able to simulate per row, the higher the running time we can rule out. Every column represents a clause and consists of gadgets \emph{decoding} the states of the path gadgets and verifying whether the resulting assignment satisfies the corresponding clause.

\subparagraph*{Path Gadgets and State Transitions.} Our main technical contribution is the design of the \emph{path gadgets} that lie at the intersection of every row and column, whereas the design of the \emph{decoding gadgets} can be adapted from known constructions by Cygan et al.~\cite{CyganNPPRW11arxiv}. Since every row should contribute one unit of clique-width, adjacent path gadgets in a row must be connected by a join. Our goal is to design a path gadget, under these restrictions, admitting as many distinct states as possible. An important issue is how the state of the path gadgets may transition along each row, as the reduction only works when the state transitions follow some \emph{transition order}.

\subparagraph*{Determining the Transition Order.} Since the number of possible label states can be large, it is not immediate how to pick an appropriate transition order. For lower bounds parameterized by path/treewidth, this is much less of an issue, since the number of possible states is much smaller. Hence, we systematically analyze the possible state transitions across a join, obtaining a \emph{transition/compatibility} matrix showing which pairs of states can lead to a globally feasible solution and which cannot. After possibly reordering the rows and columns of the compatibility matrix, a possible transition order must induce a \emph{triangular submatrix}. From a largest possible triangular submatrix of the compatibility matrix, we can then deduce an appropriate transition order which guides the design of the path gadget.

\subparagraph*{Anatomy of a Path Gadget.} Our path gadgets consist of three parts: a \emph{central clique} that communicates with the decoding gadgets, and two \emph{boundary} parts, i.e., the \emph{left} and \emph{right} part that connect to the previous and following join, respectively. In the central clique, each solution will avoid exactly one vertex representing the state of the path gadget. To implement the transition order, the left and right part have to communicate appropriate states to the two adjacent path gadgets. By taking the triangular submatrix and \emph{pairing} states along the main diagonal, we see which states must be communicated in each case. The central idea behind designing the left and right part is to isolate the constituent state properties of the boundary vertices, such as, whether they are contained in the partial solution or whether they are dominated. This allows for simple communication with the central clique and expedites the implementation of the transition order.

\section{Preliminaries}

For two integers $a,b$ we write $a \equiv_c b$ to indicate equality modulo $c \in \NN$. We use Iverson's bracket notation: for a boolean predicate $p$, we have that $[p]$ is $1$ if $p$ is true and $0$ otherwise. For a function $f$ we denote by $f[v \mapsto \alpha]$ the function $(f \setminus \{(v, f(v))\}) \cup \{(v, \alpha)\}$, viewing $f$ as a set; we also write $f[v \mapsto \alpha, w \mapsto \beta]$ instead of $(f[v \mapsto \alpha])[w \mapsto \beta]$. By $\ZZ_2$ we denote the field of two elements. 
For $n_1, n_2 \in \ZZ$, we write $[n_1, n_2] = \{x \in \ZZ \sep n_1 \leq x \leq n_2\}$ and $[n_2] = [1, n_2]$. For a function $f\colon V \rightarrow \ZZ$ and a subset $W \subseteq V$, we write $f(W) = \sum_{v \in W} f(v)$. Note that for functions $g\colon A \rightarrow B$, where $B \not\subseteq \ZZ$, and a subset $A' \subseteq B$, we still denote the \emph{image of $A'$ under $g$} by $g(A') = \{g(v) \sep v \in A'\}$. If $f \colon A \rightarrow B$ is a function and $A' \subseteq A$, then $f\big|_{A'}$ denotes the \emph{restriction} of $f$ to $A'$ and for a subset $B' \subseteq B$, we denote the \emph{preimage of $B'$ under $f$} by $f^{-1}(B') = \{a \in A \sep f(a) \in B'\}$. An ordered tuple of sets $(A_1, \ldots, A_\ell)$ is an \emph{ordered subpartition} if $A_i \cap A_j = \emptyset$ for all $i \neq j \in [\ell]$.

\paragraph*{Graph Notation}
We use common graph-theoretic notation and assume that the reader knows the essentials of parameterized complexity. Let $G = (V, E)$ be an undirected graph. For a vertex set $X \subseteq V$, we denote by $G[X]$ the subgraph of $G$ that is induced by $X$. The \emph{open neighborhood} of a vertex $v$ is given by $N(v) = \{u \in V \sep \{u,v\} \in E\}$, whereas the \emph{closed neighborhood} is given by $N[v] = N(v) \cup \{v\}$. For sets $X \subseteq V$ we define $N[X] = \bigcup_{v \in X} N[v]$ and $N(X) = N[X] \setminus X$. For two disjoint vertex subsets $A, B \subseteq V$, we define $E_G(A,B) = \{\{a,b\} \in E(G) \sep a \in A, b \in B\}$ and adding a \emph{join} between $A$ and $B$ means adding an edge between every vertex in $A$ and every vertex in $B$. For a vertex set $X \subseteq V$, we define $\delta_G(X) = E_G(X, V \setminus X)$ and we write $\delta_G(v) = \delta_G({v})$ for single vertices $v$. We denote the \emph{number of connected components} of $G$ by $\cc(G)$. A \emph{cut} of $G$ is a partition $V = V_L \cup V_R$, $V_L \cap V_R = \emptyset$, of its vertices into two parts.

\paragraph*{Clique-Expressions and Clique-Width}

A \emph{labeled graph} is a graph $G = (V,E)$ together with a \emph{label function} $\lfct \colon V \rightarrow \NN = \{1, 2, 3, \ldots\}$; we usually omit mentioning $\lfct$ explicitly. A labeled graph is \emph{$k$-labeled} if $\lfct(v) \leq k$ for all $v \in V$. We consider the following four operations on labeled graphs: 
\begin{itemize}
  \item the \emph{introduce}-operation $\intro{\ell}(v)$ which constructs a single-vertex graph whose unique vertex $v$ has label $\ell$,
  \item the \emph{union}-operation $G_1 \union G_2$ which constructs the disjoint union of two labeled graphs $G_1$ and $G_2$,
  \item the \emph{relabel}-operation $\relab{i}{j}(G)$ changes the label of all vertices in $G$ with label $i$ to label $j$, 
  \item the \emph{join}-operation $\join{i}{j}(G)$, $i \neq j$, which adds an edge between every vertex in $G$ with label $i$ and every vertex in $G$ with label $j$. 
\end{itemize}
A valid expression that only consists of introduce-, union-, relabel-, and join-operations is called a \emph{clique-expression}. The graph constructed by a clique-expression $\cexpr$ is denoted $G_\cexpr$ and the constructed label function is denoted $\lfct_\cexpr \colon V(G_\cexpr) \rightarrow \NN$.

We associate to a clique-expression $\cexpr$ the syntax tree $\tree_\cexpr$ in the natural way and to each node $t \in V(\tree_\cexpr)$ the corresponding operation. For any node $t \in V(\tree_\cexpr)$ the subtree rooted at $t$ induces a \emph{subexpression} $\cexpr_t$. When a clique-expression $\cexpr$ is fixed, we define $G_t = G_{\cexpr_t}$, $V_t = V(G_t)$, $E_t = E(G_t)$, and $\lfct_t = \lfct_{\cexpr_t}$ for any $v \in V(\tree_\cexpr)$. Furthermore, we write $V_t^\ell = \lfct_t^{-1}(\ell)$ for the set of all vertices with label $\ell$ at node $t$ and we write $\labels_t = \{\ell \in \NN \sep V_t^\ell \neq \emptyset\}$ for the set of \emph{nonempty labels at node $t$}.

We say that a clique-expression $\cexpr$ is a \emph{$k$-clique-expression} or just \emph{$k$-expression} if $G_t$ is $k$-labeled for all $t \in V(\tree_\cexpr)$. The \emph{clique-width} of a graph $G$, denoted by $\cw(G)$, is the minimum $k$ such that there exists a $k$-expression $\cexpr$ with $G = G_\cexpr$. A clique-expression $\cexpr$ is \emph{linear} if in every union-operation the second graph consists only of a single vertex. Accordingly, we also define the \emph{linear-clique-width} of a graph $G$, denoted $\lcw(G)$, by only considering linear clique-expressions.

\paragraph*{Strong Exponential-Time Hypothesis}

The \emph{Strong Exponential-Time Hypothesis} (\SETH) \cite{CalabroIP09, ImpagliazzoPZ01} concerns the complexity of $\clss$-\SAT, i.e., \SAT where every clause contains at most $\clss$ literals. We define $c_\clss = \inf \{\delta \sep \clss\text{-\SAT can be solved in time } \Oh(2^{\delta \nvars}) \}$ for all $\clss \geq 3$. The weaker \emph{Exponential-Time Hypothesis} (ETH) of Impagliazzo and Paturi~\cite{ImpagliazzoP01} posits that $c_3 > 0$, whereas the Strong Exponential-Time Hypothesis states that $\lim_{\clss \rightarrow \infty} c_\clss = 1$. 
Or equivalently, for every $\delta < 1$, there is some $\clss$ such that $\clss$-\SAT cannot be solved in time $\Oh(2^{\delta \nvars})$.
For our lower bounds, the following weaker variant of \SETH, also called \CNFSETH, is sufficient.

\begin{cnj}[\CNFSETH]
  \label{conj:cnfseth}
  For every $\eps > 0$, there is no algorithm solving \SAT with $n$ variables and $m$ clauses in time $\Oh(\poly(m)(2-\eps)^n)$.
\end{cnj}

\paragraph*{Isolation Lemma}

\begin{dfn}
  A function $\wfct \colon U \rightarrow \ZZ$ \emph{isolates} a set family $\family \subseteq \powerset{U}$ if there is a unique $S' \in \family$ with $\wfct(S') = \min_{S \in \family} \wfct(S)$, where for subsets $X$ of $U$ we define $\wfct(X) = \sum_{u \in X} \wfct(u)$.
\end{dfn}

\begin{lem}[Isolation Lemma, \cite{MulmuleyVV87}]
 \label{thm:isolation}
 Let $\family \subseteq \powerset{U}$ be a nonempty set family over a universe $U$. Let $\wmax \in \NN$ and for each $u \in U$ choose a weight $\wfct(u) \in [\wmax]$ uniformly and independently at random. Then
  $\PP[\wfct \text{ isolates } \family] \geq 1 - |U|/\wmax$.
\end{lem}

When counting objects modulo 2 the Isolation Lemma allows us to avoid unwanted cancellations by ensuring with high probability that there is a unique solution. In our applications, we will choose $\wmax$ so that we obtain an error probability of less than $1/2$.

\subsection{Cut and Count}
\label{sec:cutandcount}

Let $G = (V, E)$ denote a connected graph. To solve a vertex selection problem on $G$ involving a connectivity constraint, we make the following general definitions. The family of \emph{connected solutions} to our problem is denoted by $\sols \subseteq \powerset{V}$. We have to determine whether $\sols$ is empty or not. The \emph{cut-and-count-technique} by Cygan et al.~\cite{CyganNPPRW22} accomplishes this in two parts:
\begin{itemize}
  \item \textbf{The Cut part:} By relaxing the connectivity constraint, we obtain a set $\sols \subseteq \rsols \subseteq \powerset{V}$ of possibly connected solutions, called \emph{candidates}. The set $\dpsols{}{}$ contains pairs $(X, (X_L,X_R))$ consisting of a candidate $X \in \rsols$ and a \emph{consistent cut} $(X_L,X_R)$  of $G[X]$, cf.~\cref{dfn:cons_cut}.
  \item \textbf{The Count part:} We compute $|\dpsols{}{}|$ modulo 2 using a subprocedure. The consistent cuts are defined so that disconnected candidate solutions $X \in \rsols \setminus \sols$ cancel, because they are consistent with an even number of cuts. Hence, only connected candidates $X \in \sols$ remain and we have $\sols \neq \emptyset$ if the parity of $|\dpsols{}{}|$ is odd.
\end{itemize}
If $|\sols|$ is divisible by 2, then this approach fails, since the connected candidates also cancel when counting modulo 2. The Isolation Lemma (\cref{thm:isolation}) allows us to avoid this issue at the cost of randomization. By sampling a weight function $\wfct \colon V \rightarrow [\wmax]$, we can instead count pairs with a fixed weight and it is likely that there is a weight with a unique solution if a solution exists at all.
\begin{dfn}[\cite{CyganNPPRW11}]
  \label{dfn:cons_cut}
  A cut $(V_L, V_R)$ of an undirected graph $G = (V, E)$ is \emph{consistent} if $u \in V_L$ and $v \in V_R$ implies $\{u,v\} \notin E$, i.e., $E_G(V_L, V_R) = \emptyset$. A \emph{consistently cut subgraph} of $G$ is a pair $(X, (X_L, X_R))$ such that $X \subseteq V$ and $(X_L, X_R)$ is a consistent cut of $G[X]$. We denote the set of consistently cut subgraphs of $G$ by $\cuts(G)$.
\end{dfn}

To ensure that connected solutions are not compatible with an even number of consistent cuts, we will usually force a single vertex to the left side of the consistent cut. This results in the following fundamental property of consistent cuts.
\begin{lem}[\cite{CyganNPPRW11}]
  \label{thm:cons_cut}
  Let $X$ be a subset of vertices such that $\fixedvertex \in X \subseteq V$. The number of consistently cut subgraphs $(X, (X_L, X_R))$ such that $\fixedvertex \in X_L$ is equal to $2^{\cc(G[X]) - 1}$.
\end{lem}

With \cref{thm:cons_cut} we can distinguish disconnected candidates from connected candidates by determining the parity of the number of consistent cuts for the respective candidate. 

\begin{cor}
  \label{thm:cons_cut_family}
  Let $\rsols \subseteq \powerset{V}$ be a family of vertex sets so that every $X \in \rsols$ contains $\fixedvertex$. If the set $\dpsols{}{} = \{(X,(X_L, X_R)) \in \cuts(G) \sep X \in \rsols, \fixedvertex \in X_L\}$ has odd cardinality, then there exists an $X \in \rsols$ such that $G[X]$ is connected.
\end{cor}

\begin{proof}
  We will prove the contrapositive, so suppose that $G[X]$ is not connected for every $X \in \rsols$. By \cref{thm:cons_cut}, we have that $|\dpsols{}{}| = \sum_{X \in \rsols} 2^{\cc(G[X]) - 1} \equiv_2 0$, since $\cc(G[X]) \geq 2$ for all $X \in \rsols$.
\end{proof}

\subsection{Nice Clique-Expressions}
Let $\cexpr$ be a $\nlabels$-expression for $G = (V, E)$; the associated syntax tree is $\tree_\cexpr$. We say that a clique-expression $\cexpr$ is \emph{irredundant} if for any join-operation $\join{i}{j}(G_{t'}) = t \in V(\tree_\cexpr)$, it holds that $E_{t'}(V_{t'}^i, V_{t'}^j) = \emptyset$, i.e., no edge added by the join existed before.

\begin{thm}[\cite{CourcelleO00}]\label{thm:irredundant_expression}
  Any $\nlabels$-expression $\cexpr$ can be transformed into an equivalent, i.e., $G_{\cexpr'} = G_\cexpr$ irredundant $k$-expression $\cexpr'$ in polynomial time.
\end{thm}

\begin{lem}[\cite{BergougnouxKK20}]\label{thm:irredundant_lemma}
  If $\cexpr$ is an irredundant $\nlabels$-expression for the graph $G = (V,E)$ and $t \in V(\tree_\cexpr)$, then for all labels $\ell \in \labels_t$ and vertices $u, v \in V_t^\ell$ we have that $N_G(u) \setminus V_t = N_G(v) \setminus V_t$. Furthermore, if $u \in V_t^i$, $v \in V_t^j$ with $i \neq j \in \labels_t$ and $\{u,v\} \in E \setminus E_t$, then $E_G(V_t^i, V_t^j) \subseteq E \setminus E_t$.
\end{lem}

Henceforth, we will assume that the given $\nlabels$-expression $\cexpr$ is irredundant. Irredundancy still allows several edge cases regarding empty labels to occur, which would require special handling in the dynamic programming algorithms. To avoid this extra effort in the algorithms, we show how to transform any clique-expression such that these edge cases do not occur.

\begin{dfn}
  We say that a clique-expression $\cexpr$ of a graph $G = G_\cexpr$ is \emph{nice} if $\cexpr$ satisfies the following properties:
  \begin{itemize}
   \item $\cexpr$ is irredundant,
   \item for every join-node $\join{i}{j}(G_{t'}) = t \in V(\tree_\cexpr)$, where $t'$ is the child of $t$, we have that $G_t \neq G_{t'}$, i.e., $t$ adds at least one edge and $V_{t'}^i \neq \emptyset$ and $V_{t'}^j \neq \emptyset$,
   \item for every relabel-node $\relab{i}{j}(G_{t'}) = t \in V(\tree_\cexpr)$, where $t'$ is the child of $t$, we have that $V_{t'}^i \neq \emptyset$ and $V_{t'}^j \neq \emptyset$.
  \end{itemize}
\end{dfn}

\begin{lem}\label{thm:nice_expression}
  Any $\nlabels$-expression $\cexpr$ can be transformed into an equivalent, i.e., $G_{\cexpr'} = G_\cexpr$, nice $\nlabels$-expression $\cexpr'$ in polynomial time.
\end{lem}

\begin{proof}
	First running the algorithm of \cref{thm:irredundant_expression}, we can assume that $\cexpr$ is already irredundant. If a join node $t = \join{i}{j}(G_{t'})$ does not add any edges, then we must have $V_{t'}^i = \emptyset$ or $V_{t'}^j = \emptyset$ by irredundancy. Clearly, we can simply remove such join-nodes from the expression. The next step is to observe that any relabel node of the form $\relab{i}{j}(G_{t'}) = t \in V(\tree_\cexpr)$, where $t'$ is the child of $t$ and $V_{t'}^i = \emptyset$, can be removed from $\cexpr$ without changing the resulting graph, as no label is changed by such a node.
	
	Now, suppose that $\cexpr$ contains $p > 0$ relabel nodes of the form $\relab{i}{j}(G_{t'}) = t \in V(\tree_\cexpr)$, where $t'$ is the child of $t$ and $V_{t'}^i \neq \emptyset$ and $V_{t'}^j = \emptyset$; we call such relabels \emph{unnecessary}. We pick one such occurrence and show how to remove it. If $t$ is the root of $\cexpr$, then we can simply remove $t$ without changing the resulting graph. If $t$ is not the root of $\cexpr$, then let $s$ be the parent node of $t$. We swap the role of label $i$ and $j$ in every proper descendant of $t$ in $\tree_\cexpr$ and remove the node $t$ letting $t'$ be a child of $s$ in place of $t$. The correctness of the transformation can be proved by straightforward bottom-up induction along $\tree_\cexpr$ and this transformation does not create any new unnecessary relabels. 
  
  By repeating this transformation for every unnecessary relabel, we obtain an equivalent nice $\nlabels$-expression $\cexpr'$ in polynomial time.
\end{proof}

When designing tight algorithms for problems parameterized by clique-width, one often observes that there are states that cannot be attained by a label unless the vertices with this label have already received all their incident edges by the current subexpression. But for such labels the dynamic programming algorithm does not need to store the state as there will be no interesting interaction with other labels in the remaining expression. So, to improve the running time, we only store the state for \emph{live} labels, i.e., labels that contain vertices that are still missing some incident edges; labels that are not live are called \emph{dead}. This idea has been used several times before for dynamic programming on clique-expressions~\cite{GanianHKOS22, KatsikarelisLP19, Lampis20}.

While one could precompute for a given clique-expression $\cexpr$ which labels at a node $t \in V(\tree_\cexpr)$ are live, we choose to explicitly mark when a label is no longer live in the syntax tree. To do so, we augment the syntax tree with \emph{dead nodes} after join nodes that change at least one label from live to dead. The function of these dead nodes is comparable to that of forget vertex nodes in a nice tree decomposition~\cite{Kloks94}. Especially in the algorithm for \CDS this explicitness helps, because it allows us to cleanly separate two computations, namely the standard computation for join nodes and an extra computation that has to be performed when a label turns dead. We now proceed with the formal definitions.

For the remainder of this section, we assume that $G$ is a connected graph with at least two vertices.

\begin{dfn}
  Given a clique-expression $\cexpr$ for $G=(V,E)$ and a node $t \in V(\tree_\cexpr)$, the set of \emph{dead vertices at $t$} is defined by $\dead_t = \{v \in V_t \sep \delta_{G}(v) \subseteq E_t\}$. A vertex $v \in V_t \setminus \dead_t$ is called \emph{live at $t$}.
\end{dfn}

\begin{lem}\label{thm:dead_label_irredundant}
  Given an irredundant $\nlabels$-expression $\cexpr$ for $G=(V,E)$, a node $t \in V(\tree_\cexpr)$, and a nonempty label $\ell \in \labels_t$, we have that either $V_t^\ell \cap \dead_t = \emptyset$ or $V_t^\ell \subseteq \dead_t$.
\end{lem}

\begin{proof}
  If $x \in V_t^\ell \setminus \dead_t \neq \emptyset$, then there exists an edge $\{x,y\} \in E \setminus E_t$. Let $\tilde{t} \in V(\tree_\cexpr)$ be the lowest ancestor of $t$ such that $x,y \in V_{\tilde{t}}$; we either have $\tilde{t} = t$ or $\tilde{t}$ is some union node above $t$. In either case, we can assume that there are $i \neq j \in \labels_{\tilde{t}}$ with $x \in V_t^\ell \subseteq V_{\tilde{t}}^i$ and $y \in V_{\tilde{t}}^j$. By the second part of \cref{thm:irredundant_lemma}, we see that $\{x',y\} \in E \setminus E_{\tilde{t}} \subseteq E \setminus E_t$ for all $x' \in V_t^\ell \subseteq V_{\tilde{t}}^i$. Hence, $V_t^\ell \setminus \dead_t \neq \emptyset$ implies that $V_t^\ell \cap \dead_t = \emptyset$ which proves the lemma.
\end{proof}

If we do not require irredundancy, then it is easy to construct clique-expressions where \cref{thm:dead_label_irredundant} fails, i.e., $\emptyset \neq V_t^\ell \cap \dead_t \neq V_t^\ell$. By considering only irredundant clique-expressions, we can say that a whole label is dead or live which simplifies the handling of dead vertices, as we can perform a single computation once a label turns dead. In particular, without irredundancy a label containing only dead vertices at one node might get new live vertices later on in the expression and it is often unclear how to handle such cases.

The following definition formalizes the handling of live and dead labels and the dead nodes that are added when a label turns from live to dead.

\begin{dfn}\label{dfn:dead_nodes}
  Given an irredundant $\nlabels$-expression for $G=(V,E)$, the \emph{augmented syntax tree} $\augtree_\cexpr$ of $\cexpr$ is obtained from the syntax tree $\tree_\cexpr$ by inserting up to two \emph{dead nodes} directly above every join node $t = \join{i}{j}(G_{t'})$, where $t'$ is the child of $t$ in $\tree_\cexpr$, based on the following criteria:
  \begin{itemize}
    \item if $V_t^i \subseteq \dead_t \setminus \dead_{t'}$, then the node $\deadop{i}$ is inserted,
    \item if $V_t^j \subseteq \dead_t \setminus \dead_{t'}$, then the node $\deadop{j}$ is inserted,
    \item if both nodes $\deadop{i}$ and $\deadop{j}$ are inserted, then we insert them in any order.
  \end{itemize}
  We extend the notations $G_t$, $V_t$, $\dead_t$, $V_t^\ell$, for $\ell \in [\nlabels]$, to dead nodes by inheriting the values of the child node. 
  
	For every node $t \in V(\augtree_\cexpr)$ of the augmented syntax tree, we inductively define the set of \emph{live labels} $\livelabels{t} \subseteq \labels_t$ by 
  \begin{equation*}
    \livelabels{t} = 
    \begin{cases}
      \{\ell\} & \text{if } t = \intro{\ell}(v), \\
      \livelabels{t'} \setminus \{i\} & \text{if } t = \relab{i}{j}(G_{t'}), \\
      \livelabels{t'} & \text{if } t = \join{i}{j}(G_{t'}), \\
      \livelabels{t'} \setminus \{\ell\} & \text{if } t = \deadop{\ell}(G_{t'}), \\
      \livelabels{t_1} \cup \livelabels{t_2} & \text{if } t = G_{t_1} \union G_{t_2}.
   \end{cases}
  \end{equation*}
  Dually, the set of \emph{dead labels} $\deadlabels{t} \subseteq \labels_t$ is given by $\deadlabels{t} = \labels_t \setminus \livelabels{t}$.
\end{dfn}

We now show that, up to pending dead nodes, $\livelabels{t}$ contains all nonempty labels that only consist of live vertices at $t$. Due to \cref{thm:dead_label_irredundant}, no label of an irredundant $\nlabels$-expression can contain both live and dead vertices simultaneously.

\begin{lem}\label{thm:livelabels_characterization}
  Let $\cexpr$ be a nice $\nlabels$-expression of $G=(V,E)$ and $\augtree_\cexpr$ its augmented syntax tree. For any node $t \in V(\augtree_\cexpr)$ and $\ell \in \labels_t$, we have that $V^\ell_t \cap \dead_t = \emptyset$ implies $\ell \in \livelabels{t}$. If $t$ is not the child of a dead node, then we even have for every $\ell \in \labels_t$ that $V^\ell_t \cap \dead_t = \emptyset$ if and only if $\ell \in \livelabels{t}$.
\end{lem}

\begin{proof}
  First, recall that $V_t^\ell \cap \dead_t \in \{\emptyset, V_t^\ell\}$ for every $t \in V(\augtree_\cexpr)$ and $\ell \in \labels_t$ by \cref{thm:dead_label_irredundant}. We prove the statement inductively along the augmented syntax tree $\augtree_\cexpr$ by making a case distinction based on the current node type. Note that only join nodes and other dead nodes can be children of a dead node. 
  \begin{itemize}
    \item If $t = \intro{\ell}(v)$, then $\dead_t = \emptyset$ as $v$ cannot be isolated by assumption, so the statement holds.
    \item If $t = \relab{i}{j}(G_{t'})$, then $\dead_t = \dead_{t'}$ and by induction we have $\livelabels{t'} = \{\ell \in \labels_{t'} \sep V_{t'}^\ell \cap \dead_{t'} = \emptyset\}$. As all other labels stay the same, the only interesting labels are $i$ and $j$. Since $i \notin \labels_t$, there is nothing to prove for label $i$. We have that $V_t^j = V_{t'}^i \cup V_{t'}^j$ and, by niceness of $\cexpr$, we have $i,j \in \labels_{t'}$ and $j \in \livelabels{t'} \iff i,j \in \livelabels{t'}$. Now, 
    \begin{equation*}
      j \in \livelabels{t} \iff i,j \in \livelabels{t'} \iff V_{t'}^i \cap \dead_{t'} = \emptyset \wedge V_{t'}^j \cap \dead_{t'} = \emptyset \iff V_t^j \cap \dead_t = \emptyset
    \end{equation*}
    where the last equivalence follows from $\dead_t = \dead_{t'}$ and $V_t^j = V_{t'}^i \cup V_{t'}^j$. 
    \item If $t = \join{i}{j}(G_{t'})$, then $V_t^\ell = V_{t'}^\ell$ for all $\ell \in \labels_t = \labels_{t'}$ and $\dead_{t'} \subseteq \dead_t$ and hence the forward implication follows from the statement at $t'$. If $t$ is not the child of a dead node, then we even have $\dead_t = \dead_{t'}$, so also the reverse implication follows from the statement at $t'$.
    \item If $t$ is a dead node, then either the child or grandchild of $t$ is the join node $t^* = \join{i}{j}(G_{t'})$ that caused $t$ to exist. By induction, $\livelabels{t'} = \{\ell \in \labels_{t'} \sep V_{t'}^\ell \cap \dead_{t'} = \emptyset\}$. We have that $V_{t^*}^\ell = V_t^\ell = V_{t'}^\ell$ for every $\ell \in \labels_t = \labels_{t'}$, $\dead_t \setminus \dead_{t'} = \dead_{t^*} \setminus \dead_{t'} \subseteq V_t^i \cup V_t^j$ and $\livelabels{t} = \livelabels{t'} \setminus X$, where $\emptyset \neq X \subseteq \{i,j\}$ is the set of labels that were removed from the live labels between $t^*$ and $t$. Since $\livelabels{t} \setminus \{i,j\} = \livelabels{t'} \setminus \{i,j\}$, the equivalence holds for all $\ell \in \labels_t \setminus \{i,j\}$ by induction. For $\ell \in X = \livelabels{t'} \setminus \livelabels{t}$, we have $V_t^\ell \subseteq \dead_t \setminus \dead_{t'}$ by construction of dead nodes and hence $V_t^\ell \cap \dead_t \neq \emptyset$, therefore the forward implication holds for all labels $\ell \in \labels_t$ at node $t$. Finally, if $t$ is not the child of a dead node, then we have removed all labels $\ell \in \labels_t$ with $V_t^\ell \subseteq \dead_{t^*} \setminus \dead_{t'} = \dead_t \setminus \dead_{t'}$, so in this case the reverse implication holds for all labels $\ell \in \labels_t$ as well.
    \item If $t$ is a union node, then $V_t^\ell = V_{t_1}^\ell \cup V_{t_2}^\ell$ for all $\ell \in [\nlabels]$ and $\dead_t = \dead_{t_1} \cup \dead_{t_2}$. If $\ell \in \labels_{t_1} \symdiff \labels_{t_2} \subseteq \labels_t$, then suppose without loss of generality that $\ell \in \labels_{t_1} \setminus \labels_{t_2}$, hence $V_{t_1}^\ell = V_t^\ell$ and $V_{t_2}^{\ell} = \emptyset$, so we see that
    \begin{equation*}
      \ell \in \livelabels{t} \iff \ell \in \livelabels{t_1} \iff V_{t_1}^\ell \cap \dead_{t_1} = \emptyset \iff V_{t_1}^\ell \cap (\dead_{t_1} \cup \dead_{t_2}) = \emptyset \iff V_t^\ell \cap \dead_t = \emptyset,
    \end{equation*}
    where the third equivalence uses $V_{t_1}^\ell \cap \dead_{t_2} = \emptyset$. If on the other hand $\ell \in \labels_{t_1} \cap \labels_{t_2} \subseteq \labels_t$, then it follows from \cref{thm:dead_label_irredundant} that $V_{t_1}^\ell \cap \dead_{t_1} = \emptyset$ if and only if $V_{t_2}^\ell \cap \dead_{t_2} = \emptyset$. Therefore,
    \begin{equation*}
      \ell \in \livelabels{t} \iff \ell \in \livelabels{t_1} \cup \livelabels{t_2} \iff V_{t_1}^\ell \cap \dead_{t_1} = \emptyset \wedge V_{t_2}^\ell \cap \dead_{t_2} = \emptyset \iff V_t^\ell \cap \dead_t = \emptyset,
    \end{equation*}
    where the last equivalence follows from $V_{t_1}^\ell \cap \dead_{t_2} = \emptyset = V_{t_2}^\ell \cap \dead_{t_1}$. \qedhere
  \end{itemize}
\end{proof}

\begin{lem}\label{thm:deadlabels_recurrences}
  Let $\cexpr$ be a nice $\nlabels$-expression of $G=(V,E)$. For every node $t \in V(\augtree_\cexpr)$, the set of dead labels $\deadlabels{t}$ satisfies the following recurrences:
  \begin{equation*}
    \deadlabels{t} = 
    \begin{cases}
      \emptyset & \text{if } t = \intro{\ell}(v), \\
      \deadlabels{t'} \setminus \{i\} & \text{if } t = \relab{i}{j}(G_{t'}), \\
      \deadlabels{t'} & \text{if } t = \join{i}{j}(G_{t'}), \\
      \deadlabels{t'} \cup \{\ell\} & \text{if } t = \deadop{\ell}(G_{t'}), \\
      \deadlabels{t_1} \cup \deadlabels{t_2} & \text{if } t = G_{t_1} \union G_{t_2}.
    \end{cases}
  \end{equation*}
\end{lem}

\begin{proof}
  For every $t \in V(\augtree_\cexpr)$, the set of nonempty labels $\labels_t$ is the disjoint union of $\livelabels{t}$ and $\deadlabels{t}$ by definition of $\deadlabels{t}$. For join and relabel nodes, we have that $\labels_t = \labels_{t'}$, so the recurrences directly follow from the recurrences for $\livelabels{t}$. For introduce nodes $t = \intro{\ell}(v)$, we have $\labels_t = \livelabels{t} = \{\ell\}$ and hence $\deadlabels{t} = \emptyset$. For relabel nodes $t = \relab{i}{j}(G_{t'})$, we have $\labels_t = \labels_{t'} \setminus \{i\}$ and hence the recurrence follows. 
  
  For a union node $t = G_{t_1} \union G_{t_2}$, we have by \cref{thm:livelabels_characterization} that $\livelabels{t_i} = \{\ell \in \labels_{t_i} \sep V_{t_i}^\ell \cap \dead_{t_i} = \emptyset\}$ for $i \in \{1,2\}$. Hence, for any $\ell \in \labels_{t_1} \cap \labels_{t_2} \subseteq \labels_t$, we have $\ell \in \livelabels{t_1} \iff \ell \in \livelabels{t_2}$ by irredundancy of $\cexpr$. Therefore, we compute
  \begin{align*}
    \deadlabels{t} & = \labels_t \setminus \livelabels{t} = (\labels_{t_1} \cup \labels_{t_2}) \setminus (\livelabels{t_1} \cup \livelabels{t_2}) \\
    & = (\labels_{t_1} \setminus (\livelabels{t_1} \cup \livelabels{t_2})) \cup (\labels_{t_2} \setminus (\livelabels{t_1} \cup \livelabels{t_2})) \\
    & = (\deadlabels{t_1} \cap (\labels_{t_1} \setminus \livelabels{t_2})) \cup (\deadlabels{t_2} \cap (\labels_{t_2} \setminus \livelabels{t_1})) \\
    & = \deadlabels{t_1} \cup \deadlabels{t_2},
  \end{align*}
  since the preceding argumentation shows that $\deadlabels{t_1} \cap \livelabels{t_2} = \emptyset = \deadlabels{t_2} \cap \livelabels{t_1}$.
\end{proof}

Since the set $\dead_t$ never shrinks when going up the augmented syntax tree $\augtree_\cexpr$, a dead vertex can never turn live again. Even though the relation between $\dead_t$ and $\livelabels{t}$ does not always hold in both directions, cf.~\cref{thm:livelabels_characterization}, the next lemma still shows that a vertex can never switch from a dead label back to a live label.

\begin{lem}\label{thm:no_revival}
  Let $\cexpr$ be a nice $\nlabels$-expression of $G = (V,E)$ and $v \in V$ a vertex. For any node $t \in V(\augtree_\cexpr)$ such that $v \in V_t^\ell$ with $\ell \in \deadlabels{t}$, we have for any ancestor $t^* \in V(\augtree_\cexpr)$ and $\ell^* \in \labels_{t^*}$ with $v \in V_{t^*}^{\ell^*}$ that $\ell^* \in \deadlabels{t^*}$.
\end{lem}

\begin{proof}
  Consider the recurrences of \cref{thm:deadlabels_recurrences} and note that for join nodes and dead nodes $t$, we have that $\deadlabels{t'} \subseteq \deadlabels{t}$, where $t'$ is the child of $t$. For union nodes $t$, we have that $\deadlabels{t_i} \subseteq \deadlabels{t}$, where $t_i$, $i \in [2]$, are the children of $t$. It remains to consider relabel nodes $t = \relab{i}{j}(G_{t'})$, where $t'$ is the child of $t$. Here we have that $\deadlabels{t} = \deadlabels{t'} \setminus \{i\}$, because label $i$ is empty at $t$. If we have $i \in \deadlabels{t'}$, then also $j \in \deadlabels{t'}$ and $j \in \deadlabels{t}$ by irredundancy of $\cexpr$ and \cref{thm:livelabels_characterization}. Since $V_t^j = V_{t'}^i \cup V_{t'}^j$, this shows that also at relabel nodes no vertex can switch from a dead label to a live label.
\end{proof}

At a union node of a $\nlabels$-expression, one often has to efficiently compute a convolution-like recurrence for the dynamic programming algorithm. The first step is to handle the labels that are nonempty at only one of the children of the union node. For these, the computation is usually trivial and the remaining part is to design a fast convolution algorithm tailored to the problem for the labels which are nonempty at both children. To encapsulate this splitting of the label set, we make the following definition.
\newcommand{\splitlabels}[2]{{\tilde{L}_{#1, #2}}}

\begin{dfn}\label{dfn:union_split}
  Let $\cexpr$ be a nice $\nlabels$-expression of $G=(V,E)$ and $t \in V(\augtree_\cexpr)$ be a union node, i.e., $t = G_{t_1} \union G_{t_2}$. The \emph{union-split at $t$} of a function $f \colon \livelabels{t} \rightarrow \states$, where $\states$ is some finite set, are the functions $f_{t,1} = f\restrict{\splitlabels{t}{1}}$, $f_{t,2} = f\restrict{\splitlabels{t}{2}}$, $f_{t,12} = f\restrict{\splitlabels{t}{12}}$, where $\splitlabels{t}{1} = \livelabels{t_1} \setminus \labels_{t_2}$, $\splitlabels{t}{2} = \livelabels{t_2} \setminus \labels_{t_1}$, $\splitlabels{t}{12} = \livelabels{t_1} \cap \livelabels{t_2}$. 
\end{dfn}

\begin{lem}\label{thm:partition_union_split}
	Let $\cexpr$ be a nice $\nlabels$-expression of $G=(V,E)$ and $t \in V(\augtree_\cexpr)$ be a union node, i.e., $t = G_{t_1} \union G_{t_2}$. We have that $\livelabels{t_1} \cap \deadlabels{t_2} = \emptyset = \deadlabels{t_1} \cap \livelabels{t_2}$ and the sets $\splitlabels{t}{1}$, $\splitlabels{t}{2}$, and $\splitlabels{t}{12}$ partition $\livelabels{t}$.
\end{lem}

\begin{proof}
  The three sets are clearly disjoint by definition. Since $\livelabels{t} = \livelabels{t_1} \cup \livelabels{t_2}$, it suffices to argue that $\splitlabels{t}{1} = \livelabels{t_1} \setminus \livelabels{t_2}$ and $\splitlabels{t}{2} = \livelabels{t_2} \setminus \livelabels{t_1}$. Note that $\splitlabels{t}{1} = \livelabels{t_1} \setminus \livelabels{t_2}$ is equivalent to $\livelabels{t_1} \cap \deadlabels{t_2} = \emptyset$ and similarly for the other equality. Without loss of generality, we consider $\splitlabels{t}{1}$. We have that $\splitlabels{t}{1} = \livelabels{t_1} \setminus \labels_{t_2} \subseteq  \livelabels{t_1} \setminus \livelabels{t_2}$, since $\livelabels{t_2} \subseteq \labels_{t_2}$. For the other direction, note that any label $\ell \in \livelabels{t_1} \cap \labels_{t_2}$ must be contained in $\livelabels{t_2}$, as otherwise we would have $V_{t_1}^\ell \cap \dead_{t_1} = \emptyset$ and $V_{t_2}^\ell \subseteq \dead_{t_2}$ by \cref{thm:livelabels_characterization}, which implies $V_t^\ell \cap \dead_t = V_{t_2}^\ell \notin \{\emptyset, V_t^\ell\}$ contradicting \cref{thm:dead_label_irredundant}.
\end{proof}

\section{Dynamic Programming Algorithms}
\subsection{Connected Vertex Cover}

In \CVC, we are given a graph $G = (V,E)$, a cost function $\cfct\colon V \rightarrow \ZZ_{> 0}$, a non-negative integer $\budget$ called the \emph{budget} and we have to decide whether there exists a subset of vertices $X$ such that $G - X$ contains no edges, $G[X]$ is connected, and $\cfct(X) \leq \budget$. We only consider \CVC instances where the costs are polynomially bounded in the input size. Furthermore, we assume that $G$ is connected and contains at least two vertices.

Given a $\nlabels$-expression $\cexpr$ for $G = (V,E)$, we can assume, after polynomial-time preprocessing, that $\cexpr$ is a nice $\nlabels$-expression by \cref{thm:nice_expression}. We want to apply the cut-and-count-technique to solve \CVC in time $\Oh^*(6^{\nlabels})$. To do so, we first pick an edge in $G$, branch on one of its endpoints $\fixedvertex$, and in this branch only consider solutions containing $\fixedvertex$. Furthermore, we sample a weight function $\wfct\colon V \rightarrow [2|V|]$ for the isolation lemma, cf.\ \cref{thm:isolation}. We perform bottom-up dynamic programming along the augmented syntax tree $\augtree_\cexpr$. At every node $t \in V(\augtree_\cexpr)$, we consider the following family of partial solutions
\begin{equation*}
	\dpsols{t}{} = \{(X,(X_L, X_R)) \in \cuts(G_t) \sep G_t - X \text{ contains no edges and } \fixedvertex \in V_t \rightarrow \fixedvertex \in X_L\}.
\end{equation*}
In other words, $\dpsols{t}{}$ contains all consistently cut vertex covers of $G_t$ such that $\fixedvertex$ is on the left side of the cut if possible. For every $t \in V(\augtree_\cexpr)$, $\ctarget \in [0, \cfct(V)]$, $\wtarget \in [0, \wfct(V)]$, we define $\dpsols{t}{\ctarget, \wtarget} = \{(X,(X_L, X_R)) \in \dpsols{t}{} \sep \cfct(X) = \ctarget, \wfct(X) = \wtarget\}$. Let $\rvertex$ denote the root node of the augmented syntax tree $\augtree_\cexpr$. By \cref{thm:cons_cut_family}, it follows that there exists a connected vertex cover $X$ of $G$ with $\cfct(X) \leq \budget$ if there exist $\ctarget \in [0, \budget]$ and $\wtarget \in [0, \wfct(V)]$ such that $\dpsols{\rvertex}{\ctarget, \wtarget}$ has odd cardinality.

To facilitate the dynamic programming algorithm, we need to analyze the behavior of a partial solution $(X,(X_L,X_R)) \in \dpsols{t}{}$ with respect to a label $V_t^\ell$, $\ell \in \labels_t$. A single vertex $v \in V_t^\ell$ can take one of the states $\bstates = \{\zero, \one_L, \one_R\}$, meaning respectively $v \notin X$, or $v \in X_L$, or $v \in X_R$. To check the feasibility of $(X, (X_L,X_R))$, it is sufficient to store for each label which vertex states appear and which do not, as the constraints implied by $\dpsols{t}{}$ are "CSP-like" and they can be evaluated for every join by considering all pairs of involved vertex states. This idea yields the power set $\powerset{\bstates}$ of $\bstates$ as the set of possible states for each label.

The power set $\powerset{\bstates}$ a priori yields eight different states per label. However, we can exclude the state $\emptyset$ and the state $\bstates = \{\zero, \one_L, \one_R\}$ from consideration. The former can be excluded, since we only need to store the state for nonempty labels. The exclusion of the state $\bstates = \{\zero, \one_L, \one_R\}$ is more subtle: any additional incident join would lead to an infeasible solution for this state, hence only dead labels, cf.\ \cref{dfn:dead_nodes}, may take this state. We return to this issue in a moment. Since it suffices to store the states of live labels, we set $\states = \powerset{\bstates} \setminus \{\emptyset, \bstates\} = \{\{\zero\}, \{\one_L\}, \{\one_R\}, \{\zero, \one_L\}, \{\zero, \one_R\}, \{\one_L, \one_R\}\}$.

Given a node $t \in V(\augtree_\cexpr)$, a \emph{$t$-signature} is a function $f \colon \livelabels{t} \rightarrow \states$. For every node $t \in V(\augtree_\cexpr)$, $\ctarget \in [0, \cfct(V)]$, $\wtarget \in [0, \wfct(V)]$, and $t$-signature $f$, we define 
\begin{align*}
	\dpsols{t}{\ctarget,\wtarget}(f) = \{(X,(X_L, X_R)) \in \dpsols{t}{\ctarget, \wtarget} \sep\,\, 
  & \zero  \in f(\ell) \leftrightarrow V_t^\ell \not\subseteq X \tfa \ell \in \livelabels{t}, \\
  & \one_L \in f(\ell) \leftrightarrow X_L \cap V_t^\ell \neq \emptyset \tfa \ell \in \livelabels{t}, \\
  & \one_R \in f(\ell) \leftrightarrow X_R \cap V_t^\ell \neq \emptyset \tfa \ell \in \livelabels{t}\}.
\end{align*}
Instead of computing the sets $\dpsols{t}{\ctarget,\wtarget}(f)$ directly, we compute only the parity of their cardinality, i.e., $\dppoly{t}{\ctarget,\wtarget}(f) = |\dpsols{t}{\ctarget,\wtarget}(f)| \mod 2$.

We can now argue more formally that the exclusion of the states $\emptyset$ and $\bstates$ does not cause issues. First, for any nonempty $V_t^\ell$ at least one of the three cases $V_t^\ell \not\subseteq X$, $X_L \cap V_t^\ell \neq \emptyset$, or $X_R \cap V_t^\ell \neq \emptyset$ has to occur, hence the state $\emptyset$ cannot be attained by any $V_t^\ell$ with $\ell \in \livelabels{t}$.  

Secondly, consider some node $t$ that is not the child of a dead node, and $(X, (X_L, X_R)) \in \dpsols{t}{\ctarget, \wtarget}$ such that there is some live label $\ell \in \livelabels{t}$ for which the three cases $V_t^\ell \not\subseteq X$, $X_L \cap V_t^\ell \neq \emptyset$, and $X_R \cap V_t^\ell \neq \emptyset$ simultaneously occur. Since $\ell$ is a live label, there is some $v \in N_G(V_t^\ell) \setminus N_{G_t}(V_t^\ell)$ by \cref{thm:livelabels_characterization}. We claim that $(X, (X_L, X_R))$ cannot be extended to a consistently cut vertex cover $(X', (X'_L, X'_R))$ of $G' = G[V_t \cup \{v\}]$ (and hence also not of $G$). If $v \notin X'$, then there is an uncovered edge in $G'$ between $V_t^\ell$ and $v$. If $v \in X'$, then there is an edge in $G'$ crossing the cut $(X'_L, X'_R)$ and so the cut cannot be consistent. Hence, we can safely discard any partial solutions that attain the state $\{\zero, \one_L, \one_R\}$ with a live label, as they can never be extended to a global solution.

The state $\bstates = \{\zero, \one_L, \one_R\}$ can be obtained when two sets of vertices become united under a common label, i.e., during a relabel-operation or union-operation. We will give recurrences for the quantities $\dpsols{t}{\ctarget,\wtarget}(f)$, where $f$ is a $t$-signature which is not allowed to attain the state $\bstates$ for any label, hence such situations are implicitly filtered out in the algorithm as the recurrences simply do not consider state combinations that lead to $\bstates$.

We proceed to give recurrences for computing $\dppoly{t}{\ctarget,\wtarget}(f)$, for every $t \in V(\augtree_\cexpr)$, $t$-signature $f$, $\ctarget \in [0, \cfct(V)]$, $\wtarget \in [0, \wfct(V)]$ depending on the type of the considered node $t$.

\subparagraph*{Introduce node.} 
If $t = \intro{\ell}(v)$ for some $\ell \in [\nlabels]$, then $\livelabels{t} = \{\ell\}$ and
\begin{align*}
  \dppoly{t}{\ctarget,\wtarget}(f) & = [v \neq \fixedvertex \vee f(\ell) = \{\one_L\}] \\
  & \,\,\cdot\,\, [(f(\ell) = \{\zero\} \wedge \ctarget = \wtarget = 0) \vee (f(\ell) \in \{\{\one_L\}, \{\one_R\}\} \wedge \ctarget = \cfct(v) \wedge \wtarget = \wfct(v))], 
\end{align*}
since in a singleton graph any choice of singleton state leads to a valid solution, but if $v = \fixedvertex$ then only the solution with $\fixedvertex$ on the left side is allowed.

\subparagraph*{Relabel node.}
If $t = \relab{i}{j}(G_{t'})$, where $t'$ is the child of $t$, for some $i, j \in [\nlabels]$, then by niceness of $\cexpr$ it follows that $i \in \labels_{t'}$, $j \in \labels_{t'}$, $\labels_t = \labels_{t'} \setminus \{i\}$ and either $i,j \in \livelabels{t'}$ or $i,j \in \deadlabels{t'}$. 
\begin{itemize}
 \item If labels $i$ and $j$ are live at $t'$, then label $j$ is live at $t$ and the recurrence is given by 
 \begin{align*} 
   \dppoly{t}{\ctarget,\wtarget}(f) = \sum_{\mathclap{\substack{\stateset_1, \stateset_2 \in \states \colon \\ \stateset_1 \cup \stateset_2 = f(j)}}} \dppoly{t'}{\ctarget, \wtarget}(f[i \mapsto \stateset_1, j \mapsto \stateset_2]),
 \end{align*}
 since $V_t^j = V_{t'}^i \cup V_{t'}^j$ and we simply have to iterate over all possible combinations of previous states at labels $i$ and $j$ that yield the desired state $f(j)$.
 \item If labels $i$ and $j$ are dead at $t'$, then label $j$ is dead at $t$ and since we do not track the state of dead labels, we can simply copy the previous table, i.e.,
 \begin{equation*}
   \dppoly{t}{\ctarget, \wtarget}(f) = \dppoly{t'}{\ctarget, \wtarget}(f).
 \end{equation*}
\end{itemize}

\subparagraph*{Join node.} 
To check whether two states can lead to a feasible solution after adding a join between their labels, we introduce a helper function $\feas \colon \states \times \states \rightarrow \{0,1\}$ defined by $\feas(\stateset_1, \stateset_2) = [\zero \notin \stateset_1 \vee \zero \notin \stateset_2] [\one_L \in \stateset_1 \rightarrow \one_R \notin \stateset_2] [\one_R \in \stateset_1 \rightarrow \one_L \notin \stateset_2]$, or equivalently by the following table:
\begin{equation*}
  \begin{array}{r|cccccc}
    \feas & \{\zero\} & \{\one_L\} & \{\one_R\} & \{\zero, \one_L\} & \{\zero, \one_R\} & \{\one_L, \one_R\} \\
    \hline
    \{\zero\} &          0 & 1 & 1 & 0 & 0 & 1 \\
    \{\one_L\} &         1 & 1 & 0 & 1 & 0 & 0 \\
    \{\one_R\} &         1 & 0 & 1 & 0 & 1 & 0 \\
    \{\zero, \one_L\} &  0 & 1 & 0 & 0 & 0 & 0 \\
    \{\zero, \one_R\} &  0 & 0 & 1 & 0 & 0 & 0 \\
    \{\one_L, \one_R\} & 1 & 0 & 0 & 0 & 0 & 0
  \end{array}
\end{equation*} 
There are two reasons for infeasibility: a join edge is not covered, i.e., $\zero$ appears on both sides, or a join edge connects both sides of the cut, i.e., $\one_L$ appears on one side and $\one_R$ on the other. Using this helper function, we can now state the recurrence for the join case.

We have that $t = \join{i}{j}(G_{t'})$ for some $i \neq j \in \labels_{t'}$ and where $t'$ is the child of $t$. We must have $i, j \in \livelabels{t'}$ and if the set of dead vertices changes, i.e., $\dead_t \neq \dead_{t'}$, then this will be handled by future dead nodes. Hence, we simply have to filter out all partial solutions that became infeasible due to the new join:
 \begin{align*}
  \dppoly{t}{\ctarget,\wtarget}(f) = \feas(f(i), f(j))\dppoly{t'}{\ctarget, \wtarget}(f).
 \end{align*}
 
\subparagraph*{Dead node.}
We have that $t = \deadop{\ell}(G_{t'})$, where $t'$ is the child of $t$, $\ell \notin \livelabels{t}$, and $\livelabels{t} = \livelabels{t'} \setminus \{\ell\}$. Since the only change is that $t$-signatures do not track the state of label $\ell$ anymore, we have to add up the contributions of all previous states of label $\ell$. Hence, the recurrence is given by
\begin{equation*}
	\dppoly{t}{\ctarget, \wtarget}(f) = \sum_{\stateset \in \states} \dppoly{t'}{\ctarget, \wtarget}(f[\ell \mapsto \stateset]).
\end{equation*}

\subparagraph*{Union node.} 
We have that $t = G_{t_1} \union G_{t_2}$, where $t_1$ and $t_2$ are the children of $t$ and we have $\livelabels{t} = \livelabels{t_1} \cup \livelabels{t_2}$. Given a $t$-signature $f$, we consider the union-split $f_{t,1}$, $f_{t,2}$, $f_{t,12}$ of $f$ at $t$, cf.~\cref{dfn:union_split}. For every label $\ell \in \splitlabels{t}{12} = \livelabels{t_1} \cap \livelabels{t_2}$, we need to consider all states $\stateset_1, \stateset_2 \in \states$ such that $\stateset_1 \cup \stateset_2 = f(\ell)$, where $\stateset_i$ is the state of label $\ell$ at $t_i$. Furthermore, we have to distribute $\ctarget$ and $\wtarget$ among the partial solutions at $t_1$ and the partial solutions at $t_2$. Hence, we obtain the recurrence
\begin{align*}
  \dppoly{t}{\ctarget,\wtarget}(f) = \sum_{\substack{\ctarget_1 + \ctarget_2 = \ctarget \\ \wtarget_1 + \wtarget_2 = \wtarget}} \sum_{\substack{g_1, g_2 \colon \splitlabels{t}{12} \rightarrow \states \colon  \\ g_1(\ell) \cup g_2(\ell) = f_{t,12}(\ell) \,\forall \ell \in \splitlabels{t}{12}}} \dppoly{t_1}{\ctarget_1,\wtarget_1}(g_1 \cup f_{t,1}) \dppoly{t_2}{\ctarget_2,\wtarget_2}(g_2 \cup f_{t,2}),
\end{align*}
where we consider $g_i$ and $f_{t,i}$ as sets in $g_i \cup f_{t,i}$ for $i \in [2]$.

We now argue how this recurrence can be computed in time $\Oh^*(6^{|\livelabels{t}|})$ for fixed $\ctarget \in [0, \cfct(V)]$, $\wtarget \in [0, \wfct(V)]$, and for all $t$-signatures $f$. We first branch on all possibilities for $(\ctarget_1, \ctarget_2, \wtarget_1, \wtarget_2)$; these are $n^{\Oh(1)}$ possibilities as $\cfct(V) \leq n^{\Oh(1)}$ by assumption and $\wfct(V) \leq 2n^2$. Fixing one of these possibilities, we further branch on $f_{t,1}$ and $f_{t,2}$, this leads to $6^{|\splitlabels{t}{1}| + |\splitlabels{t}{2}|}$ choices. Now, the quantities $\dppoly{t_1}{\ctarget_1,\wtarget_1}(g_1 \cup f_{t,1})$ and $\dppoly{t_2}{\ctarget_2,\wtarget_2}(g_2 \cup f_{t,2})$ can be considered as functions of $g_1$ and $g_2$ respectively and the inner sum over these in the recurrence is their componentwise cover product over $\states$ evaluated at $f_{t,12}$.

Since the set $\states$ is clearly a closure difference, we can apply \cref{thm:fast_compwise_cover_product} to compute this componentwise cover product in time $\Oh^*(6^{|\splitlabels{t}{12}|})$ for all possible $f_{t,12}$. Since $\livelabels{t}$ is partitioned into $\splitlabels{t}{1}$, $\splitlabels{t}{2}$, and $\splitlabels{t}{12}$ by \cref{thm:partition_union_split}, we need in total time $\Oh^*(6^{|\livelabels{t}|})$ to compute $\dppoly{t}{\ctarget,\wtarget}(f)$ for all choices of $\ctarget$, $\wtarget$, and $f$.

\begin{lem}\label{thm:cvc_count_part}
 Given a nice $\nlabels$-expression $\cexpr$ of $G = (V,E)$, there is an algorithm that computes the quantities $\dppoly{t}{\ctarget,\wtarget}(f)$ for all nodes $t \in V(\augtree_\cexpr)$, all $t$-signatures $f$, and all $\ctarget \in [0,\cfct(V)]$, $\wtarget \in [0,2n^2]$, in time $\Oh^*(6^{\nlabels})$.
\end{lem}

\begin{proof}
  The algorithm proceeds by bottom-up dynamic programming along the augmented syntax tree $\augtree_\cexpr$ of the nice $\nlabels$-expression $\cexpr$ and computes the quantities $\dppoly{t}{\ctarget,\wtarget}(f)$ via the given recurrences. For an introduce node, relabel node, or join node, the recurrence for $\dppoly{t}{\ctarget,\wtarget}(f)$ for fixed $f$, $\ctarget$, and $\wtarget$, can clearly be computed in polynomial time, since additions and multiplications in $\ZZ_2$ take constant time. For a union node $t$, we have argued how to compute the recurrences for all $f$, $\ctarget$, and $\wtarget$ simultaneously in time $\Oh^*(6^{|\livelabels{t}|})$. As $\cexpr$ is a $\nlabels$-expression, we have $|\livelabels{t}| \leq \nlabels$ for all $t \in V(\augtree_\cexpr)$ and in particular at most $6^\nlabels$ $t$-signatures for any node $t \in V(\tree_\cexpr)$. Hence, the running time follows.
 
  It remains to prove the correctness of the recurrences. The proof of correctness for introduce nodes, relabel nodes, join nodes, and union nodes is straightforward and hence omitted. Consider a dead node $t = \deadop{\ell}(G_{t'})$, where $t'$ is the child of $t$. We claim that $\dpsols{t}{\ctarget, \wtarget}(f) = \bigcup_{\stateset \in \states} \dpsols{t'}{\ctarget, \wtarget}(f[\ell \mapsto \stateset])$ for all $t$-signatures $f$, $\ctarget$, and $\wtarget$. Since the union on the right-hand side is clearly disjoint, this claim immediately proves the correctness of the recurrence for dead nodes. We proceed with proving the claim. The right-hand side is contained in the left-hand side, since the set on the left-hand side is defined by fewer constraints. 
  
  For the other direction of the claim, suppose there is some $(X, (X_L, X_R)) \in \dpsols{t}{\ctarget, \wtarget}(f) \setminus \left( \bigcup_{\stateset \in \states} \dpsols{t'}{\ctarget, \wtarget}(f[\ell \mapsto \stateset]) \right)$ and consider $V_t^\ell$. Since $V_t^\ell \neq \emptyset$, at least one of $V_t^\ell \not\subseteq X$, $X_L \cap V_t^\ell \neq \emptyset$, $X_R \cap V_t^\ell \neq \emptyset$ has to be satisfied. Indeed, all three statements are satisfied simultaneously, because all remaining cases are covered by $\states$. Consider the join node $t^*$ that caused the dead node $t$ to exist. The node $t^*$ adds the final join incident to $V_t^\ell = V_{t^*}^\ell$, say between $V_t^\ell$ and some $V_t^i = V_{t^*}^i$, $i \neq \ell$, and by niceness of $\cexpr$ we have $V_t^i \neq \emptyset$. Hence, at least one of $V_t^i \not\subseteq X$, $X_L \cap V_t^i \neq \emptyset$, $X_R \cap V_t^i \neq \emptyset$ has to be satisfied. Therefore, there is an uncovered edge between $V_t^\ell$ and $V_t^i$ or an edge crossing the cut $(X_L, X_R)$, contradicting that $(X, (X_L, X_R)) \in \dpsols{t}{\ctarget,\wtarget}(f)$. Hence, $\dpsols{t}{\ctarget, \wtarget}(f) \subseteq \bigcup_{\stateset \in \states} \dpsols{t'}{\ctarget, \wtarget}(f[\ell \mapsto \stateset])$, proving the claim.
\end{proof}

\begin{thm}
 There is a randomized algorithm that given a nice $\nlabels$-expression $\cexpr$ for a graph $G = (V, E)$ can solve \CVC in time $\Oh^*(6^k)$. The algorithm does not return false positives and returns false negatives with probability at most $1/2$.
\end{thm}

\begin{proof}
  We begin by sampling a weight function $\wfct\colon V \rightarrow [2n]$ uniformly at random. Then, we pick an edge in $G$ and branch on its endpoints; the chosen endpoint takes the role of $\fixedvertex$ in the current branch. We then run the algorithm of \cref{thm:cvc_count_part} to compute the quantities $\dppoly{t}{\ctarget, \wtarget}(f)$. Let $\rvertex$ denote the root node of the expression $\cexpr$. At the root, we have that $\livelabels{\rvertex} = \emptyset$. The algorithm returns true if in one of the branches there is some choice of $\ctarget \leq \budget$, $\wtarget \in [0, 2n^2]$, such that $\dppoly{\rvertex}{\ctarget, \wtarget}(\emptyset) \neq 0$, otherwise the algorithm returns false.
  
  The running time directly follows from \cref{thm:cvc_count_part}. For the correctness, first note that at the root, we have $\dpsols{\rvertex}{\ctarget, \wtarget}(\emptyset) = \dpsols{\rvertex}{\ctarget, \wtarget}$. The algorithm only returns true, if there are some $\ctarget \in [0, \budget]$, $\wtarget \in [0, 2n^2]$ such that $\dpsols{\rvertex}{\ctarget, \wtarget}$ has odd cardinality. By \cref{thm:cons_cut_family}, this implies that there is a connected vertex cover $X$ of $G$ with $\cfct(X) = \ctarget \leq \budget$, hence the algorithm does not return false positives.
  
  For the error probability, suppose that the weight function $\wfct$ isolates an optimum connected vertex cover $X^*$ of $G$ and that $\cfct(X^*) \leq \budget$; by \cref{thm:isolation}, the isolation happens with probability greater than or equal to $1/2$. Furthermore, consider a branch with $\fixedvertex \in X^*$. Set $\ctarget = \cfct(X^*) \leq \budget$ and $\wtarget = \wfct(X^*)$. By \cref{thm:cons_cut}, the connected vertex cover $X^*$ contributes an odd number to $|\dpsols{\rvertex}{\ctarget, \wtarget}|$ and all other contributing sets cannot be connected due to isolation and hence contribute an even number to $|\dpsols{\rvertex}{\ctarget, \wtarget}|$. Therefore $\dppoly{\rvertex}{\ctarget, \wtarget}(\emptyset) \equiv_2 |\dpsols{\rvertex}{\ctarget, \wtarget}| = 1 \neq 0$ and hence the algorithm returns true with probability at least $1/2$ given a positive instance.
\end{proof}
\subsection{Connected Dominating Set}
\label{sec:cw_cds_algo}

\newcommand{\sA}{\mathbf{A}}
\newcommand{\sF}{\mathbf{F}}
\newcommand{\sL}{\mathbf{L}}
\newcommand{\sR}{\mathbf{R}}
\newcommand{\sQ}{\mathbf{2_+}}

In the \CDS problem, we are given a graph $G = (V,E)$, a cost function $\cfct \colon V \rightarrow \ZZ_{> 0}$ and a non-negative integer $\budget$ and the task is to decide whether there exists a set of vertices $X \subseteq V$ with $\cfct(X) \leq \budget$ such that $X$ is a dominating set of $G$, i.e., $N[X] = V$, and $G[X]$ is connected. We only consider \CDS instances where the costs are polynomially bounded in the input size. Furthermore, we assume that $G$ is connected and contains at least two vertices.

We begin by motivating our algorithmic approach for \CDS. Following the approach for \CVC, we would consider partial solutions $(X, (X_L, X_R))$ consisting of a \emph{partial} dominating set $X$ and a consistent cut $(X_L, X_R)$ of the subgraph induced by $X$. A single vertex $v$ can take four states with respect to $(X, (X_L, X_R))$: $\zero_1$, $\zero_0$, $\one_L$, $\one_R$, where the former two indicate that $v \notin X$ and the subscript denotes whether $v$ is dominated by $X$ or not, and the latter two indicate that $v \in X$ and the subscript denotes which cut side contains $v$. We can again store for each label which vertex states appear, yielding as possible label states all subsets of $\{\zero_1, \zero_0, \one_L, \one_R\}$. By observing that the state $\zero_1$ does not impose any constraint for future joins, we can even argue that it suffices to only consider the subsets of $\{\zero_0, \one_L, \one_R\}$. Furthermore, similar to \CVC, the label state $\{\zero_0, \one_L, \one_R\}$ cannot be sensibly attained by live labels, hence we are down to seven states per live label. This approach yields a running time of $\Oh^*(7^{\cw(G)})$, but we can obtain an even faster algorithm.

To obtain the improved running time of $\Oh^*(5^{\cw(G)})$, we instead work with a different set of vertex states common for domination problems~\cite{HegerfeldK20, NederlofRD14, PilipczukW18, RooijBR09}. Instead of considering partial solutions with \emph{dominated} and \emph{undominated} vertices, we consider \emph{allowed} vertices (state $\sA$) and \emph{forbidden} vertices (state $\sF$). As their names imply, allowed vertices may be dominated or undominated, but forbidden vertices may not be dominated. When lifting the \emph{vertex} states to \emph{label} states, the state $\sA$ can be ignored, because it imposes no constraint on joins, therefore we obtain the subsets of $\{\sF, \sL, \sR\}$ as \emph{label states}. The advantage of this set of states is that all subsets of size at least two behave the same with respect to joins, allowing us to collapse them to a single state, and that we do not have to update states from undominated to dominated when handling joins; this step yields the desired five label states.

However, recovering the solutions to the original problem from this set of states usually requires some type of inclusion-exclusion argument. The application of this step is non-standard for clique-width. For sparse graph parameters, such as treewidth, the inclusion-exclusion argument can be applied to \emph{single} vertices, i.e., if we subtract the partial solutions where a vertex $v$ has state $\sF$ from those where $v$ has state $\sA$, then only partial solutions dominating $v$ remain. For clique-width however, we have to apply the argument to groups of vertices and such a subtraction would only yield that \emph{some} vertices in the label must be dominated and not, as is desired, all of them. Moreover, the collapsing of several label states into a single one complicates the inclusion-exclusion argument further. Surprisingly, working modulo 2 resolves all of these problems simultaneously and it is also the natural setting for the cut-and-count-technique. Lastly, the inclusion-exclusion argument should only be applied when all edges incident to a label are already constructed, i.e., the considered label is \emph{dead}, hence we again use augmented syntax trees.

We proceed by giving the formal details of the algorithm. Given a $\nlabels$-expression $\cexpr$ for $G=(V,E)$, we can assume that $\cexpr$ is nice after polynomial-time preprocessing, see \cref{thm:nice_expression}. We sample a weight function $\wfct \colon V \rightarrow [2n]$ for the isolation lemma, cf.~\cref{thm:isolation}. To solve \CDS, we perform bottom-up dynamic programming along the augmented syntax tree $\augtree_\cexpr$ of $\cexpr$. We pick some $\fixedvertex \in V$ and only consider solutions containing $\fixedvertex$ for now.

To implement the inclusion-exclusion and cut-and-count approach, the dynamic programming algorithm considers the following family of partial solutions at a node $t \in V(\augtree_\cexpr)$.
\begin{dfn}
  At a node $t \in V(\augtree_\cexpr)$, the family $\dpsols{t}{}$ of \emph{partial solutions} consists of all ordered subpartitions $(X_L, X_R, F)$ of $V_t$ satisfying the following properties:
  \begin{itemize}
    \item $X = X_L \cup X_R$, $(X_L, X_R)$ is a consistent cut of $G_t[X]$,
    \item $\fixedvertex \in X_L$ if $\fixedvertex \in V_t$,
    \item $N_{G_t}[X] \cap F = \emptyset$,
    \item $V_t^\ell \subseteq N_{G_t}[X]$ for all $\ell \in \deadlabels{t}$.
  \end{itemize}
  Furthermore, for every node $t \in V(\augtree_\cexpr)$, $\ctarget \in [0, \cfct(V)]$, and $\wtarget \in [0, \wfct(V)]$, we define $\dpsols{t}{\ctarget, \wtarget} = \{ (X_L, X_R, F) \in \dpsols{t}{} \sep \cfct(X_L \cup X_R) = \ctarget, \wfct(X_L \cup X_R) = \wtarget\}$.
\end{dfn}

Essentially, every $(X_L, X_R, F) \in \dpsols{t}{}$ consists of a consistently cut \emph{partial} dominating set $X = X_L \cup X_R$ of $G_t$ that dominates all vertices with dead labels and does \emph{not} dominate the vertices in $F$. To any $(X_L, X_R, F) \in \dpsols{t}{}$, there is also an associated set of vertices $A = V_t \setminus (X_L \cup X_R \cup F)$ that are \emph{allowed} to be dominated.

Notice that for any $(X_L, X_R, F) \in \dpsols{t}{}$, we must have that $F \cap \left(\bigcup_{\ell \in \deadlabels{t}} V_t^\ell \right) = \emptyset$ as otherwise the third and fourth property in the definition of $\dpsols{t}{}$ cannot be simultaneously satisfied. In particular, for the root node $\rvertex$ we must have $F = \emptyset$ and $\deadlabels{\rvertex} = \labels_{\rvertex}$, so that $\dpsols{\rvertex}{}$ only contains consistently cut dominating sets of $G$. Hence, if there are $\ctarget \in [0, \budget]$, $\wtarget \in [0,2n^2]$, such that $\dpsols{\rvertex}{\ctarget, \wtarget}$ has odd cardinality, then there exists a connected dominating set $X$ of $G$ with $\cfct(X) \leq \budget$ by \cref{thm:cons_cut_family}.

To compute the sets $\dpsols{t}{\ctarget, \wtarget}$ via dynamic programming, we partition the partial solutions according to their states on the live labels $\livelabels{t}$. The state of a partial solution $(X_L, X_R, F) \in \dpsols{t}{}$ at a label $\ell \in \livelabels{t}$ is based on which of the sets $X_L$, $X_R$, and $F$ are intersected by $V_t^\ell$. To capture this, we make the following definition.

\begin{dfn}
  Let $t \in V(\augtree_\cexpr)$ be a node of $\augtree_\cexpr$ and $\ell \in \labels_{t}$ be a nonempty label. Given an ordered subpartition $(X_L, X_R, F)$ of $V_t$, we define the set $\stateset_t^\ell(X_L, X_R, F) \subseteq \{\sF, \sL, \sR\}$ by
  \begin{itemize}
    \item $\sF \in \stateset_t^\ell(X_L, X_R, F) \iff F \cap V_t^\ell \neq \emptyset$,
    \item $\sL \in \stateset_t^\ell(X_L, X_R, F) \iff X_L \cap V_t^\ell \neq \emptyset$,
    \item $\sR \in \stateset_t^\ell(X_L, X_R, F) \iff X_R \cap V_t^\ell \neq \emptyset$.
  \end{itemize}
\end{dfn}

Naively, using the sets $\stateset_t^\ell(X_L, X_R, F)$ would yield $2^3 = 8$ states per label. Note that also $\stateset_t^\ell(X_L, X_R, F) = \emptyset$ is sensible for $(X_L, X_R, F) \in \dpsols{t}{}$ and nonempty label $\ell \in \labels_t$, because this simply means $V_t^\ell \subseteq A = V_t \setminus (X_L \cup X_R \cup F)$, i.e., all vertices in $V_t^\ell$ are allowed to be dominated and not part of the partial dominating set $X_L \cup X_R$. Surprisingly, it turns out that all $\stateset \subseteq \{\sF, \sL, \sR\}$ with $|\stateset| \geq 2$ can be handled in the same way, since all such $\stateset$ can only be feasibly joined to the state $\emptyset$. This enables us to solve \CDS with 5 states per label instead of 8. We define the set $\states = \{\emptyset, \{\sF\}, \{\sL\}, \{\sR\}, \sQ\}$, where $\sQ$ is not a subset of $\{\sF, \sL, \sR\}$ but a formal symbol representing the subsets of size at least 2. 

\begin{dfn}
  Given a node $t \in V(\augtree_\cexpr)$, a \emph{$t$-signature} is a function $f \colon \livelabels{t} \rightarrow \states$. 
  
  A subpartition $(X_L, X_R, F)$ of $V_t$ is \emph{compatible} with a $t$-signature $f$ if for all $\ell \in \livelabels{t}$ the following two properties hold:
  \begin{itemize}
   \item $\stateset_t^\ell(X_L, X_R, F) = f(\ell)$ if $f(\ell) \neq \sQ$,
   \item $|\stateset_t^\ell(X_L, X_R, F)| \geq 2$ if $f(\ell) = \sQ$.
  \end{itemize}
  Furthermore, for every node $t \in V(\augtree_\cexpr)$, $t$-signature $f$, $\ctarget \in [0, \cfct(V)]$, and $\wtarget \in [0, \wfct(V)]$, we define $\dpsols{t}{\ctarget, \wtarget}(f) = \{ (X_L, X_R, F) \in \dpsols{t}{\ctarget, \wtarget} \sep (X_L, X_R, F) \text{ is compatible with } f\}$.
\end{dfn}

Unlike \CVC, there are no states in \CDS that can only appear at dead labels, since the state $\emptyset$ can be joined to any state without making the partial solution infeasible. Regardless, distinguishing live and dead labels remains useful, as we only want to require the domination of vertices with dead labels; vertices with live labels may be dominated by an edge that is missing at the current node. Hence, the dead nodes of $\augtree_\cexpr$ serve as natural nodes to apply the inclusion-exclusion step.

As usual, we do not compute the sets $\dpsols{t}{\ctarget, \wtarget}(f)$ directly, but the parity of their cardinality, i.e., $\dppoly{t}{\ctarget, \wtarget}(f) = |\dpsols{t}{\ctarget, \wtarget}(f)| \mod 2$. We now proceed by presenting the various recurrences for $\dppoly{t}{\ctarget, \wtarget}(f)$, given node $t \in V(\augtree_\cexpr)$, $t$-signature $f$, $\ctarget \in [0, \cfct(V)]$, $\wtarget \in [0, \wfct(V)]$, based on the type of the node $t$.

\subparagraph*{Introduce node.}
If $t = \ell(v)$ for some $\ell \in [\nlabels]$ and $v \in V$, then $\livelabels{t} = \{\ell\}$ as $v$ cannot be an isolated vertex by assumption and 
\begin{align*}
  \dppoly{t}{\ctarget, \wtarget}(f) & = [v \neq \fixedvertex \vee f(\ell) = \{\sL\}] \\
                                    & \,\,\cdot\,\, [(\ctarget = \wtarget = 0 \wedge f(\ell) \in \{\emptyset, \{\sF\}\}) \vee (\ctarget = \cfct(v) \wedge \wtarget = \wfct(v) \wedge f(\ell) \in \{\{\sL\}, \{\sR\}\})],
\end{align*}
first checking for the edge case that $v = \fixedvertex$ and then that $\ctarget$ and $\wtarget$ agree with the chosen state. Note that the state $\sQ$ cannot be achieved here.

\subparagraph*{Relabel node.} 
If $t = \relab{i}{j}(G_{t'})$, where $t'$ is the only child of $t$, then label $i$ and $j$ are nonempty at node $t'$, since $\cexpr$ is a nice expression. By irredundancy of the expression $\cexpr$, either $\{i,j\} \subseteq \deadlabels{t'}$ or $\{i,j\} \cap \deadlabels{t'} = \emptyset$. In the first case, we have that $j \in \deadlabels{t}$ and the recurrence is simply $\dppoly{t}{\ctarget, \wtarget}(f) = \dppoly{t'}{\ctarget, \wtarget}(f)$, since we do not store the state of dead labels and the domination requirement for dead labels remains satisfied.

In the second case, we have $j \in \livelabels{t}$ and we must iterate through all pairs of states that combine to the current state at label $j$. Since $V_t^j = V_{t'}^i \cup V_{t'}^j$, the recurrence is given by
\begin{equation*}
  \dppoly{t}{\ctarget, \wtarget}(f) = \sum_{\mathclap{\substack{\stateset_1, \stateset_2 \in \states\colon \\ \merge(\stateset_1, \stateset_2) = f(j)}}} \dppoly{t'}{\ctarget, \wtarget}(f[i \mapsto \stateset_1, j \mapsto \stateset_2]),
\end{equation*}
where $\merge \colon \states \times \states \rightarrow \states$ is given by:
\begin{equation*}
  \begin{array}{l|ccccc}
    \merge & \emptyset & \{\sF\} & \{\sL\} & \{\sR\} & \sQ \\
    \hline
    \emptyset & \emptyset & \{\sF\} & \{\sL\} & \{\sR\} & \sQ \\
    \{\sF\} & \{\sF\} & \{\sF\} & \sQ & \sQ & \sQ \\
    \{\sL\} & \{\sL\} & \sQ & \{\sL\} & \sQ & \sQ \\
    \{\sR\} & \{\sR\} & \sQ & \sQ & \{\sR\} & \sQ \\
    \sQ & \sQ & \sQ & \sQ & \sQ & \sQ \\
  \end{array}
\end{equation*}
Note that $\stateset_t^j(X_L, X_R, F) = \stateset_{t'}^i(X_L, X_R, F) \cup \stateset_{t'}^j(X_L, X_R, F)$ and if label $i$ had state $\stateset_1$ at $t'$ and label $j$ state $\stateset_2$ at $t'$, then label $j$ has state $\merge(\stateset_1, \stateset_2)$ at node $t$.

\subparagraph*{Join node.}
If $t = \join{i}{j}(G_{t'})$, $i \neq j$, where $t'$ is the only child of $t$, then we know that $i, j \in \livelabels{t'}$ since the expression $\cexpr$ is nice. Due to this join, the vertices with label $i$ or $j$ could receive their final incident edges, possibly leading to $V_t^i \subseteq \dead_t$ or $V_t^j \subseteq \dead_t$. If this happens, this join will be followed by up to two dead nodes. At this join, we filter out the partial solutions that are invalidated by the newly added edges and the recurrence is given by
\begin{equation*}
  \dppoly{t}{\ctarget, \wtarget}(f) = \feas(f(u), f(v)) \dppoly{t'}{\ctarget, \wtarget}(f),
\end{equation*}
where $\feas \colon \states \times \states \rightarrow \{0,1\}$ is given by the following table:
\begin{equation*}
\begin{array}{l|ccccc}
  \feas   & \emptyset & \{\sF\} & \{\sL\} & \{\sR\} & \sQ \\
   \hline
   \emptyset & 1 & 1 & 1 & 1 & 1 \\
   \{\sF\} & 1 & 1 & 0 & 0 & 0 \\
   \{\sL\} & 1 & 0 & 1 & 0 & 0 \\
   \{\sR\} & 1 & 0 & 0 & 1 & 0 \\
   \sQ & 1 & 0 & 0 & 0 & 0 \\
\end{array}
\end{equation*}
There are two possible reasons why a partial solution $(X_L, X_R, F)$ might be invalidated; a new edge connects $F$ and $X = X_L \cup X_R$ or a new edge connects $X_L$ and $X_R$. 

\subparagraph*{Dead node.}
Suppose that $t = \deadop{\ell}(G_{t'})$, where $\ell \in \labels'_t$ and $t'$ is the child of $t$. We have that $\livelabels{t} = \livelabels{t'} \setminus \{\ell\}$ and $\ell \in \deadlabels{t}$. Due to the definition of $\dpsols{t}{}$, we only want to count the partial solutions from $\dpsols{t'}{}$ that dominate $V_t^\ell$ completely. If $V_t^\ell$ contains only a single vertex, then this is easy to check with the given states: we could simply compute for a $t$-signature $f$
\begin{equation*}
  \dppoly{t'}{\ctarget, \wtarget}(f[\ell \mapsto \{\sL\}]) + \dppoly{t'}{\ctarget, \wtarget}(f[\ell \mapsto \{\sR\}]) + (\dppoly{t'}{\ctarget, \wtarget}(f[\ell \mapsto \emptyset]) - \dppoly{t'}{\ctarget, \wtarget}(f[\ell \mapsto \{\sF\}])),
\end{equation*}
where the last part is the inclusion-exclusion argument that checks that the vertex in $V_t^\ell$ is dominated. However, if $|V_t^\ell| \geq 2$ then we also need to handle the state $\sQ$ for which it is unclear how to incorporate it into an inclusion-exclusion argument as we do not even know the precise value of $\stateset_{t'}^\ell(X_L, X_R, F)$ in this case. Furthermore, also the previous inclusion-exclusion argument is invalid over $\ZZ$ if $|V_t^\ell| \geq 2$ as partial solutions with multiple undominated vertices in $V_t^\ell$ are counted several times by $\dppoly{t'}{\ctarget, \wtarget}(f[\ell \mapsto \{\sF\}])$.

Surprisingly, there is a very simple recurrence that avoids all these issues modulo 2. If $f$ is a $t$-signature, and $\ctarget \in [0,n]$, $\wtarget \in [0,2n^2]$, then the recurrence is given by
\begin{equation*}
  \dppoly{t}{\ctarget, \wtarget}(f) = \sum_{\stateset \in \states} \dppoly{t'}{\ctarget, \wtarget}(f[\ell \mapsto \stateset]).
\end{equation*}
This recurrence works because any partial dominating set containing exactly $u$ undominated vertices in $V_t^\ell$ is counted $2^u$ times by the right-hand side of the recurrence and hence cancels modulo 2 for $u > 0$. We proceed by giving the formal proof of correctness for this recurrence.

\begin{proof}  
  Fix $f$, $\ctarget$, and $\wtarget$. The left side counts the cardinality of $\dpsols{t}{\ctarget, \wtarget}(f)$ modulo 2 and the right side clearly computes $\sum_{\stateset \in \states} |\dpsols{t'}{\ctarget, \wtarget}(f[\ell \mapsto \stateset])|$ modulo 2. We have to prove that these terms agree. For readability, set $\dpsols{LHS}{} := \dpsols{t}{\ctarget, \wtarget}(f)$ and $\dpsols{RHS}{} := \bigcup_{\stateset \in \states}\dpsols{t'}{\ctarget, \wtarget}(f[\ell \mapsto \stateset])$. 
  
  Recall that $G_t = G_{t'}$, $V_t^i = V_{t'}^i$ for all $i \in \labels_t$, and $\deadlabels{t} = \deadlabels{t'} \cup \{\ell\}$ with $\ell \notin \deadlabels{t'}$. First, notice that $\dpsols{LHS}{} \subseteq \dpsols{RHS}{}$, since for every possibility of $\stateset_{t'}^\ell(X_L, X_R, F)$, there is an $\stateset \in \states$ such that $(X_L, X_R, F)$ is compatible with $f[\ell \mapsto \stateset]$ at node $t'$. Also note that the sets $\dpsols{t'}{\ctarget, \wtarget}(f[\ell \mapsto \stateset])$ are disjoint for distinct $\stateset \in \states$ and hence every element of $\dpsols{RHS}{}$ is counted exactly once on the right side of the equation. 
  
  Next, we see that for any $(X_L, X_R, F) \in \dpsols{RHS}{}$ that $(X_L, X_R, F) \in \dpsols{LHS}{}$ holds if and only if $V_t^\ell \subseteq N_{G_t}[X]$, where $X = X_L \cup X_R$, since the only requirements that change are the compatibility with $f$, which is easier to satisfy at $t$ than at $t'$, and the requirement that all dead labels are dominated which requires $V_t^\ell \subseteq N_{G_t}[X]$ in addition to $(X_L, X_R, F) \in \dpsols{RHS}{}$.
  
  It remains to show that $\dpsols{RHS}{} \setminus \dpsols{LHS}{}$ contains an even number of elements and hence cancels modulo 2. Consider $(X_L, X_R, F) \in \dpsols{RHS}{} \setminus \dpsols{LHS}{}$, let $X = X_L \cup X_R$ and let $U_t^\ell := U_t^\ell(X_L, X_R) := V_t^\ell \setminus N_{G_t}[X]$ be the set of undominated vertices with label $\ell$ at node $t$. From $(X_L, X_R, F) \in \dpsols{RHS}{} \setminus \dpsols{LHS}{}$ it follows that $U_t^\ell \neq \emptyset$, as otherwise $V_t^\ell \subseteq N_{G_t}[X]$ which contradicts $(X_L, X_R, F) \notin \dpsols{LHS}{}$ by the previous paragraph. For any $F' \subseteq U_t^\ell$, we see that $(X_L, X_R, (F \setminus V_t^\ell) \cup F') \in \dpsols{RHS}{} \setminus \dpsols{LHS}{}$, since $U_t^\ell$ remains unchanged but we only change which vertices are declared forbidden. Since $U_t^\ell \neq \emptyset$, there is an even number of choices, namely $2^{|U_t^\ell|}$ many, for $F'$. 
  
  Fixing some $\emptyset \neq U \subseteq V_t^\ell$, this shows there are an even number of elements $(X_L, X_R, F) \in \dpsols{RHS}{} \setminus \dpsols{LHS}{}$ with $U_t^\ell(X_L, X_R) = U$. Since every element of $\dpsols{RHS}{} \setminus \dpsols{LHS}{}$ is covered by some choice of $U$, it follows that the cardinality of $\dpsols{RHS}{} \setminus \dpsols{LHS}{}$ is even. 
\end{proof}

\subparagraph*{Union node.} 
We have that $t = G_{t_1} \union G_{t_2}$, where $t_1$ and $t_2$ are the children of $t$ and we have $\livelabels{t} = \livelabels{t_1} \cup \livelabels{t_2}$. Given a $t$-signature $f$, we consider the union-split $f_{t,1}$, $f_{t,2}$, $f_{t,12}$ of $f$ at $t$, cf.~\cref{dfn:union_split}. For every label $\ell \in \splitlabels{t}{12} = \livelabels{t_1} \cap \livelabels{t_2}$, we need to consider all states $\stateset_1, \stateset_2 \in \states$ such that $\merge(\stateset_1, \stateset_2) = f(\ell)$, where $\stateset_i$ is the state of label $\ell$ at $t_i$. For two functions $g,h\colon B \rightarrow \states$, where $B$ is some set, we write $\merge(g,h) \colon B \rightarrow \states$ for the \emph{componentwise application} of $\merge$, i.e.\ $\merge(g,h)(\ell) = \merge(g(\ell), h(\ell))$ for all $\ell \in B$. Furthermore, we have to distribute $\ctarget$ and $\wtarget$ among the partial solutions at $t_1$ and the partial solutions at $t_2$, leading to the recurrence
\begin{equation*}
  \dppoly{t}{\ctarget, \wtarget}(f) = \sum_{\substack{\ctarget_1 + \ctarget_2 = \ctarget \\ \wtarget_1 + \wtarget_2 = \wtarget}} \sum_{\substack{g_1, g_2\colon \splitlabels{t}{12} \rightarrow \states \colon \\ \merge(g_1, g_2) = f}} \dppoly{t_1}{\ctarget_1, \wtarget_1}(g_1 \cup f_{t,1}) \dppoly{t_2}{\ctarget_2, \wtarget_2}(g_2 \cup f_{t,2}),
\end{equation*}
where we consider $g_i$ and $f_{t,i}$ as sets in $g_i \cup f_{t,i}$ for $i \in [2]$.
Here, we simply fix the state for the labels that are live at only one child by using the parts $f_{t,1}$ and $f_{t,2}$ and for the labels that are live at both children we sum over all valid state combinations.

 We now argue how this recurrence can be computed in time $\Oh^*(5^{|\livelabels{t}|})$ for fixed $\ctarget \in [0, \cfct(V)]$, $\wtarget \in [0, \wfct(V)]$, and for all $t$-signatures $f$. First, we branch on the numbers $\ctarget_1, \ctarget_2 \in [0, \cfct(V)]$, $\wtarget_1, \wtarget_2 \in [0, \wfct(V)]$, and functions $h_1 \colon \splitlabels{t}{1} \rightarrow \states$, $h_2 \colon \splitlabels{t}{2} \rightarrow \states$. Having fixed these choices, we can calculate the inner sum of the recurrence for all $t$-signatures $f$ with $f_{t,1} = h_1$ and $f_{t,2} = h_2$, by setting $B_{i}(g) := \dppoly{t_i}{\ctarget_i, \wtarget_i}(g \cup h_i)$, for $i \in [2]$ and $g\colon \splitlabels{t}{12} \rightarrow \states$, and computing the \emph{CDS-product}
\begin{equation*}
  (B_1 \otimes_{CDS} B_2)(g) := \sum_{\mathclap{\substack{g_1, g_2\colon \splitlabels{t}{12} \rightarrow \states \colon \\ \merge(g_1, g_2) = g }}} B_1(g_1) B_2(g_2)
\end{equation*}
for all $g\colon \splitlabels{t}{12} \rightarrow \states$. By the forthcoming \cref{thm:fast_cds_product}, we can compute the CDS-product in time $\Oh^*(5^{|\splitlabels{t}{12}|})$. Since there are $\Oh^*(5^{|\splitlabels{t}{1}| + |\splitlabels{t}{2}|})$ branches and $\livelabels{t}$ is partitioned into the three sets $\splitlabels{t}{1}$, $\splitlabels{t}{2}$, $\splitlabels{t}{12}$ by \cref{thm:partition_union_split}, we need time $\Oh^*(5^{|\livelabels{t}|})$ to compute the recurrence for all choices of $\ctarget$, $\wtarget$, and $f$.

\begin{lem}\label{thm:fast_cds_product}
  Given two tables $B_1, B_2 \colon \states^{\idxset} \rightarrow \ZZ_2$, where $I$ is some index set, their CDS-product $B_1 \otimes_{CDS} B_2 \colon \states^{\idxset} \rightarrow \ZZ_2$ can be computed in time $\Oh^*(|\states|^{|\idxset|}) = \Oh^*(5^{|\idxset|})$.
\end{lem}

\begin{proof}
  Consider the set family $\lattice = \{\emptyset, \{\sF\}, \{\sL\}, \{\sR\}, \{\sF, \sL, \sR\}\} \subseteq \powerset{\universe}$ over $\universe = \{\sF, \sL, \sR\}$ and the partial order on $\lattice$ induced by set inclusion $\subseteq$. It is easy to verify that this forms a \emph{lattice}, since every pair of elements has a greatest lower bound (\emph{meet}/$\wedge$) and a least upper bound (\emph{join}/$\vee$). In particular, the least upper bounds are given by:
  \begin{equation*}
    \begin{array}{l|ccccc}
      \vee & \emptyset & \{\sF\} & \{\sL\} & \{\sR\} & U \\
      \hline
      \emptyset & \emptyset & \{\sF\} & \{\sL\} & \{\sR\} & U \\
      \{\sF\} & \{\sF\} & \{\sF\} & U & U & U \\
      \{\sL\} & \{\sL\} & U & \{\sL\} & U & U \\
      \{\sR\} & \{\sR\} & U & U & \{\sR\} & U  \\
      U & U & U & U & U & U \\
    \end{array}
  \end{equation*}
  The bijection $\kappa\colon \states \rightarrow \lattice$, with $\kappa(\stateset) = \stateset$ for $\stateset \neq \sQ$ and $\kappa(\sQ) = U = \{\sF, \sL, \sR\}$, turns $\merge$ on $\states$ into $\vee$ on $\lattice$, i.e., $\kappa(\merge(\stateset_1, \stateset_2)) = \kappa(\stateset_1) \vee \kappa(\stateset_2)$ for all $\stateset_1, \stateset_2 \in \states$. Hence, we can write
  \begin{align*}
    (B_1 \otimes_{CDS} B_2)(g) & = \sum_{\substack{g_1, g_2\colon \idxset \rightarrow \states \colon \\ \merge(g_1, g_2) = g}} B_1(g_1) B_2(g_2) \\
    & = \sum_{\substack{h_1, h_2\colon \idxset \rightarrow \lattice \colon \\ h_1(i) \vee h_2(i) = \kappa(g(i)) \,\forall i \in \idxset}} B_1(\kappa^{-1} \circ h_1) B_2(\kappa ^{-1} \circ h_2) \\
    & = \sum_{\substack{h_1, h_2 \in \lattice^{\idxset} \colon \\ h_1 \vee h_2 = \kappa \circ g}} B_1(\kappa^{-1} \circ h_1) B_2(\kappa ^{-1} \circ h_2) = \sum_{\substack{h_1, h_2 \in \lattice^{\idxset} \colon \\ h_1 \vee h_2 = \kappa \circ g}} B'_1(h_1) B'_2(h_2) \\
    & = (B'_1 \otimes_{\lattice^{\idxset}} B'_2)(\kappa \circ g),
  \end{align*}
  where $B'_j(h) = B_j(\kappa^{-1} \circ h)$ for $h \in \lattice^{\idxset}$, $j = 1,2$, and $\vee$ is the join in $\lattice$ or $\lattice^{\idxset}$ depending on the context. By identifying $\lattice^{\idxset}$ with $\lattice^{|\idxset|}$ in the natural way, we can therefore apply \cref{thm:power_lattice_fast_product} to compute the CDS-product in time $\Oh^*(|\lattice|^{|\idxset|}) = \Oh^*(5^{|\idxset|})$ as the calls to $\algo_\lattice$ can be answered in constant time.
\end{proof}

\begin{lem}\label{thm:cds_count_part}
  Given a nice $\nlabels$-expression $\cexpr$ of a graph $G$, there is an algorithm that computes the quantities $\dppoly{t}{\ctarget,\wtarget}(f)$ for all nodes $t \in V(\tree_\cexpr)$, all $t$-signatures $f$, and all $\ctarget \in [0,\cfct(V)]$, $\wtarget \in [0,2n^2]$, in time $\Oh^*(5^{\nlabels})$.
\end{lem}

\begin{proof}  
  The algorithm proceeds by bottom-up dynamic programming along the augmented syntax tree $\augtree_\cexpr$ of the nice clique-expression $\cexpr$ and computes the quantities $\dppoly{t}{\ctarget,\wtarget}(f)$ via the given recurrences. For an introduce node, relabel node, join node, or dead node, the recurrence for $\dppoly{t}{\ctarget,\wtarget}(f)$ for fixed $f$, $\ctarget$, and $\wtarget$, can clearly be computed in polynomial time, since additions and multiplications in $\ZZ_2$ take constant time. For a union node $t$, we have argued how to compute the recurrences for all $f$, $\ctarget$, and $\wtarget$ simultaneously in time $\Oh^*(5^{|\livelabels{t}|})$. As $\cexpr$ is a $\nlabels$-expression, we have $|\livelabels{t}| \leq |\labels_t| \leq \nlabels$ for all $t \in V(\tree_\cexpr)$ and in particular at most $5^\nlabels$ $t$-signatures for any node $t \in V(\tree_\cexpr)$. Hence, the running time follows.
  
  It remains to prove the correctness of the recurrences. We have already proven the correctness of the recurrence for the dead nodes. For the other node types, the proofs are straightforward, we highlight some interesting parts. 
  
  For the join node, we highlight that the function $\feas \colon \states \times \states \rightarrow \{0,1\}$ satisfies $\feas(\stateset_1, \stateset_2) = 1 - [\stateset_1 \neq \emptyset][\stateset_2 \neq \emptyset][\stateset_1 = \stateset_2 = \sQ \vee \stateset_1 \neq \stateset_2]$ which characterizes the state pairs where adding a join between their underlying labels results in an edge between $X = X_L \cup X_R$ and $F$ or an edge across the cut $(X_L, X_R)$, hence yielding an infeasible solution.
  
  For relabel and union nodes, we highlight that for any two subsets $\stateset_1, \stateset_2 \subseteq \{\sF, \sL, \sR\}$ we have $\rho(\stateset_1 \cup \stateset_2) = \merge(\rho(\stateset_1), \rho(\stateset_2))$, where $\rho \colon \powerset{\{\sF, \sL, \sR\}} \rightarrow \states$ with $\rho(\stateset) = \stateset$ if $|\stateset| \leq 1$ and $\rho(\stateset) = \two_+$ if $|\stateset| \geq 2$. Hence, $\merge$ correctly updates the state for relabel and union nodes.
\end{proof}

\begin{thm}
  There is a randomized algorithm that given a nice $\nlabels$-expression $\cexpr$ for a graph $G = (V, E)$ can solve \CDS in time $\Oh^*(5^k)$. The algorithm does not return false positives and returns false negatives with probability at most $1/2$.
\end{thm}

\begin{proof}
  We begin by sampling a weight function $\wfct\colon V \rightarrow [2n]$ uniformly at random. Then, we pick an arbitrary vertex $v$ and branch on its closed neighborhood $N_G[v]$, since every dominating set intersects $N_G[v]$ in at least one vertex; the chosen vertex takes the role of $\fixedvertex$ in the current branch. We then run the algorithm of \cref{thm:cds_count_part} to compute the quantities $\dppoly{t}{\ctarget, \wtarget}(f)$.
  At the root, we have that $\livelabels{\rvertex} = \emptyset$. The algorithm returns true if there is some branch and choice of $\ctarget \leq \budget$, $\wtarget \in [0, 2n^2]$, such that $\dppoly{\rvertex}{\ctarget, \wtarget}(\emptyset) \neq 0$, otherwise the algorithm returns false.
  
  The running time directly follows from \cref{thm:cds_count_part}. For the correctness, first note that at the root, we have $\dpsols{\rvertex}{\ctarget, \wtarget}(\emptyset) = \dpsols{\rvertex}{\ctarget, \wtarget}$. Since all labels are dead at $\rvertex$, we have for any $(X_L, X_R, F) \in \dpsols{\rvertex}{}$ that $(X, (X_L, X_R)) \in \cuts(G)$, $\fixedvertex \in X_L$, $V \subseteq N_G[X]$ and hence $F = \emptyset$, where $X = X_L \cup X_R$. So, $X$ is a dominating set of $G$ that contains $\fixedvertex$ and vice versa $\dpsols{\rvertex}{}$ contains all consistent cuts of such dominating sets. The algorithm only returns true, if there are some $\ctarget \in [0, \budget]$, $\wtarget \in [0, 2n^2]$ such that $\dpsols{\rvertex}{\ctarget, \wtarget}$ has odd cardinality. Defining $\rsols = \{X \subseteq V \sep V \subseteq N_G[X], \fixedvertex \in X\}$ and $\rsols^{\ctarget, \wtarget} = \{X \in \rsols \sep |X| = \ctarget, \wfct(X) = \wtarget\}$, this implies that there exists a connected dominating set $X$ of $G$ with $\cfct(X) = \ctarget \leq \budget$ by \cref{thm:cons_cut_family}, hence the algorithm does not return false positives.
  
  For the error probability, suppose that the weight function $\wfct$ isolates an optimum connected dominating set $X^*$ and that $\cfct(X^*) \leq \budget$; by \cref{thm:isolation}, this happens with probability $\geq 1/2$. In a branch with $\fixedvertex \in X^*$, we have that $X^* \in \rsols^{\ctarget, \wtarget}$, where $\ctarget = |X^*|$ and $\wtarget = \wfct(X^*)$. By \cref{thm:cons_cut}, the connected dominating set contributes an odd number to $|\dpsols{\rvertex}{\ctarget, \wtarget}|$ and all other contributing sets $Y \in \rsols^{\ctarget, \wtarget}$ cannot be connected due to isolation and therefore contribute an even number to $|\dpsols{\rvertex}{\ctarget, \wtarget}|$. Therefore $\dppoly{\rvertex}{\ctarget, \wtarget}(\emptyset) \equiv_2 |\dpsols{\rvertex}{\ctarget, \wtarget}| = 1 \neq 0$ and hence the algorithm returns true with probability at least $1/2$ given a positive instance.
\end{proof}

\section{Lower Bounds}
\label{sec:cw_lb}

\subparagraph*{Construction Principle.} In this section, we prove the tight lower bounds for \CVC and \CDS parameterized by linear clique-width. The high-level construction principle follows the style of Lokshtanov et al.~\cite{LokshtanovMS18}. That means the resulting graphs can be interpreted as a \emph{matrix of blocks}, where each block spans several rows and columns. Every row is a long \emph{path-like gadget} that simulates a constant number of variables of the \SAT instance and which contributes 1 unit of linear clique-width. The number of simulated variables is tied to the running time that we want to rule out. For technical reasons, we consider bundles of rows simulating a \emph{variable group} of appropriate size. Every column corresponds to a clause and consists of gadgets that \emph{decode} the states on the path gadgets and check whether the resulting assignment satisfies the clause; considering only one clause per column expedites building a graph of low width.

\subparagraph*{Path Gadgets.} In both lower bounds, the main technical contribution is the design of the \emph{path gadgets}. Whereas the design of the \emph{decoding gadgets} can be adapted from known constructions. There is one path gadget at the intersection of each row and column. The goal is to construct a path gadget admitting as many states as possible for the target problem, since the number of such states corresponds to the base of the running time for which we obtain a lower bound. Since every row should contribute one unit of linear clique-width, we must connect adjacent path gadgets in a row with a join, thus restricting the design space.

\subparagraph*{State Transitions.} When proving that a solution to the target problem yields a satisfying assignment for \SAT, we require that the state remains \emph{stable} along each \emph{row}. Otherwise, if the states were allowed to change, one could pick a satisfying assignment for each clause \emph{separately}, which does not necessarily lead to an assignment satisfying \emph{all} clauses simultaneously. Unfortunately, it is often not possible to construct path gadgets where the state remains stable. However, we can construct path gadgets where the states can only transition in a controlled way.

\subparagraph*{Controlling the State Transitions.} Suppose that there is some \emph{transition order} on the states, i.e., state $i$ can transition to state $j$ only if $i \leq j$; this implies that the state can change only a finite number of times along each row. By making the rows long enough and repeating clauses, the pigeonhole principle allows us to find a region of columns spanning all clauses where all states remain stable, hence obtaining the same assignment for every clause. Therefore, our goal is to find a large set of states and an appropriate path gadget admitting such a transition order.

\subparagraph*{Determining the Transition Order.} For sparse width-parameter such as pathwidth, determining an appropriate transition order is much simpler, because the number of possible states is very limited, e.g., there are at most four vertex states for the considered benchmark problems. The possible set of states for clique-width is much larger and usually we need to select a specific subset of states, as not all of them admit a transition order. Since adjacent path gadgets are connected by a join, some state transitions are \emph{forced} due to the joins. Hence, we begin by analyzing the possible state transitions across a join and obtain a \emph{transition/compatibility} matrix showing which pairs of states can lead to a globally feasible solution and which cannot. After possibly reordering the rows and columns of the compatibility matrix, a transition order must induce a \emph{triangular submatrix} with ones on the diagonal, which represent that states can remain stable. Hence, a large triangular submatrix of the compatibility matrix serves as our starting point for designing an appropriate path gadget.

\subparagraph*{Anatomy of a Path Gadget.} Our path gadgets consist of three parts: a \emph{central clique} that communicates with the decoding gadgets, and two \emph{boundary} parts, i.e., the \emph{left} and \emph{right} part that connect to the previous and following join, respectively. In the central clique, each solution will avoid exactly one vertex representing the state of the path gadget. To implement the transition order, the left and right part have to communicate appropriate states to the two adjacent path gadgets. By taking the triangular submatrix and \emph{pairing} states along the main diagonal, where the pairs might be asymmetric due to the reordering of rows and columns, we see which states must be communicated in each case.

\subparagraph*{Designing the Boundary of a Path Gadget.} The state of a vertex in a partial solution consists of several properties: whether it is contained in the partial solution, a connectivity property, and possibly whether it is dominated (for \CDS). Our strategy for designing the boundary parts is to \emph{isolate} these properties for the vertices incident to the joins and represent by pairs of \emph{indicator vertices} whether each property is satisfied or not. We ensure that a solution can only pick one vertex per indicator pair, thus enabling simple communication between the central and boundary parts. This concludes the description of the high-level ideas for constructing the path gadgets.

\subparagraph*{Root-Connectivity.} To capture the connectivity constraint of \CVC and \CDS, we create a distinguished vertex $\rvertex$ called the \emph{root} and by attaching a vertex of degree 1 to $\rvertex$ we ensure that every connected vertex cover or connected dominating set has to contain $\rvertex$. Given a vertex subset $X \subseteq V(G)$ with $\rvertex \in X$, we say that a vertex $v \in X$ is \emph{root-connected} in $X$ if there is a $v,\rvertex$-path in $G[X]$. We will just say \emph{root-connected} if $X$ is clear from the context. The graph $G[X]$ is connected if and only if all vertices of $X$ are root-connected in $X$. For the state of a partial solution $X$, it is important to consider which vertices are root-connected in $X$ and which are not. 

\subsection{Connected Vertex Cover}
\label{sec:cw_cvc_lb}

This subsection is devoted to proving that \CVC (with unit costs) cannot be solved in time $\Oh^*((6-\eps)^{\lcw(G)})$ for some $\eps > 0$ unless the \CNFSETH fails. We first design the path gadget, approaching the design as presented in the outline, and analyze it in isolation and afterwards we present the complete construction and correctness proofs. The decoding gadgets are directly adapted from the lower bound for \CVC parameterized by pathwidth given by Cygan et al.~\cite{CyganNPPRW11}. 

\subsubsection{Path Gadget}

To rule out the running time $\Oh^*((6-\eps)^{\lcw(G)})$ for any $\eps > 0$, we have to build a path gadget that admits 6 distinct states and narrows down to a single label, so that each row of the construction contributes one unit of linear clique-width. As discussed previously, we begin by analyzing the possible behaviors of a partial solution on a label. 

Each single vertex $v$ has one of 3 states with respect to a partial solution $X$: $v \notin X$ (state $\zero$), $v \in X$ and $v$ is root-connected (state $\one_1$) or not (state $\one_0$). First, we observe that it is irrelevant how often each vertex state appears inside a label, rather we only care whether a vertex state appears in a label or not. Hence, we can describe the state of a label as a subset of $\{\zero, \one_0, \one_1\}$, where the empty subset is excluded, as we do not consider empty labels.

\begin{figure}
  \centering
  \input{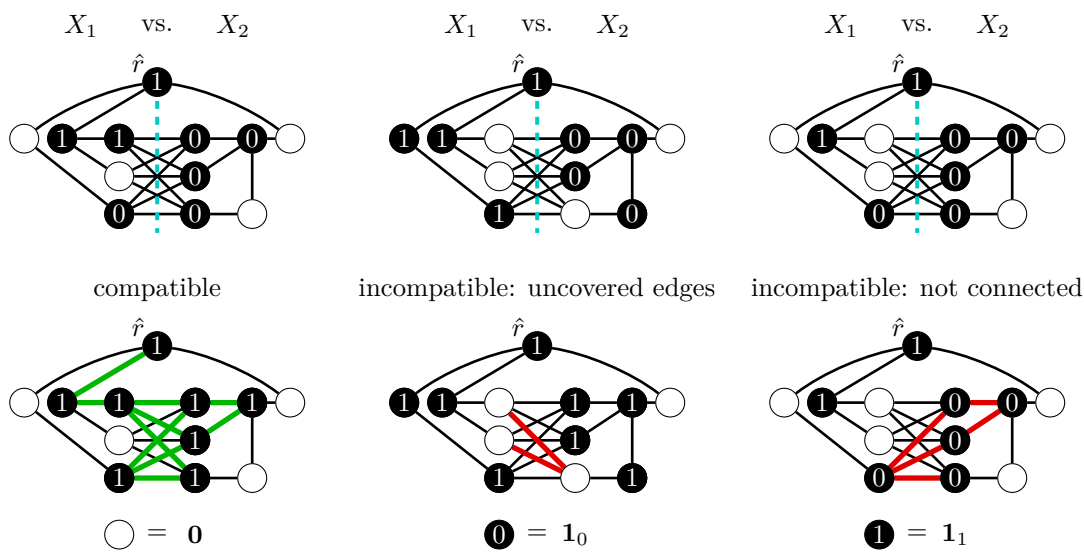}
  \caption{Several cases of partial solution compatibility across a join. The first row depicts the vertex states in $X_1$ and $X_2$, separated by the dashed line. The second row depicts the vertex states in $X_1 \cup X_2$ and highlights, from left to right, the induced edges, the uncovered edges, and a connected component not containing the root $\rvertex$.}
  \label{fig:cw_cvc_partial_comp}
\end{figure}

We proceed by studying the compatibility of theses label states across a join, but we will only give an informal description here. Essentially, we assume that the considered join is the final opportunity for two partial solutions $X_1, X_2 \subseteq V$ with $\rvertex \in X_i$, $i \in [2]$, living on separate sides of the join (with the exception of $\rvertex$) to connect. Hence, the partial solutions $X_1$ and $X_2$ are considered to be \emph{compatible} when in $X_1 \cup X_2$ every vertex incident to the considered join has state $\zero$ or $\one_1$ and every edge of the join is covered by $X_1 \cup X_2$; see \cref{fig:cw_cvc_partial_comp}. If $X_1$ and $X_2$ are compatible, then $X_1 \cup X_2$ \emph{can} be a \emph{global} solution, if outside of the considered join all constraints are satisfied. Since the interaction of $X_i$, $i \in [2]$ with the respective side of the join is captured by the aforementioned states, we obtain the compatibility matrix in \cref{table:comp_cvc}.

\begin{table}
\centering
\begin{tabular}{l|ccccccc}%
  $X_1$ vs.\ $X_2$ & $\{\zero\}$ & $\{\one_0\}$ & $\{\one_1\}$ & $\{\one_0, \zero\}$ & $\{\one_1, \zero\}$ & $\{\one_1, \one_0\}$ & $\{\one_1, \one_0, \zero\}$ \\%
  \hline%
  $\{\zero\}$ & 0 & 0 & 1 & 0 & 0 & 0 & 0 \\%
  $\{\one_0\}$ & 0 & 0 & 1 & 0 & 1 & 1 & 1 \\%
  $\{\one_1\}$ & 1 & 1 & 1 & 1 & 1 & 1 & 1 \\%
  $\{\one_0, \zero\}$ & 0 & 0 & 1 & 0 & 0 & 1 & 0 \\%
  $\{\one_1, \zero\}$ & 0 & 1 & 1 & 0 & 0 & 1 & 0 \\%
  $\{\one_1, \one_0\}$ & 0 & 1 & 1 & 1 & 1 & 1 & 1 \\%
  $\{\one_1, \one_0, \zero\}$ & 0 & 1 & 1 & 0 & 0 & 1 & 0%
\end{tabular}\caption{The compatibility matrix for \CVC. The rows describe the label state of the partial solution $X_1$ and the columns the label state of $X_2$. Ones correspond to compatible pairs of label states and zeroes to incompatible pairs.}\label{table:comp_cvc}
\vspace*{-0.8cm}
\end{table}

To determine a transition order and which states should be implemented by a path gadget, we find the triangular submatrix in \cref{table:triangular_cvc}, after reordering rows and columns.
\begin{table}
\centering
  \begin{tabular}{l|cccccc}%
    $X_1$ vs.\ $X_2$ & $\{\zero\}$ & $\{\one_0, \zero\}$ & $\{\one_0\}$ & $\{\one_1, \zero\}$ & $\{\one_1, \one_0\}$ & $\{\one_1\}$ \\%
    \hline%
    $\{\one_1\}$ & 1 & 1 & 1 & 1 & 1 & 1 \\%
    $\{\one_1, \one_0\}$ & 0 & 1 & 1 & 1 & 1 & 1 \\%
    $\{\one_1, \zero\}$ & 0 & 0 & 1 & 0 & 1 & 1 \\%
    $\{\one_0\}$ & 0 & 0 & 0 & 1 & 1 & 1  \\%
    $\{\one_0, \zero\}$ & 0 & 0 & 0 & 0 & 1 & 1 \\%
    $\{\zero\}$ & 0 & 0 & 0 & 0 & 0 & 1%
  \end{tabular}\caption{A large triangular submatrix in the compatibility matrix of \CVC. The rows and columns have been reordered.}\label{table:triangular_cvc}
  \vspace*{-0.8cm}
\end{table}
Note that the triangular submatrix only involves states consisting of at most two vertex states, hence labels consisting of two vertices should be sufficient to generate these states. Indeed, in the forthcoming construction, the labels incident to the join are independent sets of size two and the state sets will be represented by the following ordered pairs of vertex states: $(\zero, \zero)$, $(\one_0, \zero)$, $(\one_0, \one_0)$, $(\one_1, \zero)$, $(\one_1, \one_0)$, $(\one_1, \one_1)$. Pairing the states along the diagonal then tells us for each case which states the path gadget should communicate to the left and right boundary respectively. For example, for the third position of the diagonal, we pair $(\one_0, \one_0)$ with $(\one_1, \zero)$, meaning that for the third state of the transition order, the path gadget should communicate the states $(\one_0, \one_0)$ to the left boundary and the states $(\one_1, \zero)$ to the right boundary.

\subparagraph*{Formal Definition of States.} We define the three atomic states $\atoms = \{\zero, \one_0, \one_1\}$ and define the two predicates $\sol, \conn \colon \atoms \rightarrow \{0,1\}$ by $\sol(\bolda) = [\bolda \in \{\one_0, \one_1\}]$ and $\conn(\bolda) = [\bolda = \one_1]$. The atom $\zero$ means that a vertex is not inside the partial solution; $\one_1$ and $\one_0$ indicate that a vertex is inside the partial solution and the subscript indicates whether it is root-connected or not. Building on these atomic states, we define six (gadget) states consisting of four atomic states each: 
\newcommand{\zerophantom}{{\zero_{\phantom{0}}}}
\begin{align*}
  \state^1 & = (\zerophantom, \zerophantom, \one_1, \one_1), \\
  \state^2 & = (\one_0, \zerophantom, \one_1, \one_0), \\
  \state^3 & = (\one_0, \one_0, \one_1, \zerophantom), \\
  \state^4 & = (\one_1, \zerophantom, \one_0, \one_0), \\
  \state^5 & = (\one_1, \one_0, \one_0, \zerophantom), \\
  \state^6 & = (\one_1, \one_1, \zerophantom, \zerophantom).
\end{align*}
The gadget states are numbered in the transition order. We collect the six gadget states in the set $\states = \{\state^1, \ldots, \state^6\}$ and use the notation $\state^\ell_i \in \atoms$, $i \in [4]$, $\ell \in [6]$, to refer to the $i$-th atomic component of state $\state^\ell$. Observe that $\state^\ell$ can be obtained from $\state^{7-\ell}$ by swapping the first and second component with the third and fourth component, i.e., $\state_1^\ell = \state_3^{7-\ell}$ and $\state_2^\ell = \state_4^{7-\ell}$ for all $\ell \in [6]$.

Given a partial solution $Y \subseteq V(G)$, we formally associate to each vertex its state in $Y$ with the map $\statemap_Y \colon V(G) \setminus \{\rvertex\} \rightarrow \atoms$, which is defined by
\begin{equation*}
  \statemap_Y(v) = \begin{cases}
    \zero & \text{if } v \notin Y, \\
    \one_0 & \text{if } v \in Y \text{ and $v$ is not root-connected in $Y \cup \{\rvertex\}$}, \\
    \one_1 & \text{if } v \in Y \text{ and $v$ is root-connected in $Y \cup \{\rvertex\}$}. \\
  \end{cases}
\end{equation*}

\subparagraph*{Formal Construction.} We proceed by describing how to construct the path gadget $P$. We create 4 \emph{join} vertices $u_1, \ldots, u_4$, 12 \emph{auxiliary} vertices $a_{1,1}, a_{1,2}, a_{1,3}, a_{2,1}, \ldots, a_{4,3}$, 8 \emph{solution indicator} vertices $b_{1,0}, b_{1,1}, b_{2,0}, b_{2,1}, \ldots, b_{4,1}$, 8 \emph{connectivity indicator} vertices $c_{1,0}, c_{1,1}, c_{2,0}, c_{2,1}, \ldots, c_{4,1}$ and 6 \emph{clique} vertices $v_1, \ldots, v_6$. For every $i \in [4]$, we add the edges $\{u_i, a_{i,1}\}$, $\{u_i, a_{i,3}\}$, $\{u_i, b_{i,0}\}$, $\{a_{i,1},a_{i,2}\}$, $\{a_{i,1}, b_{i,0}\}$, $\{a_{i,1}, c_{i,1}\}$, $\{a_{i,3}, b_{i,1}\}$, and $\{b_{i,0}, b_{i,1}\}$ and $\{c_{i,0}, c_{i,1}\}$. Furthermore, we make $a_{i,3}$ for all $i \in [4]$, all solution indicator vertices, all connectivity indicator vertices, and all clique vertices adjacent to the root vertex $\rvertex$. We add all possible edges between the clique vertices $v_\ell$, $\ell \in [6]$, so that they induce a clique of size 6. 

Finally, we explain how to connect the indicator vertices to the clique vertices. The clique vertex $v_\ell$ corresponds to choosing state $\state^\ell$ on the join vertices $(u_1, u_2, u_3, u_4)$. The desired behavior of $P$ is that a partial solution $X$ of $P + \rvertex$ contains $b_{i,1}$ if and only if $X$ contains $u_i$ and for the connectivity indicators, that $X$ contains $c_{i,1}$ if and only if $X$ contains $u_i$ and $u_i$ is root-connected in $X$. Accordingly, for all $i \in [4]$ and $\ell \in [6]$, we add the edges $\{v_\ell, b_{i,\sol(\state^\ell_i)}\}$ and $\{v_\ell, c_{i,\conn(\state^\ell_i)}\}$. This concludes the construction of $P$, see \cref{fig:cw_cvc_path}.

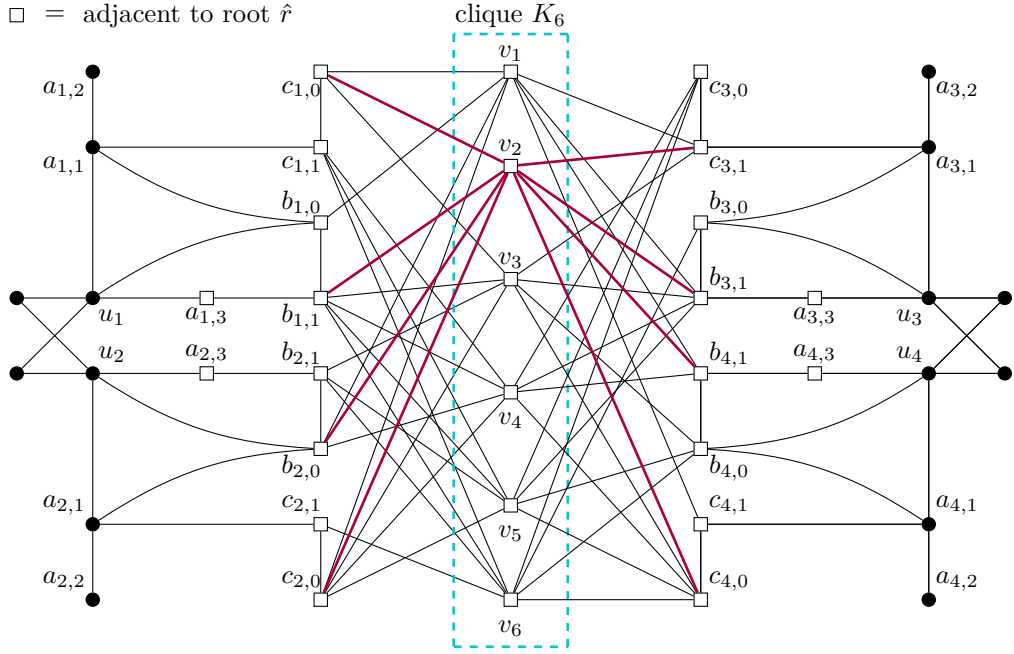
\begin{figure}
  \centering
  \begin{tikzpicture}
	\begin{pgfonlayer}{nodelayer}
		\node [style=filled] (1) at (2, 1) {};
		\node [style=filled] (3) at (2, -1) {};
		\node [style=tiny rectangle] (4) at (5, 1) {};
		\node [style=tiny rectangle] (5) at (8, 1) {};
		\node [style=tiny rectangle] (6) at (8, 3) {};
		\node [style=filled] (7) at (2, -5) {};
		\node [style=filled] (8) at (2, -7) {};
		\node [style=tiny rectangle] (9) at (8, -5) {};
		\node [style=tiny rectangle] (10) at (8, -7) {};
		\node [style=tiny rectangle] (11) at (5, -1) {};
		\node [style=tiny rectangle] (12) at (8, -1) {};
		\node [style=tiny rectangle] (13) at (8, -3) {};
		\node [style=tiny rectangle] (15) at (8, 5) {};
		\node [style=tiny rectangle] (16) at (8, 7) {};
		\node [style=filled] (17) at (2, 5) {};
		\node [style=filled] (18) at (2, 7) {};
		\node [style=none] (51) at (11.5, 8) {};
		\node [style=none] (52) at (14.5, 8) {};
		\node [style=none] (53) at (11.5, -8.25) {};
		\node [style=none] (54) at (14.5, -8.25) {};
		\node [style=tiny rectangle] (82) at (13, 7) {};
		\node [style=tiny rectangle] (83) at (13, 4.5) {};
		\node [style=tiny rectangle] (84) at (13, 1.5) {};
		\node [style=tiny rectangle] (85) at (13, -1.5) {};
		\node [style=tiny rectangle] (86) at (13, -4.5) {};
		\node [style=tiny rectangle] (87) at (13, -7) {};
		\node [style=filled] (88) at (0, 1) {};
		\node [style=filled] (89) at (0, -1) {};
		\node [style=none] (92) at (2.5, 0.5) {$u_1$};
		\node [style=none] (93) at (2.5, -0.5) {$u_2$};
		\node [style=none] (94) at (5, 0.5) {$a_{1,3}$};
		\node [style=none] (95) at (7.5, 0.5) {$b_{1,1}$};
		\node [style=none] (96) at (5, -0.5) {$a_{2,3}$};
		\node [style=none] (97) at (7.5, -0.5) {$b_{2,1}$};
		\node [style=none] (98) at (7.5, -3.5) {$b_{2,0}$};
		\node [style=none] (99) at (1.25, -4.5) {$a_{2,1}$};
		\node [style=none] (100) at (1.25, -6.5) {$a_{2,2}$};
		\node [style=none] (101) at (7.5, -4.5) {$c_{2,1}$};
		\node [style=none] (102) at (7.5, -6.5) {$c_{2,0}$};
		\node [style=none] (103) at (1.25, 4.5) {$a_{1,1}$};
		\node [style=none] (104) at (1.25, 6.5) {$a_{1,2}$};
		\node [style=none] (105) at (7.5, 3.5) {$b_{1,0}$};
		\node [style=none] (106) at (7.5, 4.5) {$c_{1,1}$};
		\node [style=none] (107) at (7.5, 6.5) {$c_{1,0}$};
		\node [style=none] (108) at (13, 7.5) {$v_1$};
		\node [style=none] (109) at (13, 5) {$v_2$};
		\node [style=none] (110) at (13, 2) {$v_3$};
		\node [style=none] (111) at (13, -2.25) {$v_4$};
		\node [style=none] (112) at (13, -5.25) {$v_5$};
		\node [style=none] (113) at (13, -7.75) {$v_6$};
		\node [style=none] (114) at (13, 8.5) {clique $K_6$};
		\node [style=filled] (115) at (24, 1) {};
		\node [style=filled] (116) at (24, -1) {};
		\node [style=tiny rectangle] (117) at (21, 1) {};
		\node [style=tiny rectangle] (118) at (18, 1) {};
		\node [style=tiny rectangle] (119) at (18, 3) {};
		\node [style=filled] (120) at (24, -5) {};
		\node [style=filled] (121) at (24, -7) {};
		\node [style=tiny rectangle] (122) at (18, -5) {};
		\node [style=tiny rectangle] (123) at (18, -7) {};
		\node [style=tiny rectangle] (124) at (21, -1) {};
		\node [style=tiny rectangle] (125) at (18, -1) {};
		\node [style=tiny rectangle] (126) at (18, -3) {};
		\node [style=tiny rectangle] (127) at (18, 5) {};
		\node [style=tiny rectangle] (128) at (18, 7) {};
		\node [style=filled] (129) at (24, 5) {};
		\node [style=filled] (130) at (24, 7) {};
		\node [style=filled] (131) at (26, 1) {};
		\node [style=filled] (132) at (26, -1) {};
		\node [style=none] (133) at (23.5, 0.5) {$u_3$};
		\node [style=none] (134) at (23.5, -0.5) {$u_4$};
		\node [style=none] (135) at (21, 0.5) {$a_{3,3}$};
		\node [style=none] (136) at (18.75, 1.5) {$b_{3,1}$};
		\node [style=none] (137) at (21, -0.5) {$a_{4,3}$};
		\node [style=none] (138) at (18.75, -0.5) {$b_{4,1}$};
		\node [style=none] (139) at (18.75, -3.5) {$b_{4,0}$};
		\node [style=none] (140) at (24.75, -4.5) {$a_{4,1}$};
		\node [style=none] (141) at (24.75, -6.5) {$a_{4,2}$};
		\node [style=none] (142) at (18.75, -4.5) {$c_{4,1}$};
		\node [style=none] (143) at (18.75, -6.5) {$c_{4,0}$};
		\node [style=none] (144) at (24.75, 4.5) {$a_{3,1}$};
		\node [style=none] (145) at (24.75, 6.5) {$a_{3,2}$};
		\node [style=none] (146) at (18.75, 3.5) {$b_{3,0}$};
		\node [style=none] (147) at (18.75, 4.5) {$c_{3,1}$};
		\node [style=none] (148) at (18.75, 6.5) {$c_{3,0}$};
		\node [style=tiny rectangle] (149) at (0, 8.5) {};
		\node [style=none] (150) at (1, 8.5) {=};
		\node [style=none] (151) at (4.5, 8.5) {adjacent to root $\hat{r}$};
	\end{pgfonlayer}
	\begin{pgfonlayer}{edgelayer}
		\draw (7) to (9);
		\draw [bend left=15] (7) to (13);
		\draw (3) to (7);
		\draw (1) to (4);
		\draw (4) to (5);
		\draw [bend left=15] (1) to (6);
		\draw (3) to (11);
		\draw (11) to (12);
		\draw (12) to (13);
		\draw [bend left=15] (13) to (3);
		\draw (7) to (8);
		\draw (9) to (10);
		\draw (17) to (1);
		\draw (17) to (15);
		\draw [bend right=15] (17) to (6);
		\draw (16) to (15);
		\draw (6) to (5);
		\draw (18) to (17);
		\draw (82) to (16);
		\draw (82) to (6);
		\draw (82) to (13);
		\draw (82) to (10);
		\draw (84) to (5);
		\draw (84) to (16);
		\draw (84) to (12);
		\draw (84) to (10);
		\draw (85) to (5);
		\draw (15) to (85);
		\draw (85) to (13);
		\draw (85) to (10);
		\draw (86) to (15);
		\draw (5) to (86);
		\draw (86) to (10);
		\draw (12) to (86);
		\draw (87) to (9);
		\draw (12) to (87);
		\draw (87) to (15);
		\draw (5) to (87);
		\draw [style=dashed cyan thick] (51.center) to (53.center);
		\draw [style=dashed cyan thick] (53.center) to (54.center);
		\draw [style=dashed cyan thick] (54.center) to (52.center);
		\draw [style=dashed cyan thick] (52.center) to (51.center);
		\draw (88) to (1);
		\draw (1) to (89);
		\draw (89) to (3);
		\draw (3) to (88);
		\draw [style=packing] (83) to (5);
		\draw [style=packing] (16) to (83);
		\draw [style=packing] (83) to (13);
		\draw [style=packing] (83) to (10);
		\draw (120) to (122);
		\draw (116) to (120);
		\draw (115) to (117);
		\draw (117) to (118);
		\draw (116) to (124);
		\draw (124) to (125);
		\draw (125) to (126);
		\draw (120) to (121);
		\draw (122) to (123);
		\draw (129) to (115);
		\draw (129) to (127);
		\draw (128) to (127);
		\draw (119) to (118);
		\draw (130) to (129);
		\draw (131) to (115);
		\draw (115) to (132);
		\draw (132) to (116);
		\draw (116) to (131);
		\draw (120) to (122);
		\draw [bend right=15] (120) to (126);
		\draw (117) to (118);
		\draw (124) to (125);
		\draw (125) to (126);
		\draw (120) to (121);
		\draw (122) to (123);
		\draw (129) to (127);
		\draw [bend left=15] (129) to (119);
		\draw (128) to (127);
		\draw (119) to (118);
		\draw (130) to (129);
		\draw (117) to (115);
		\draw [bend right=15] (115) to (119);
		\draw (129) to (115);
		\draw (116) to (120);
		\draw (124) to (116);
		\draw [bend left=15] (116) to (126);
		\draw (82) to (127);
		\draw (82) to (118);
		\draw (82) to (125);
		\draw (82) to (122);
		\draw (84) to (127);
		\draw (84) to (118);
		\draw (84) to (126);
		\draw (84) to (123);
		\draw (85) to (118);
		\draw (85) to (128);
		\draw (85) to (125);
		\draw (85) to (123);
		\draw (86) to (123);
		\draw (86) to (126);
		\draw (86) to (118);
		\draw (86) to (128);
		\draw (87) to (123);
		\draw (87) to (126);
		\draw (87) to (119);
		\draw (87) to (128);
		\draw (115) to (131);
		\draw (131) to (116);
		\draw (116) to (132);
		\draw (132) to (115);
		\draw [style=packing] (83) to (127);
		\draw [style=packing] (83) to (118);
		\draw [style=packing] (83) to (125);
		\draw [style=packing] (83) to (123);
	\end{pgfonlayer}
\end{tikzpicture}
  \caption{The path gadget $P$ with the join vertices $u_1, u_2$ and $u_3, u_4$ already joined to further vertices. All vertices that are depicted with a rectangle are adjacent to the root vertex $\rvertex$. The vertices inside the cyan dashed rectangle induce a clique. For sake of understanding, we have highlighted the edges incident to $v_2$ as one example.}
  \label{fig:cw_cvc_path}
\end{figure}

\subparagraph*{Behavior of a Single Path Gadget.} For the upcoming lemmas, we assume that $G$ is a graph that contains $P + \rvertex$ as an induced subgraph and that only the join vertices $u_i, i \in [4]$, and clique vertices $v_\ell, \ell \in [6]$, have neighbors outside this copy of $P + \rvertex$. Furthermore, let $X$ be a connected vertex cover of $G$ with $\rvertex \in X$. We study the behavior of such connected vertex covers on $P$; we will abuse notation and write $X \cap P$ instead of $X \cap V(P)$. The assumption on how $P$ connects to the remaining graph implies that any vertex $v \in V(P)$ with $\statemap_{X \cap P}(v) = \one_0$ has to be root-connected in $X$ through some path that leaves $P + \rvertex$ via one of the join vertices $u_i$, $i \in [4]$. Note that any path leaving $P + \rvertex$ through some clique vertex $v_\ell$, $\ell \in [6]$, immediately yields a path to $\rvertex$ in $P + \rvertex$ as $\{v_\ell \sep \ell \in [6]\} \subseteq N(\rvertex)$.

We begin by showing a lower bound for $|X \cap P|$ via a vertex-disjoint packing of subgraphs.

\begin{lem}\label{thm:cvc_cw_path_gadget_lb}
  We have that $|X \cap P| \geq 21 = 4 \cdot 4 + 5$ and more specifically $a_{i,1} \in X$, $|X \cap \{u_i, a_{i,3}, b_{i,1}, b_{i,0}\}| \geq 2$, $|X \cap \{c_{i,0}, c_{i,1}\}| \geq 1$ for all $i \in [4]$ and $|X \cap \{v_1, \ldots, v_6\}| \geq 5$. 
\end{lem}

\begin{proof}
  For all $i \in [4]$, the vertex $a_{i,2}$ has degree 1 in $G$ with unique neighbor $a_{i,1}$, hence we must have $a_{i,1} \in X$ to either connect $a_{i,2}$ to $\rvertex$ or to cover the edge $\{a_{i,1}, a_{i,2}\}$ if $a_{i,2} \notin X$. The vertices $u_i$, $a_{i,3}$, $b_{i,1}$, $b_{i,0}$ induce a cycle of length 4 for all $i \in [4]$ and any vertex cover has to contain at least 2 vertices of every such cycle. The edge $\{c_{i,0}, c_{i,1}\}$ has to be covered for all $i \in [4]$. Finally, the clique formed by the vertices $v_1, \ldots, v_6$ has size 6 and any vertex cover has to contain at least 5 vertices of such a clique. Since we only considered pairwise disjoint sets of vertices, these lower bounds simply add up and we obtain $|X \cap P| \geq 4 (1 + 2 + 1) + 5 = 21$. 
\end{proof}

Using \cref{thm:cvc_cw_path_gadget_lb}, we can precisely analyze the solutions that match the lower bound of 21 on $P$ and show that such solutions have the desired state behavior. We define for any $Y \subseteq V(P)$ the 4-tuple $\statemap(Y) = (\statemap_Y(u_1), \statemap_Y(u_2), \statemap_Y(u_3), \statemap_Y(u_4))$ and the following lemma shows that the states communicated to the boundary depend on the state of the central clique as desired.
 
 \begin{lem}\label{thm:cvc_cw_path_gadget_tight}
   If $|X \cap P| \leq 21$, then $|X \cap P| = 21$ and $a_{i,1} \in X$, $X \cap \{u_i, a_{i,3}, b_{i,1}, b_{i,0}\} \in \{\{u_i, b_{i,1}\}, \{a_{i,3}, b_{i,0}\}\}$, $|X \cap \{c_{i,0}, c_{i,1}\}| = 1$ for all $i \in [4]$ and $|X \cap \{v_1, \ldots, v_6\}| = 5$. Furthermore, we have $\statemap(X \cap P) = \state^\ell$ for the unique integer $\ell \in [6]$ with $v_\ell \notin X$.
 \end{lem}

\begin{proof}
  Due to $|X \cap P| \leq 21$, all the inequalities of \cref{thm:cvc_cw_path_gadget_lb} have to be tight. Since $|X \cap \{u_i, a_{i,3}, b_{i,1}, b_{i,0}\}| = 2$ and $\{u_i, a_{i,3}, b_{i,1}, b_{i,0}\}$ induces a cycle of length 4, $X$ has to contain a pair of antipodal vertices of this cycle, so either the pair $\{u_i, b_{i,1}\}$ or the pair $\{a_{i,3}, b_{i,0}\}$. It remains to prove the last property regarding the states of the join vertices.
  
  Consider any $i \in [4]$, we will show that $\statemap_{X \cap P}(u_i) = \state^\ell_i$. Since $X$ is a vertex cover and $v_\ell \notin X$, we must have that $N(v_\ell) \subseteq X$. By construction, this means that $X$ must in particular contain the vertices $b_{i,\sol(\state^\ell_i)}$ and $c_{i,\conn(\state^\ell_i)}$ and due to the previous equations $X$ cannot contain the other vertex of each of these indicator pairs. Due to the choice of an antipodal pair we have $u_i \in X$ if and only if $b_{i,1} \in X$ if and only if $\sol(\state^\ell_i) = 1$ if and only if $\state^\ell_i \in \{\one_0, \one_1\}$. If $u_i \in X$, then $a_{i,3} \notin X$ and $b_{i,0} \notin X$, so $u_i$ is root-connected in $X \cap (P + \rvertex)$ if and only if $c_{i,1} \in X$, i.e., via the path $u_i, a_{i,1}, c_{i,1}, \rvertex$. Furthermore, $c_{i,1} \in X \Leftrightarrow \conn(\state^\ell_i) = 1 \Leftrightarrow \state^\ell_i = \one_1$. This shows that $\statemap_{X \cap P}(u_i) = \state^\ell_i$ and since this argument applies to all $i \in [4]$, we obtain that $\statemap(X \cap Y) = \state^\ell$.
\end{proof}

Moving on, we will establish that for every state $\state^\ell \in \states$ a partial solution attaining this state actually exists. Furthermore, for these partial solutions it is sufficient to check for root-connectivity at the join vertices.

\begin{lem}\label{thm:cvc_cw_state_exists}  
  For every $\ell \in [6]$, there exists a vertex cover $X_P^\ell$ of $P$ such that $|X_P^\ell| = 21$, $X_P^\ell \cap \{v_1, \ldots, v_6\} = \{v_1, \ldots, v_6 \} \setminus \{v_\ell\}$, and $\statemap(X_P^\ell) = \state^\ell$. If $X$ is a vertex cover of $G$ with $\rvertex \in X$ and $X \cap P = X_P^\ell$ and for every $i \in [4]$ either $u_i \notin X$ or $u_i$ is root-connected in $X$, then every vertex of $X_P^\ell$ is root-connected in $X$.
\end{lem}

\begin{proof}
  We define 
  \begin{align*}
    X^\ell_P & = (\{v_1, \ldots, v_6\} \setminus \{v_\ell\}) \cup \{a_{i,1} \sep i \in [4]\} \\
    & \cup \{u_i \sep i \in [4] \text{ and } \sol(\state^\ell_i) = 1\} \cup \{a_{i,3} \sep i \in [4] \text{ and } \sol(\state^\ell_i) = 0\} \\
    & \cup \{b_{i, \sol(\state^\ell_i)} \sep i \in [4]\} \cup \{c_{i, \conn(\state^\ell_i)} \sep i \in [4]\}
  \end{align*} 
  and claim that $X^\ell_P$ is a vertex cover with the desired properties. We clearly have that $|X^\ell_P| = 21$ and $v_\ell \notin X^\ell_P$. We proceed by showing that $X^\ell_P$ is a vertex cover of $P$. For every $i \in [4]$, all four edges incident to $a_{i,1}$ are covered and $X^\ell_P$ contains at least one vertex from the edge $\{c_{i,0}, c_{i,1}\}$. For every $i \in [4]$, all edges of the $C_4$ induced by $\{u_i, a_{i,3}, b_{i,0}, b_{i,1}\}$ are covered by $X^\ell_P$, since $X^\ell_P$ picks one of the two antipodal pairs depending on $\sol(\state^\ell_i)$. The clique induced by $v_1, \ldots, v_6$ is fully covered by $X^\ell_P$, since it picks five out of six vertices. The edges between $v_\ell$ and the indicator vertices are covered, because by construction of $P$, we have that $N(v_\ell) = (\{v_1, \ldots, v_6\} \setminus \{v_\ell\}) \cup \{b_{i, \sol(\state^\ell_i)} \sep i \in [4]\} \cup \{c_{i, \conn(\state^\ell_i)} \sep i \in [4]\} \subseteq X^\ell_P$. This shows that $X^\ell_P$ is a vertex cover of $P$. Very similar to the proof of \cref{thm:cvc_cw_path_gadget_tight}, we see that $\statemap(X^\ell_P) = \state^\ell$.
  
  It remains to show the property regarding connectivity. By assumption, only the join vertices $u_1, \ldots, u_4$ and clique vertices $v_1, \ldots, v_6$ can be adjacent to vertices outside of $P + \rvertex$. However, as the clique vertices are neighbors of the root $\rvertex$, their other connections to the outside of $P$ cannot provide any new root-connectivity. If we have $\statemap_{X_P^\ell}(u_i) = \one_0$ for some $i \in [4]$, then $u_i$ is root-connected in $X$ via some path that leaves $P + \rvertex$. Note that all vertices in $X_P^\ell \setminus \{u_i, a_{i,1} \sep i \in [4]\}$ are directly adjacent to the root $\rvertex$, hence it just remains to handle the root-connectivity of $a_{i,1}$. If $\statemap_{X_P^\ell}(u_i) = \zero$, then $a_{i,1}$ is root-connected via the path $a_{i,1}, b_{i,0}, \rvertex$ in $X$. If $\statemap_{X_P^\ell}(u_i) = \one_0$, then we can extend the path that leaves $P$ and connects $u_i$ to $\rvertex$ by $a_{i,1}$. Finally, if $\statemap_{X_P^\ell}(u_i) = \one_1$, then the path $a_{i,1}, c_{i,1}, \rvertex$ exists in $X_P^\ell + \rvertex$. This concludes the proof.
\end{proof}

\subparagraph*{State Transitions.} In the lower bound construction, we will create long paths by repeatedly concatenating the path gadgets $P$. To study how the state can change between two consecutive path gadgets, suppose that we have two copies $P^1$ and $P^2$ of $P$ such that the vertices $u_3$ and $u_4$ in $P^1$ are joined to the vertices $u_1$ and $u_2$ in $P^2$. We denote the vertices of $P^1$ with a superscript $1$ and the vertices of $P^2$ with a superscript $2$, e.g., $u^1_3$ refers to the vertex $u_3$ of $P^1$. Again, suppose that $P^1$ and $P^2$ are embedded as induced subgraphs in a larger graph $G$ with a root vertex $\rvertex$ and that only the vertices $u^1_1, u^1_2, u^2_3, u^2_4$ and the clique vertices $v^1_\ell, v^2_\ell$, $\ell \in [6]$, have neighbors outside of $P^1 + P^2 + \rvertex$. Furthermore, $X$ denotes a connected vertex cover with $\rvertex \in X$. 

Using the previous lemmas, we now show that states can only transition from one path gadget to the next according to the transition order and that it is also feasible for the state to remain stable.

\begin{lem}\label{thm:cvc_cw_path_transition}
  Suppose that $|X \cap P^1| \leq 21$ and $|X \cap P^2| \leq 21$, then $\statemap(X \cap P^1) = \state^{\ell_1}$ and $\statemap(X \cap P^2) = \state^{\ell_2}$ with $\ell_1 \leq \ell_2$. 
  
  Additionally, for each $\ell \in [6]$, the set $X^\ell = X^\ell_{P^1} \cup X^\ell_{P^2}$ is a vertex cover of $P^1 + P^2$ with $\statemap_{X^\ell}(\{u^1_3, u^1_4, u^2_1, u^2_2\}) \subseteq \{\zero, \one_1\}$.
\end{lem}

\begin{proof}
  By \cref{thm:cvc_cw_path_gadget_tight}, we have that $\statemap(X \cap P^1) = \state^{\ell_1}$ and $\statemap(X \cap P^2) = \state^{\ell_2}$ for some $\ell_1 \in [6]$ and $\ell_2 \in [6]$. It remains to show that $\ell_1 \leq \ell_2$. 
  
  Define $U^1 = \{u^1_3, u^1_4\}$, $U^2 = \{u^2_1, u^2_2\}$, and $U = U^1 \cup U^2$.
  By assumption only the clique vertices of $P^1$ and $P^2$ and the vertices $u^1_1, u^1_2, u^2_3, u^2_4$ are allowed to have neighbors outside of $P^1 + P^2 + \rvertex$, hence $\{v^i_\ell \sep i \in [2], \ell \in [6]\} \subseteq N(\rvertex)$ separates the vertices in $U$ from the root $\rvertex$ in the whole graph $G$. Hence, we can see whether the vertices of $X \cap U$ are root-connected in $X$ by just considering the graph $P^1 + P^2 + \rvertex$.

  \newcommand{\statepair}{\bar{\state}}  
  \begin{figure}[h]
    \centering
    \begin{tikzpicture}
	\begin{pgfonlayer}{nodelayer}
		\node [style=filled W] (1) at (13.5, 5.5) {$\phantom{0}$};
		\node [style=filled 10pt] (2) at (13.5, 4.5) {0};
		\node [style=filled 10pt] (3) at (13.5, 3.5) {1};
		\node [style=none] (5) at (14.5, 5.5) {=};
		\node [style=none] (6) at (14.5, 4.5) {=};
		\node [style=none] (7) at (14.5, 3.5) {=};
		\node [style=none] (9) at (15.5, 5.5) {$\mathbf{0}$};
		\node [style=none] (10) at (15.5, 4.5) {$\mathbf{1}_0$};
		\node [style=none] (11) at (15.5, 3.5) {$\mathbf{1}_1$};
		\node [style=none] (38) at (2.75, 6.75) {};
		\node [style=none] (39) at (2.75, -7.25) {};
		\node [style=none] (40) at (0.5, 5) {};
		\node [style=none] (41) at (12.25, 5) {};
		\node [style=none] (42) at (9.75, 3.75) {$\times$};
		\node [style=none] (43) at (8.25, 3.75) {$\times$};
		\node [style=none] (44) at (6.75, 3.75) {$\times$};
		\node [style=none] (45) at (5.25, 3.75) {$\times$};
		\node [style=none] (46) at (3.75, 3.75) {$\times$};
		\node [style=none] (47) at (9.75, -0.25) {$\times$};
		\node [style=none] (49) at (6.75, -0.25) {$\times$};
		\node [style=none] (50) at (9.75, -2.25) {$\times$};
		\node [style=none] (51) at (9.75, 1.75) {$\times$};
		\node [style=none] (52) at (9.75, -4.25) {$\times$};
		\node [style=none] (53) at (5.25, 1.75) {$\times$};
		\node [style=none] (54) at (8.25, -2.25) {$\times$};
		\node [style=none] (55) at (2.75, 5) {};
		\node [style=none] (56) at (1, 6.5) {};
		\node [style=none] (57) at (1.25, 5.5) {$\overline{s}^1$};
		\node [style=none] (61) at (2.25, 6.25) {$\overline{s}^2$};
		\node [style=filled W] (62) at (2, -6.5) {$\phantom{0}$};
		\node [style=filled W] (63) at (2, -5.75) {$\phantom{0}$};
		\node [style=filled W] (64) at (2, -4.5) {$\phantom{0}$};
		\node [style=filled W] (65) at (2, -0.5) {$\phantom{0}$};
		\node [style=filled 10pt] (66) at (2, -3.75) {0};
		\node [style=filled 10pt] (67) at (2, -2.5) {0};
		\node [style=filled 10pt] (68) at (2, -1.75) {0};
		\node [style=filled 10pt] (69) at (2, 1.5) {0};
		\node [style=filled 10pt] (70) at (2, 0.25) {1};
		\node [style=filled 10pt] (71) at (2, 2.25) {1};
		\node [style=filled 10pt] (72) at (2, 3.5) {1};
		\node [style=filled 10pt] (73) at (2, 4.25) {1};
		\node [style=filled W] (74) at (3.75, 6.5) {$\phantom{0}$};
		\node [style=filled W] (75) at (3.75, 5.75) {$\phantom{0}$};
		\node [style=filled W] (76) at (5.25, 5.75) {$\phantom{0}$};
		\node [style=filled 10pt] (77) at (5.25, 6.5) {0};
		\node [style=filled 10pt] (78) at (6.75, 6.5) {0};
		\node [style=filled 10pt] (79) at (6.75, 5.75) {0};
		\node [style=filled W] (80) at (8.25, 5.75) {$\phantom{0}$};
		\node [style=filled 10pt] (81) at (9.75, 5.75) {0};
		\node [style=filled 10pt] (82) at (8.25, 6.5) {1};
		\node [style=filled 10pt] (83) at (9.75, 6.5) {1};
		\node [style=filled 10pt] (84) at (11.25, 6.5) {1};
		\node [style=filled 10pt] (85) at (11.25, 5.75) {1};
		\node [style=none] (86) at (11.25, -4.25) {$\times$};
		\node [style=none] (87) at (11.25, -2.25) {$\times$};
		\node [style=none] (88) at (11.25, -0.25) {$\times$};
		\node [style=none] (89) at (11.25, 1.75) {$\times$};
		\node [style=none] (90) at (11.25, 3.75) {$\times$};
		\node [style=none] (91) at (11.25, -6.25) {$\times$};
		\node [style=none] (92) at (6.75, 1.75) {$\times$};
		\node [style=none] (93) at (8.25, 1.75) {$\times$};
	\end{pgfonlayer}
	\begin{pgfonlayer}{edgelayer}
		\draw [style=thick] (38.center) to (39.center);
		\draw [style=thick] (40.center) to (41.center);
		\draw [style=thick] (56.center) to (55.center);
	\end{pgfonlayer}
\end{tikzpicture}
    \caption{A matrix depicting the possible state transitions between two consecutive path gadgets. The rows are labeled with $\statepair^1 = (\state_3^{\ell_1}, \state_4^{\ell_1})$, $\ell_1 \in [6]$, and the columns are labeled with $\statepair^2 = (\state_1^{\ell_2}, \state_2^{\ell_2})$, $\ell_2 \in [6]$. An $\times$ marks possible state transitions in a connected vertex cover.}
    \label{fig:cvc_cw_transition_table}
  \end{figure}
  
  Define the atomic state pairs $\statepair^1 = (\state^{\ell_1}_3, \state^{\ell_1}_4) = (\statemap_{X \cap P^1}(u^1_3), \statemap_{X \cap P^1}(u^1_4))$ and $\statepair^2 = (\state^{\ell_2}_1, \state^{\ell_2}_2) = (\statemap_{X \cap P^2}(u^2_1), \statemap_{X \cap P^2}(u^2_2))$. We claim that $X$ does not cover some edge in $G[U]$ or some vertex in $X \cap U$ is not root-connected in $X$ whenever $\ell_1 > \ell_2$, see also \cref{fig:cvc_cw_transition_table}. Some edge in $G[U]$ is not covered by $X$ if and only if both $\statepair^1$ and $\statepair^2$ contain a $\zero$. Hence, $(\ell_1, \ell_2) \notin \{(3,1),(3,2),(5,1),(5,2),(5,4),(6,1),(6,2),(6,4)\}$. Some vertex in $X \cap U$ is not root-connected  in $X$ if and only if both $\statepair^1$ and $\statepair^2$ contain no $\one_1$s at all or one consists of two $\zero$s and the other one contains a $\one_0$. This additionally shows that $(\ell_1, \ell_2) \notin \{(2,1),(4,1),(4,2),(4,3),(5,3),(6,3),(6,5)\}$ and concludes the proof of the first part. 
  
  \begin{figure}[h]
    \centering
    \begin{tikzpicture}
	\begin{pgfonlayer}{nodelayer}
		\node [style=filled W] (1) at (-12.25, -3) {$\phantom{0}$};
		\node [style=filled 10pt] (2) at (-2.25, -3) {0};
		\node [style=filled 10pt] (3) at (8, -3) {1};
		\node [style=none] (5) at (-11.25, -3) {=};
		\node [style=none] (6) at (-1.25, -3) {=};
		\node [style=none] (7) at (9, -3) {=};
		\node [style=none] (9) at (-10.25, -3) {$\mathbf{0}$};
		\node [style=none] (10) at (-0.25, -3) {$\mathbf{1}_0$};
		\node [style=none] (11) at (10, -3) {$\mathbf{1}_1$};
		\node [style=none] (16) at (-15.5, 0.75) {$u^1_3$};
		\node [style=none] (17) at (-15.5, -0.75) {$u^1_4$};
		\node [style=none] (18) at (-12, 0.75) {$u^2_1$};
		\node [style=none] (19) at (-12, -0.75) {$u^2_2$};
		\node [style=none] (20) at (-14.75, 2.75) {Case 1:};
		\node [style=filled W] (21) at (-9.5, 0.75) {};
		\node [style=filled W] (22) at (-9.5, -0.75) {};
		\node [style=filled W] (23) at (-8, 0.75) {};
		\node [style=filled W] (24) at (-8, -0.75) {};
		\node [style=none] (29) at (-9.75, 2.75) {Case 2:};
		\node [style=none] (38) at (-4.75, 2.75) {Case 3:};
		\node [style=none] (47) at (0.25, 2.75) {Case 4:};
		\node [style=none] (56) at (5.25, 2.75) {Case 5:};
		\node [style=none] (63) at (-13.75, 2) {};
		\node [style=none] (64) at (-13.75, -2) {};
		\node [style=none] (65) at (-8.75, 2) {};
		\node [style=none] (66) at (-8.75, -2) {};
		\node [style=none] (67) at (-3.75, 2) {};
		\node [style=none] (68) at (-3.75, -2) {};
		\node [style=none] (69) at (1.25, 2) {};
		\node [style=none] (70) at (1.25, -2) {};
		\node [style=none] (71) at (6.25, 2) {};
		\node [style=none] (72) at (6.25, -2) {};
		\node [style=none] (73) at (-10.5, 0.75) {$u^1_3$};
		\node [style=none] (74) at (-10.5, -0.75) {$u^1_4$};
		\node [style=none] (75) at (-7, 0.75) {$u^2_1$};
		\node [style=none] (76) at (-7, -0.75) {$u^2_2$};
		\node [style=none] (77) at (-5.5, 0.75) {$u^1_3$};
		\node [style=none] (78) at (-5.5, -0.75) {$u^1_4$};
		\node [style=none] (79) at (-2, 0.75) {$u^2_1$};
		\node [style=none] (80) at (-2, -0.75) {$u^2_2$};
		\node [style=none] (81) at (-0.5, 0.75) {$u^1_3$};
		\node [style=none] (82) at (-0.5, -0.75) {$u^1_4$};
		\node [style=none] (83) at (3, 0.75) {$u^2_1$};
		\node [style=none] (84) at (3, -0.75) {$u^2_2$};
		\node [style=none] (85) at (4.5, 0.75) {$u^1_3$};
		\node [style=none] (86) at (4.5, -0.75) {$u^1_4$};
		\node [style=none] (87) at (8, 0.75) {$u^2_1$};
		\node [style=none] (88) at (8, -0.75) {$u^2_2$};
		\node [style=filled W] (89) at (-4.5, 0.75) {};
		\node [style=filled W] (90) at (-4.5, -0.75) {};
		\node [style=filled W] (91) at (-3, 0.75) {};
		\node [style=filled W] (92) at (-3, -0.75) {};
		\node [style=filled W] (93) at (0.5, 0.75) {};
		\node [style=filled W] (94) at (0.5, -0.75) {};
		\node [style=filled W] (95) at (2, 0.75) {};
		\node [style=filled W] (96) at (2, -0.75) {};
		\node [style=filled W] (97) at (5.5, 0.75) {};
		\node [style=filled W] (98) at (5.5, -0.75) {};
		\node [style=filled W] (99) at (7, 0.75) {};
		\node [style=filled W] (100) at (7, -0.75) {};
		\node [style=none] (101) at (10.25, 2.75) {Case 6:};
		\node [style=none] (102) at (11.25, 2) {};
		\node [style=none] (103) at (11.25, -2) {};
		\node [style=none] (104) at (9.5, 0.75) {$u^1_3$};
		\node [style=none] (105) at (9.5, -0.75) {$u^1_4$};
		\node [style=none] (106) at (13, 0.75) {$u^2_1$};
		\node [style=none] (107) at (13, -0.75) {$u^2_2$};
		\node [style=filled W] (108) at (10.5, 0.75) {};
		\node [style=filled W] (109) at (10.5, -0.75) {};
		\node [style=filled W] (110) at (12, 0.75) {};
		\node [style=filled W] (111) at (12, -0.75) {};
		\node [style=filled W] (112) at (-14.5, 0.75) {};
		\node [style=filled W] (113) at (-14.5, -0.75) {};
		\node [style=filled W] (114) at (-13, 0.75) {};
		\node [style=filled W] (115) at (-13, -0.75) {};
		\node [style=filled W] (116) at (-13, 0.75) {$\phantom{0}$};
		\node [style=filled W] (117) at (-13, -0.75) {$\phantom{0}$};
		\node [style=filled W] (118) at (-8, -0.75) {$\phantom{0}$};
		\node [style=filled W] (119) at (2, -0.75) {$\phantom{0}$};
		\node [style=filled W] (120) at (-4.5, -0.75) {$\phantom{0}$};
		\node [style=filled W] (121) at (5.5, -0.75) {$\phantom{0}$};
		\node [style=filled W] (122) at (10.5, -0.75) {$\phantom{0}$};
		\node [style=filled W] (123) at (10.5, 0.75) {$\phantom{0}$};
		\node [style=filled 10pt] (124) at (-8, 0.75) {0};
		\node [style=filled 10pt] (125) at (-3, 0.75) {0};
		\node [style=filled 10pt] (126) at (-3, -0.75) {0};
		\node [style=filled 10pt] (127) at (7, -0.75) {0};
		\node [style=filled 10pt] (128) at (-9.5, -0.75) {0};
		\node [style=filled 10pt] (129) at (0.5, -0.75) {0};
		\node [style=filled 10pt] (130) at (0.5, 0.75) {0};
		\node [style=filled 10pt] (131) at (5.5, 0.75) {0};
		\node [style=filled 10pt] (132) at (-14.5, 0.75) {1};
		\node [style=filled 10pt] (133) at (-14.5, -0.75) {1};
		\node [style=filled 10pt] (134) at (12, -0.75) {1};
		\node [style=filled 10pt] (135) at (12, 0.75) {1};
		\node [style=filled 10pt] (136) at (7, 0.75) {1};
		\node [style=filled 10pt] (137) at (2, 0.75) {1};
		\node [style=filled 10pt] (138) at (-4.5, 0.75) {1};
		\node [style=filled 10pt] (139) at (-9.5, 0.75) {1};
	\end{pgfonlayer}
	\begin{pgfonlayer}{edgelayer}
		\draw (21) to (23);
		\draw (23) to (22);
		\draw (22) to (24);
		\draw (24) to (21);
		\draw [style=dotted] (63.center) to (64.center);
		\draw [style=dotted] (65.center) to (66.center);
		\draw [style=dotted] (67.center) to (68.center);
		\draw [style=dotted] (69.center) to (70.center);
		\draw [style=dotted] (71.center) to (72.center);
		\draw (89) to (91);
		\draw (91) to (90);
		\draw (90) to (92);
		\draw (92) to (89);
		\draw (93) to (95);
		\draw (95) to (94);
		\draw (94) to (96);
		\draw (96) to (93);
		\draw (97) to (99);
		\draw (99) to (98);
		\draw (98) to (100);
		\draw (100) to (97);
		\draw [style=dotted] (102.center) to (103.center);
		\draw (108) to (110);
		\draw (110) to (109);
		\draw (109) to (111);
		\draw (111) to (108);
		\draw (112) to (114);
		\draw (114) to (113);
		\draw (113) to (115);
		\draw (115) to (112);
	\end{pgfonlayer}
\end{tikzpicture}
    \caption{Case distinction to show $\statemap_{X^\ell}(\{u^1_3, u^1_4, u^2_1, u^2_2\}) \subseteq \{\zero, \one_1\}$. The left side in each case depicts $\statepair^1$ and the right side depicts $\statepair^2$.}
    \label{fig:cvc_cw_equal_transition}
  \end{figure}
  
  For the second part, notice that $\statemap(X_{P^1}^\ell) = \statemap(X_{P^2}^\ell) = \state^\ell$ by \cref{thm:cvc_cw_state_exists} and by the same approach as in the last paragraph, we see that for $\ell = \ell_1 = \ell_2$ all edges in $G[U]$ are covered and all vertices in $X^\ell \cap U$ are root-connected in $X^\ell \cup \{\rvertex\}$, see \cref{fig:cvc_cw_equal_transition} for the different cases.
\end{proof}

We say that a \emph{cheat occurs} if $\ell_1 < \ell_2$. Creating arbitrarily long paths of the path gadgets $P$, \cref{thm:cvc_cw_path_transition} tells us that at most $|\states| - 1 = 5 = \Oh(1)$ cheats may occur on such a path which allows us to find a \emph{cheat-free region} as outlined previously.

\subsubsection{Complete Construction}

\newcommand{\nregions}{{5\ngrps\grpsize + 1}}
\newcommand{\ncolumns}{{\nclss(\nregions)}}
\newcommand{\sequence}{\mathbf{h}}

\subparagraph*{Setup.} 
Assume that \CVC can be solved in time $\Oh^*((6 - \eps)^{\lcw(G)})$ for some $\eps > 0$. Given a \SAT-instance $\formula$ with $\nvars$ variables and $\nclss$ clauses, we construct an equivalent \CVC instance with linear clique-width approximately $\nvars \log_6(2)$ so that the existence of such an algorithm for \CVC would imply that \SETH is false. 

We pick an integer $\vgrpsize$ only depending on $\eps$; the precise choice of $\vgrpsize$ will be discussed at a later point. The variables of $\formula$ are partitioned into groups of size at most $\vgrpsize$, resulting in $\ngrps = \lceil \nvars / \vgrpsize \rceil$ groups. Furthermore, we pick the smallest integer $\grpsize$ that satisfies $6^\grpsize \geq 2^\vgrpsize$, i.e., $\grpsize = \lceil \log_6(2^\vgrpsize) \rceil$. We now begin with the construction of the \CVC instance $(G = G(\formula, \vgrpsize), \budget)$.

We create the root vertex $\rvertex$ and attach a leaf $\rvertex'$ which forces $\rvertex$ into any connected vertex cover.
For every group $i \in [\ngrps]$, we create $\grpsize$ long path-like gadgets $P^{i, j}$, $j \in [\grpsize]$, where each $P^{i, j}$ consists of $\ncolumns$ copies $P^{i, j, \ell}$, $\ell \in [\ncolumns]$, of the path gadget $P$ and consecutive copies are connected by a join. More precisely, the vertices in some $P^{i, j, \ell}$ inherit their names from $P$ and the superscript of $P^{i, j, \ell}$ and for every $i \in [\ngrps]$, $j \in [\grpsize]$, $\ell \in [\ncolumns - 1]$, the vertices $\{u^{i,j,\ell}_3, u^{i,j,\ell}_4\}$ are joined to the vertices $\{u^{i,j,\ell + 1}_1, u^{i,j,\ell + 1}_2\}$. The ends of each path $P^{i, j}$, i.e.\ the vertices $u^{i, j, 1}_1, u^{i, j, 1}_2, u^{i, j, \ncolumns}_3, u^{i, j, \ncolumns}_4$, are made adjacent to the root $\rvertex$. 

\begin{figure}[h]
  \centering
  \begin{tikzpicture}
	\begin{pgfonlayer}{nodelayer}
		\node [style=tiny rectangle] (0) at (4, 8) {};
		\node [style=tiny rectangle] (1) at (4, 4) {};
		\node [style=tiny rectangle] (2) at (4, -0.75) {};
		\node [style=none] (3) at (4, 2.25) {$\vdots$};
		\node [style=filled] (4) at (0, 8) {};
		\node [style=filled] (5) at (0, 4) {};
		\node [style=filled] (6) at (0, -0.75) {};
		\node [style=filled] (7) at (1, 0.25) {};
		\node [style=filled] (8) at (1, 5) {};
		\node [style=filled] (9) at (1, 9) {};
		\node [style=filled] (10) at (10, 5) {};
		\node [style=filled] (11) at (12.25, 5) {};
		\node [style=none] (18) at (10.25, 5.75) {$z^{i,\ell}$};
		\node [style=none] (19) at (12.25, 5.75) {$\bar{z}^{i, \ell}$};
		\node [style=none] (20) at (0, -1.5) {$x^{i, \ell}_{(6,6,\ldots)}$};
		\node [style=none] (21) at (1.25, 1) {$\bar{x}^{i, \ell}_{(6,6,\ldots)}$};
		\node [style=none] (22) at (0, 3.25) {$x^{i, \ell}_{(2,1,\ldots)}$};
		\node [style=none] (23) at (1.25, 5.75) {$\bar{x}^{i, \ell}_{(2,1,\ldots)}$};
		\node [style=none] (24) at (-0.25, 7.25) {$x^{i, \ell}_{(1,1,\ldots)}$};
		\node [style=none] (25) at (1, 9.75) {$\bar{x}^{i, \ell}_{(1,1,\ldots)}$};
		\node [style=none] (26) at (4, 8.75) {$y^{i, \ell}_{(1,1,\ldots)}$};
		\node [style=none] (27) at (4, 4.75) {$y^{i, \ell}_{(2,1,\ldots)}$};
		\node [style=none] (28) at (4, 0) {$y^{i, \ell}_{(6,6,\ldots)}$};
		\node [style=tiny rectangle] (29) at (-8, 10) {};
		\node [style=tiny rectangle] (30) at (-8, 9) {};
		\node [style=tiny rectangle] (31) at (-8, 8) {};
		\node [style=tiny rectangle] (32) at (-8, 7) {};
		\node [style=tiny rectangle] (33) at (-8, 5) {};
		\node [style=none] (34) at (-10, 11) {};
		\node [style=none] (35) at (-7, 11) {};
		\node [style=none] (36) at (-10, 4) {};
		\node [style=none] (37) at (-7, 4) {};
		\node [style=none] (38) at (-9, 10) {$v_1^{i,1,\ell}$};
		\node [style=none] (39) at (-9, 9) {$v_2^{i,1,\ell}$};
		\node [style=none] (40) at (-9, 8) {$v_3^{i,1,\ell}$};
		\node [style=none] (41) at (-9, 7) {$v_4^{i,1,\ell}$};
		\node [style=none] (42) at (-9, 5) {$v_6^{i,1,\ell}$};
		\node [style=tiny rectangle] (43) at (-8, 2) {};
		\node [style=tiny rectangle] (44) at (-8, 1) {};
		\node [style=tiny rectangle] (45) at (-8, 0) {};
		\node [style=tiny rectangle] (46) at (-8, -1) {};
		\node [style=tiny rectangle] (47) at (-8, -2) {};
		\node [style=none] (48) at (-10, 3) {};
		\node [style=none] (49) at (-7, 3) {};
		\node [style=none] (50) at (-10, -4) {};
		\node [style=none] (51) at (-7, -4) {};
		\node [style=none] (52) at (-9, 2) {$v_1^{i,2,\ell}$};
		\node [style=none] (53) at (-9, 1) {$v_2^{i,2,\ell}$};
		\node [style=none] (54) at (-9, 0) {$v_3^{i,2,\ell}$};
		\node [style=none] (55) at (-9, -1) {$v_4^{i,2,\ell}$};
		\node [style=none] (56) at (-9, -2) {$v_5^{i,2,\ell}$};
		\node [style=none] (57) at (-8.5, -5) {$\vdots$};
		\node [style=none] (58) at (0, 2.25) {$\vdots$};
		\node [style=filled] (59) at (8, -3) {};
		\node [style=filled] (60) at (10, -3) {};
		\node [style=none] (61) at (8.5, -2.5) {$o^\ell$};
		\node [style=none] (62) at (10.5, -2.5) {$\bar{o}^\ell$};
		\node [style=none] (63) at (7.5, -2.25) {};
		\node [style=none] (64) at (7.25, -2.5) {};
		\node [style=tiny rectangle] (65) at (-8, 6) {};
		\node [style=none] (66) at (-9, 6) {$v_5^{i,1,\ell}$};
		\node [style=tiny rectangle] (67) at (-8, -3) {};
		\node [style=none] (68) at (-9, -3) {$v_6^{i,2,\ell}$};
		\node [style=tiny rectangle] (69) at (8, 10) {};
		\node [style=none] (70) at (9, 10) {$=$};
		\node [style=none] (71) at (11.5, 10) {adjacent to $\hat{r}$};
		\node [style=none] (72) at (-8.5, 11.5) {clique $K_6$};
	\end{pgfonlayer}
	\begin{pgfonlayer}{edgelayer}
		\draw [style=thin] (9) to (4);
		\draw [style=thin] (8) to (5);
		\draw [style=thin] (7) to (6);
		\draw [style=thin] (6) to (2);
		\draw [style=thin] (5) to (1);
		\draw [style=thin] (4) to (0);
		\draw [style=thin] (11) to (10);
		\draw [style=thin, in=0, out=150, looseness=0.50] (10) to (0);
		\draw [style=thin, in=0, out=-180, looseness=0.50] (10) to (1);
		\draw [style=thin, in=0, out=-90, looseness=0.50] (10) to (2);
		\draw [style=dashed cyan thick] (37.center) to (35.center);
		\draw [style=dashed cyan thick] (35.center) to (34.center);
		\draw [style=dashed cyan thick] (34.center) to (36.center);
		\draw [style=dashed cyan thick] (36.center) to (37.center);
		\draw [style=dashed cyan thick] (51.center) to (49.center);
		\draw [style=dashed cyan thick] (49.center) to (48.center);
		\draw [style=dashed cyan thick] (48.center) to (50.center);
		\draw [style=dashed cyan thick] (50.center) to (51.center);
		\draw [style=thin, in=180, out=0] (29) to (4);
		\draw [style=thin, in=0, out=-180, looseness=0.75] (4) to (43);
		\draw [style=thin, in=0, out=-180, looseness=0.50] (5) to (30);
		\draw [style=thin, in=0, out=-180] (5) to (43);
		\draw [style=thin, in=0, out=-180, looseness=0.50] (6) to (33);
		\draw [style=thin, in=90, out=-30, looseness=0.50] (1) to (59);
		\draw [style=thin] (59) to (60);
		\draw [style=thin dashed] (63.center) to (59);
		\draw [style=thin dashed] (64.center) to (59);
		\draw [in=0, out=-180] (6) to (67);
	\end{pgfonlayer}
\end{tikzpicture}
  \caption{Decoding and clause gadget for \CVC.}
  \label{fig:cvc_cw_decoding_gadget}
\end{figure}

For every group $i \in [\ngrps]$ and column $\ell \in [\ncolumns]$, we create a \emph{decoding gadget} $D^{i,\ell}$ in the same style as Cygan et al.~\cite{CyganNPPRW11} for \CVC parameterized by pathwidth. Every variable group $i$ has at most $2^\vgrpsize$ possible truth assignments and by choice of $\grpsize$ we have that $6^\grpsize \geq 2^\vgrpsize$, so we can find an injective mapping $\embedding \colon \{0,1\}^\vgrpsize \rightarrow [6]^\grpsize$ which assigns to each truth assignment $\tassign \in \{0,1\}^\vgrpsize$ a sequence $\embedding(\tassign) \in [6]^\grpsize$. For each sequence $\sequence = (h_1, \ldots, h_\grpsize) \in [6]^\grpsize$, we create vertices $x^{i, \ell}_\sequence$, $\bar{x}^{i, \ell}_\sequence$, $y^{i, \ell}_\sequence$ and edges $\{x^{i, \ell}_\sequence, \bar{x}^{i, \ell}_\sequence\}$, $\{x^{i, \ell}_\sequence, y^{i, \ell}_\sequence\}$, $\{y^{i, \ell}_\sequence, \rvertex\}$. Furthermore, we add the edge $\{x^{i, \ell}_\sequence, v^{i, j, \ell}_{h_j}\}$ for all $\sequence = (h_1, \ldots, h_\grpsize) \in [6]^\grpsize$ and $j \in [\grpsize]$. Finally, we create two adjacent vertices $z^{i, \ell}$ and $\bar{z}^{i, \ell}$ and edges $\{z^{i, \ell}, y^{i, \ell}_\sequence\}$ for all $\sequence \in [6]^\grpsize$. Each decoding gadget $D^{i, \ell}$ together with the adjacent path gadgets $P^{i,j,\ell}$, $j \in [\grpsize]$, forms the \emph{block} $B^{i,\ell}$, see \cref{fig:cvc_cw_schematic} for a high-level depiction of the connections between different blocks.

Lastly, we construct the \emph{clause gadgets}.
We number the clauses of $\formula$ by $C_0, \ldots, C_{\nclss - 1}$. For every column $\ell \in [\ncolumns]$, we create an adjacent pair of vertices $o^{\ell}$ and $\bar{o}^{\ell}$. Let $\ell' \in [0, \nclss - 1]$ be the remainder of $(\ell - 1)$ modulo $\nclss$; for all $i \in [\ngrps]$, $\sequence \in \embedding(\{0,1\}^\vgrpsize)$ such that $\embedding^{-1}(\sequence)$ is a truth assignment for variable group $i$ satisfying clause $C_{\ell'}$, we add the edge $\{o^{\ell}, y^{i, \ell}_\sequence\}$. See \cref{fig:cvc_cw_decoding_gadget} for a depiction of the decoding and clause gadgets.

\begin{figure}[h]
  \centering
  \scalebox{.9}{\input{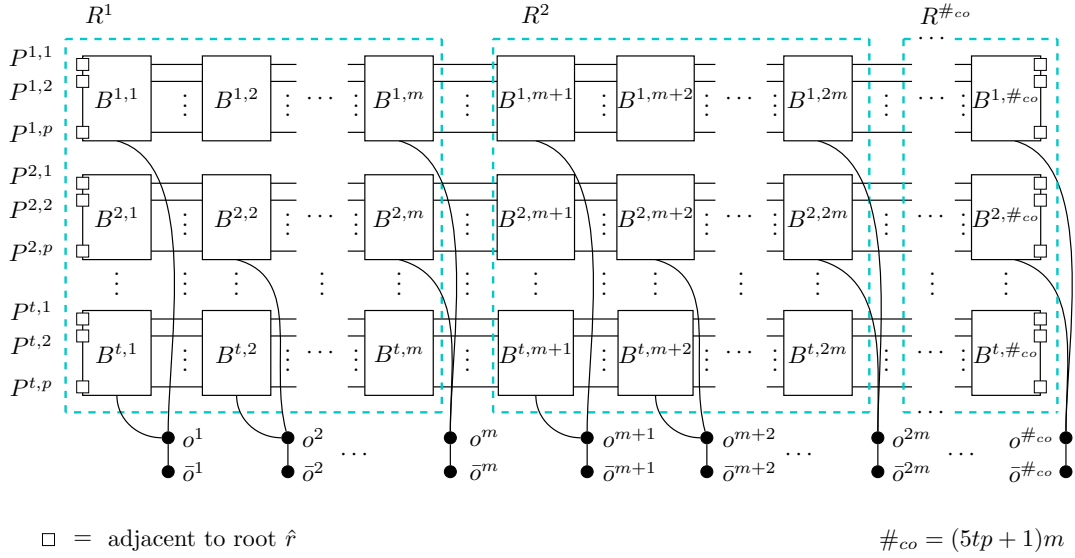}}
  \caption{High-level construction for \CVC. Every edge between two blocks $B^{i,\ell}$ and $B^{i,\ell+1}$ represents a join. The regions $R^\gamma$ are important for the proof of \cref{thm:cvc_cw_sol_to_sat}.}
  \label{fig:cvc_cw_schematic}
\end{figure}

\begin{lem}\label{thm:cvc_cw_sat_to_sol}
 If $\formula$ is satisfiable, then there exists a connected vertex cover $X$ of $G = G(\formula, \vgrpsize)$ of size $|X| \leq (21\ngrps\grpsize + (6^\grpsize + 2)\ngrps + 1)\ncolumns + 1 = \budget$.
\end{lem}

\begin{proof}
  Let $\tassign$ be a satisfying truth assignment of $\formula$ and let $\tassign^i$ denote the restriction of $\tassign$ to the $i$-th variable group for every $i \in [\ngrps]$ and let $\embedding(\tassign^i) = \sequence^i = (h^i_1, \ldots, h^i_\grpsize) \in [6]^{\grpsize}$ be the corresponding sequence. 
 The connected vertex cover is given by 
 \begin{equation*}
   X = \{\rvertex\} \cup \bigcup_{\ell \in [\ncolumns]} \left(\{o^\ell\} \cup \bigcup_{i \in [\ngrps]} \left(\{y^{i, \ell}_{\sequence^i}, z^{i,\ell}\} \cup \bigcup_{\sequence \in [6]^\grpsize} \{x^{i,\ell}_\sequence\} \cup \bigcup_{j \in [\grpsize]} X_{P^{i,j,\ell}}^{h^i_j} \right) \right),
 \end{equation*}
 where $X_{P^{i,j,\ell}}^{h^i_j}$ refers to the sets given by \cref{thm:cvc_cw_state_exists}. 
 
 Clearly, $|X| = \budget$, so it remains to prove that $X$ is a connected vertex cover. By \cref{thm:cvc_cw_state_exists} and the second part of \cref{thm:cvc_cw_path_transition} all edges induced by the path gadgets are covered by $X$ and all vertices on the path gadgets that belong to $X$ are root-connected, except for possibly the vertices at the ends, i.e. $\bigcup_{i \in [\ngrps]} \bigcup_{j \in [\grpsize]} \{u^{i,j,1}_1, u^{i,j,1}_2, u^{i,j,\ncolumns}_3, u^{i,j,\ncolumns}_4\}$, but these are contained in the neighborhood of $\rvertex$ by construction of $G$.
 
 Fix $i \in [\ngrps]$, $\ell \in [\ncolumns]$, and consider the corresponding decoding gadget. Since $z^{i, \ell} \in X$ and $x^{i,\ell}_\sequence \in X$ for all $\sequence \in [6]^\grpsize$, all edges induced by the decoding gadget and all edges between the decoding gadget and the path gadgets are covered by $X$. Furthermore, since $o^{\ell} \in X$, all edges inside the clause gadget and all edges between the clause gadget and the decoding gadgets are covered by $X$. Hence, $X$ has to be a vertex cover of $G$.
 
 It remains to prove that the vertices in the decoding and clause gadgets that belong to $X$ are also root-connected. Again, fix $i \in [\ngrps]$, $\ell \in [\ncolumns]$, and $\sequence = (h_1, \ldots, h_\grpsize) \in [6]^\grpsize \setminus \{\sequence^i\}$. Since $\sequence \neq \sequence^i$, there is some $j \in [\grpsize]$ such that $h_j \neq h_j^i$ and hence $v^{i,j,\ell}_{h_j} \in X$ by \cref{thm:cvc_cw_state_exists} which connects $x_\sequence^{i, \ell}$ to the root $\rvertex$. The vertices $x_{\sequence^i}^{i, \ell}$ and $z^{i, \ell}$ are root-connected via $y^{i, \ell}_{\sequence^i} \in X$.
 
 We conclude by showing that $o^\ell$ is root-connected for all $\ell \in [\ncolumns]$. Since $\tassign$ is a satisfying truth assignment of $\formula$, there is some variable group $i \in [\ngrps]$ such that $\tassign^i$ already satisfies clause $C_{\ell'}$, where $\ell'$ is the remainder of $(\ell - 1)$ modulo $\nclss$. By construction of $G$ and $X$, the vertex $y^{i, \ell}_{\sequence^i} \in X$ is adjacent to $o^\ell$, since $\embedding(\tassign^i) = \sequence^i$, and connects $o^\ell$ to the root $\rvertex$. This shows that all vertices of $X$ are root-connected, so $G[X]$ has to be connected.
\end{proof}

\begin{lem}\label{thm:cvc_cw_sol_to_sat}
  If there exists a connected vertex cover $X$ of $G = G(\formula, \vgrpsize)$ of size $|X| \leq (21\ngrps\grpsize + (6^\grpsize + 2)\ngrps + 1)\ncolumns + 1 = \budget$, then $\formula$ is satisfiable.
\end{lem}

\begin{proof}
  We begin by arguing that $X$ has to satisfy $|X| = \budget$. First, we must have that $\rvertex \in X$, because $\rvertex$ has a neighbor of degree 1. By \cref{thm:cvc_cw_path_gadget_lb}, we have that $|X \cap P^{i,j,\ell}| \geq 21$ for all $i \in [\ngrps]$, $j \in [\grpsize]$, $\ell \in [\ncolumns]$. In every decoding gadget, i.e. one for every $i \in [\ngrps]$ and $\ell \in [\ncolumns]$, the set $\{z^{i,\ell}\} \cup \bigcup_{\sequence \in [6]^\grpsize} x_\sequence^{i, \ell}$ has to be contained in $X$, since every vertex in this set has a neighbor of degree 1. Furthermore, to connect $z^{i,\ell}$ to $\rvertex$, at least one of the vertices $y^{i,\ell}_\sequence$, $\sequence \in [6]^\grpsize$, has to be contained in $X$. Hence, $X$ must contain at least $6^\grpsize + 2$ vertices per decoding gadget. Lastly, $o^\ell \in X$ for all $\ell \in [\ncolumns]$, since $o^\ell$ has a neighbor of degree 1. Since we have only considered disjoint vertex sets, this shows that $|X| = \budget$ and all of the previous inequalities have to be tight.
  
  By \cref{thm:cvc_cw_path_gadget_tight}, we know that $X$ assumes one of the six possible states on each $P^{i,j,\ell}$. Fix some $P^{i,j} = \bigcup_{\ell \in [\ncolumns]} P^{i,j,\ell}$ and note that due to \cref{thm:cvc_cw_path_transition} the state can change at most five times along $P^{i,j}$. Such a state change is called a \emph{cheat}. Let $\gamma \in [0, 5 \ngrps \grpsize]$ and define the $\gamma$-th \emph{region} $R^\gamma = \bigcup_{i \in [\ngrps]} \bigcup_{j \in [\grpsize]} \bigcup_{\ell = \gamma \nclss + 1}^{(\gamma + 1) \nclss} P^{i,j,\ell}$. Since there are $\nregions$ regions, there is at least one region $R^\gamma$ such that no cheat occurs in $R^\gamma$. We will consider this region for the remainder of the proof and read off a satisfying truth assignment from this region.
  
  For $i \in [\ngrps]$, define $\sequence^{i} = (h^i_1, \ldots, h^i_\grpsize) \in [6]^\grpsize$ such that $v^{i,j,\gamma \nclss + 1}_{h^i_j} \notin X$ for all $j \in [\grpsize]$; this is well-defined by \cref{thm:cvc_cw_path_gadget_tight}. Since $R^\gamma$ does not contain any cheats, the definition of $\sequence^i$ is independent of which column $\ell \in [\gamma \nclss + 1, (\gamma + 1)\nclss]$ we consider. For every $i \in [\ngrps]$ and $\ell \in [\gamma \nclss + 1, (\gamma + 1)\nclss]$, we claim that $y^{i, \ell}_\sequence \in X$ if and only if $\sequence = \sequence^i$. We have already established that for every $i$ and $\ell$, there is exactly one $\sequence$ such that $y^{i, \ell}_\sequence \in X$. Consider the vertex $x^{i, \ell}_{\sequence^i} \in X$, its neighbors in $G$ are $v^{i, 1, \ell}_{h^i_1}, v^{i, 2, \ell}_{h^i_2}, \ldots, v^{i, \grpsize, \ell}_{h^i_\grpsize}$, $\bar{x}^{i,\ell}_{\sequence^i}$, and $y^{i,\ell}_{\sequence^i}$. By construction of $\sequence^i$ and the tight allocation of the budget, we have $(N_G(x^{i, \ell}_{\sequence^i}) \setminus \{y^{i, \ell}_{\sequence^i}\}) \cap X = \emptyset$. Therefore, $X$ has to include $y^{i, \ell}_{\sequence^i}$ to connect $x^{i, \ell}_{\sequence^i}$ to the root $\rvertex$. This shows the claim.
  
  For $i \in [\ngrps]$, we define the truth assignment $\tassign^i$ for variable group $i$ by taking an arbitrary truth assignment if $\sequence^i \notin \embedding(\{0,1\}^\vgrpsize)$ and setting $\tassign^i = \embedding^{-1}(\sequence^i)$ otherwise. By setting $\tassign = \bigcup_{i \in [\ngrps]} \tassign^i$ we obtain a truth assignment for all variables and we claim that $\tassign$ satisfies $\formula$. Consider some clause $C_{\ell'}$, $\ell' \in [0, \nclss - 1]$, and let $\ell = \gamma \nclss + \ell' + 1$. We have already argued that $o^\ell \in X$ and to connect $o^\ell$ to the root $\rvertex$, there has to be some $y^{i, \ell}_{\sequence} \in N_G(o^\ell) \cap X$. By the previous claim, $\sequence = \sequence^i$ and therefore $\tassign^i$, and also $\tassign$, satisfy clause $C_{\ell'}$ due to the construction of $G$. Because the choice of $C_{\ell'}$ was arbitrary, $\tassign$ has to be a satisfying assignment of $\formula$.
\end{proof}

\begin{lem}\label{thm:cvc_cw_bound}
  The constructed graph $G = G(\formula, \vgrpsize)$ has $\lcw(G) \leq \ngrps \grpsize + 3 \cdot 6^\grpsize + \Oh(1)$ and a linear clique-expression of this width can be constructed in polynomial time.
\end{lem}

\begin{proof}
  We will describe how to construct a linear clique-expression for $G$ of width $\ngrps \grpsize + 3 \cdot 6^\grpsize + \Oh(1)$. The clique-expression will use one path label $(i,j)$ for every long path $P^{i,j}$, $i \in [\ngrps]$, $j \in [\grpsize]$, temporary decoding labels for every vertex of a decoding gadget, i.e.\ $3 \cdot 6^\grpsize + 2$ many, temporary gadget labels for every vertex of a path gadget, i.e.\ 38 many, two temporary clause labels for the clause gadget, one label for the root vertex and a trash label.
  
  \begin{algorithm}
    Introduce $\rvertex$ with root label and $\rvertex'$ with trash label and add the edge $\{\rvertex, \rvertex'\}$\;
    \For{$\ell \in [\ncolumns]$}
    {
      Introduce $o^\ell$ and $\bar{o}^\ell$ and the edge $\{o^\ell, \bar{o}^\ell\}$ with the clause labels\;
      \For{$i \in [\ngrps]$}
      {
        Build $D^{i,\ell}$ using decoding labels and add edges to $o^\ell$\;
        \For{$j \in [\grpsize]$}
        {
          Build $P^{i,j,\ell}$ using gadget labels and add edges to $D^{i,j}$\;
          Join $u^{i,j,\ell}_1$ and $u^{i,j,\ell}_2$ to path label $(i,j)$ if $\ell > 1$ and to root label otherwise\;
          Relabel path label $(i,j)$ to trash label\;
          Relabel $u^{i,j,\ell}_3$ and $u^{i,j,\ell}_4$ to path label $(i,j)$\;
          Relabel all other gadget labels to the trash label\;
        }
        Relabel all decoding labels to the trash label\;
      }
      Relabel all clause labels to the trash label\;
    }
    \caption{Constructing a linear clique-expression for $G$.}
    \label{algo:cvc_cw_expression}
  \end{algorithm}
  The construction of the clique-expression is described in \cref{algo:cvc_cw_expression} and the central idea is to proceed column by column and group by group in each column. By reusing the temporary labels, we keep the total number of labels small. The maximum number of labels used simultaneously occurs in line 7 and is $\ngrps \grpsize + (3 \cdot 6^\grpsize + 2) + 38 + 2 + 1 + 1$. This concludes the proof.
\end{proof}

\begin{thm}
  There is no algorithm that solves \CVC, given a linear $k$-expression, in time $\Oh^*((6 - \eps)^k)$ for some $\eps > 0$, unless \CNFSETH fails.
\end{thm}

\begin{proof}
  Assume that there exists an algorithm $\algo$ that solves \CVC in time $\Oh^*((6 - \eps)^k)$ for some $\eps > 0$ given a linear $k$-expression. Given $\vgrpsize$, we define $\delta_1 < 1$ such that $(6 - \eps)^{\log_6(2)} = 2^{\delta_1}$ and $\delta_2$ such that $(6 - \eps)^{1 / \vgrpsize} = 2^{\delta_2}$. By picking $\vgrpsize$ large enough, we can ensure that $\delta = \delta_1 + \delta_2 < 1$. We will show how to solve \SAT using $\algo$ in time $\Oh^*(2^{\delta \nvars})$, where $\nvars$ is the number of variables, thus contradicting \CNFSETH.
  
  Given a \SAT instance $\formula$, we construct $G = G(\formula, \vgrpsize)$ and the linear clique-expression from \cref{thm:cvc_cw_bound} in polynomial time, note that we have $\vgrpsize = \Oh(1)$ and hence $\grpsize = \Oh(1)$; recall $\grpsize = \lceil \log_6(2^\vgrpsize) \rceil$. We then run $\algo$ on $G$ and return its answer. This is correct by \cref{thm:cvc_cw_sat_to_sol} and \cref{thm:cvc_cw_sol_to_sat}. Due to \cref{thm:cvc_cw_bound}, the running time is
  \begin{alignat*}{6}
    \phantom{\leq} \quad & \Oh^*\left( (6 - \eps)^{\ngrps \grpsize + 3 \cdot 6^\grpsize + \Oh(1)} \right)
    & \,\,\leq\,\, & \Oh^*\left( (6 - \eps)^{\ngrps \grpsize} \right) 
    & \,\,\leq\,\, & \Oh^*\left( (6 - \eps)^{\lceil \frac{\nvars}{\vgrpsize} \rceil \grpsize} \right) \\
    \leq \quad & \Oh^*\left( (6 - \eps)^{\frac{\nvars}{\vgrpsize} \grpsize} \right) 
    & \,\,\leq\,\, & \Oh^*\left( (6 - \eps)^{\frac{\nvars}{\vgrpsize} \lceil \log_6(2^\vgrpsize) \rceil} \right) 
    & \,\,\leq\,\, & \Oh^*\left( (6 - \eps)^{\frac{\nvars}{\vgrpsize} \log_6(2^\vgrpsize)} (6 - \eps)^{\frac{\nvars}{\vgrpsize}} \right)\\
    \leq \quad & \Oh^*\left( 2^{\delta_1 \vgrpsize \frac{\nvars}{\vgrpsize}} 2^{\delta_2 \nvars} \right)
    & \,\,\leq\,\, & \Oh^*\left( 2^{(\delta_1 + \delta_2) \nvars} \right)
    & \,\,\leq\,\, & \Oh^*\left( 2^{\delta \nvars} \right),
    \end{alignat*}
    hence completing the proof.
\end{proof}

\subsection{Connected Dominating Set}
\label{sec:cw_cds_lb}

This subsection is devoted to proving that \CDS (with unit costs) cannot be solved in time $\Oh^*((5-\eps)^{\lcw(G)})$ for some $\eps > 0$ unless the \CNFSETH fails. We briefly outline the intuition behind the design of the path gadget, which largely follows the same approach as for \CVC. Afterwards, we present the construction of the path gadget and analyze it, then we move on to the complete construction and correctness proofs. The decoding gadgets are again directly adapted from the lower bound for \CVC parameterized by pathwidth given by Cygan et al.~\cite{CyganNPPRW11}. 

\subparagraph*{Root.} We create a distinguished vertex $\rvertex$ called the \emph{root} and by attaching a vertex of degree 1 to $\rvertex$ we ensure that every connected dominating set has to contain $\rvertex$.

\subsubsection{Path Gadget}

To rule out the running time $\Oh^*((5-\eps)^{\lcw(G)})$ for any $\eps > 0$, we have to build a path gadget that admits 5 distinct states and narrows down to a single label, so that each row of the construction contributes one unit of linear clique-width. We begin by analyzing the possible behaviors of a partial solution on a label. 

First, we consider the possible states of a single vertex $v$ with respect to a partial solution $X$. Compared to \CVC, there is one more state, as the state $\zero$ splits into the states $\zero_1$ and $\zero_0$, which denote that $v \notin X$ and whether $v$ is dominated by $X$ or not, and we keep the states $\one_1$ and $\one_0$. Hence, the state of a label can be represented by a subset of $\{\zero_0, \zero_1, \one_0, \one_1\}$. 

Similar to before, we study the compatibility of these label states across a join. The result is a matrix of size $15 \times 15$, as we can exclude the empty subset. However, many states lead to the same compatibility pattern, e.g.\ for any subset $\emptyset \neq \stateset \subseteq \{\zero_0, \zero_1, \one_0, \one_1\}$ the states $\stateset$ and $\stateset \cup \{\zero_1\}$ yield the same compatibility pattern, since the vertex state $\zero_1$ does not add any additional constraint. It turns out that the compatibility matrix contains only five distinct rows, which give rise to the triangular submatrix \cref{table:triangular_cds} of size $5 \times 5$ after reordering. Surprisingly, the number of redundancies is so large that, although \CDS has four vertex states compared to only three for \CVC, we end up with fewer label states than for \CVC.

\begin{table}
\begin{center}
  \begin{tabular}{l|ccccc}%
    & $\{\zero_1, \zero_0\}$ & $\{\zero_1\}$ & $\{\one_0, \zero_0\}$ & $\{\one_1, \zero_0\}$ & $\{\one_1\}$ \\%
    \hline%
    $\{\one_1\}$ & 1 & 1 & 1 & 1 & 1 \\%
    $\{\zero_1\}$ & 0 & 1 & 0 & 0 & 1 \\%
    $\{\one_1, \zero_0\}$ & 0 & 0 & 1 & 1 & 1 \\%
    $\{\one_0, \zero_0\}$ & 0 & 0 & 0 & 1 & 1 \\%
    $\{\zero_1, \zero_0\}$ & 0 & 0 & 0 & 0 & 1 
  \end{tabular}%
\end{center}\caption{A largest triangular submatrix, after reordering rows and columns, of the compatibility matrix for \CDS.}\label{table:triangular_cds}
\vspace*{-0.8cm}
\end{table}

Again, it is sufficient to take independent sets of size two as labels incident to the join. The relevant label states will be represented by the following ordered pairs of vertex states: $(\zero_1, \zero_0)$, $(\zero_1, \zero_1)$, $(\one_0, \zero_0)$, $(\one_1, \zero_0)$, $(\one_1, \one_1)$. By pairing these states along the diagonal of the triangular submatrix, we then obtain the desired states for the path gadget in the transition order.

\subparagraph*{Formal Definition of States.} We define the four atomic states $\atoms = \{\zero_0, \zero_1, \one_0, \one_1\}$ and define the three predicates $\sol, \conn, \dom \colon \atoms \rightarrow \{0,1\}$ by $\sol(\bolda) = [\bolda \in \{\one_0, \one_1\}]$, $\conn(\bolda) = [\bolda = \one_1]$, and $\dom(\bolda) = [\bolda = \zero_1]$. The atoms $\zero_1$ and $\zero_0$ mean that a vertex is not inside the partial solution and the subscript denotes whether the vertex is dominated by the partial solution or not; $\one_1$ and $\one_0$ indicate that a vertex is inside the partial solution and the subscript indicates whether it is root-connected or not. Building on these atomic states, we define five states consisting of four atomic states each: 
\begin{align*}
  \state^1 & = (\zero_1, \zero_0, \one_1, \zero_1) \\
  \state^2 & = (\zero_1, \zero_1, \zero_1, \zero_1) \\
  \state^3 & = (\one_0, \zero_0, \one_1, \zero_0) \\
  \state^4 & = (\one_1, \zero_0, \one_0, \zero_0) \\
  \state^5 & = (\one_1, \zero_1, \zero_1, \zero_0) 
\end{align*}
We collect these states into the set $\states = \{\state^1, \ldots, \state^5\}$ and use the notation $\state^\ell_i \in \atoms$, $\ell \in [5]$, $i \in [4]$, to refer to the $i$-th atomic component of state $\state^\ell$. Note that $\state^5$ can be obtained from $\state^1$ by swapping the first two components with the last two components; in the same way $\state^4$ can be obtained from $\state^3$.

Given a partial solution $Y \subseteq V(G)$, we associate to each vertex its state in $Y$ with the map $\statemap_Y \colon V(G) \setminus \{\rvertex\} \rightarrow \atoms$, which is defined by
\begin{equation*}
  \statemap_Y(v) = \begin{cases}
    \zero_0 & \text{if } v \notin N[Y \cup \{\rvertex\}], \\
    \zero_1 & \text{if } v \in N[Y \cup \{\rvertex\}] \setminus Y, \\
    \one_0 & \text{if } v \in Y \text{ and $v$ is not root-connected in $Y \cup \{\rvertex\}$}, \\
    \one_1 & \text{if } v \in Y \text{ and $v$ is root-connected in $Y \cup \{\rvertex\}$}. \\
  \end{cases}
\end{equation*}

\subparagraph*{Subdivided edges.} In the graph construction, we frequently need \emph{subdivided edges}. Given two vertices $u$ and $v$, adding a \emph{subdivided edge} between $u$ and $v$ means adding a new vertex $w_{\{u,v\}}$ and the edges $\{u, w_{\{u,v\}}\}$ and $\{w_{\{u,v\}}, v\}$. The crucial property of a subdivided edge between $u$ and $v$ is that any connected dominating set $X$ has to contain at least one of $u$ and $v$, since $u$ and $v$ remain the only neighbors of $w_{\{u,v\}}$ throughout the entire construction. In this sense, connected dominating sets behave in regards to subdivided edges as vertex covers do to normal edges, hence allowing us to adapt a substantial part of the construction from \CVC to \CDS.

\subparagraph*{Formal Construction.} We proceed by describing how to construct the path gadget $P$. We create 4 \emph{join} vertices $u_{1,1}, u_{1,2}, u_{2,1}, u_{2,2}$, 6 \emph{auxiliary} vertices $a_{1,1}, a_{1,2}, a_{1,3}, a_{2,1}, a_{2,2}, a_{2,3}$, 4 \emph{solution indicator} vertices $b_{1,0}, b_{1,1}, b_{2,0}, b_{2,1}$, 4 \emph{connectivity indicator} vertices $c_{1,0}, c_{1,1}, c_{2,0}, c_{2,1}$, 4 \emph{domination indicator} vertices $d_{1,0}, d_{1,1}, d_{2,0}, d_{2,1}$ and 5 \emph{clique} vertices $v_1, \ldots, v_5$. For every $i \in [2]$, we add the edges $\{u_{i,1}, a_{i,1}\}$, $\{a_{i,1},a_{i,2}\}$, $\{a_{i,1}, b_{i,0}\}$, $\{a_{i,1}, c_{i,1}\}$, and $\{u_{i,2}, d_{i,1}\}$. Moreover, for every $i \in [2]$, we add subdivided edges from $u_{i,1}$ to $a_{i,3}$ and $b_{i,0}$; from $a_{i,3}$ to $b_{i,1}$; from $b_{i,0}$ to $b_{i,1}$; from $c_{i,0}$ to $c_{i,1}$; and from $d_{i,0}$ to $d_{i,1}$.

Furthermore, we make $a_{i,3}$ for all $i \in [2]$, all solution indicator vertices, all connectivity indicator vertices, all domination indicator vertices, and all clique vertices adjacent to the root vertex $\rvertex$. We add all possible subdivided edges between the clique vertices $v_\ell$, $\ell \in [5]$, so that they induce a clique of size 5 where every edge is subdivided. 

Finally, we explain how to connect the indicator vertices to the clique vertices. The clique vertex $v_\ell$ corresponds to choosing state $\state^\ell$ on the join vertices $(u_{1,1}, u_{1,2}, u_{2,1}, u_{2,2})$. The desired behavior of a partial solution $X$ of $P + \rvertex$ is as follows:
\begin{itemize}
  \item $b_{i,1} \in X \iff u_{i,1} \in X$,
  \item $c_{i,1} \in X \iff u_{i,1} \in X$ and $u_{i,1}$ is root-connected in $X \cup \{\rvertex\}$,
  \item $d_{i,1} \in X \iff u_{i,2} \in N[X] \setminus X$, i.e., $u_{1,2}$ is dominated by $X$.
\end{itemize}
Accordingly, for all $i \in [2]$ and $\ell \in [5]$, we add subdivided edges from $v_\ell$ to $b_{i,\sol(\state^\ell_{2i-1})}$ and $c_{i,\conn(\state^\ell_{2i-1})}$ and $d_{i,\dom(\state^\ell_{2i})}$. This concludes the construction of $P$.

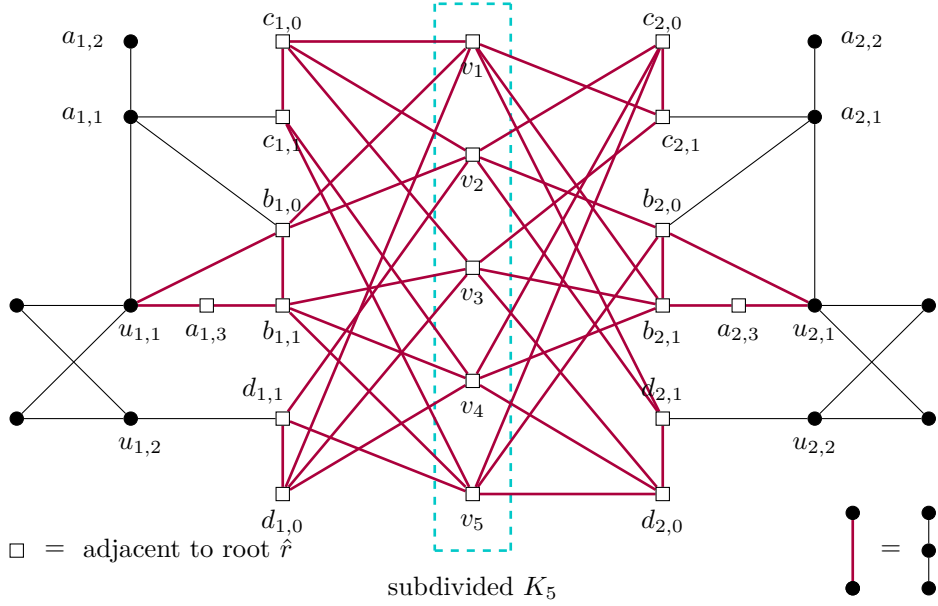
\begin{figure}[h]
  \centering
  \begin{tikzpicture}
	\begin{pgfonlayer}{nodelayer}
		\node [style=tiny rectangle] (0) at (0, -6) {};
		\node [style=tiny rectangle] (1) at (0, -3) {};
		\node [style=tiny rectangle] (2) at (0, 0) {};
		\node [style=tiny rectangle] (3) at (0, 3) {};
		\node [style=tiny rectangle] (4) at (0, 6) {};
		\node [style=filled] (6) at (-9, -1) {};
		\node [style=filled] (7) at (-9, -4) {};
		\node [style=filled] (8) at (-12, -1) {};
		\node [style=filled] (9) at (-12, -4) {};
		\node [style=filled] (10) at (-9, 4) {};
		\node [style=filled] (11) at (-9, 6) {};
		\node [style=tiny rectangle] (12) at (-7, -1) {};
		\node [style=tiny rectangle] (13) at (-5, -1) {};
		\node [style=tiny rectangle] (14) at (-5, 1) {};
		\node [style=tiny rectangle] (15) at (-5, 4) {};
		\node [style=tiny rectangle] (16) at (-5, 6) {};
		\node [style=tiny rectangle] (17) at (-5, -4) {};
		\node [style=tiny rectangle] (18) at (-5, -6) {};
		\node [style=filled] (19) at (9, -1) {};
		\node [style=filled] (20) at (9, -4) {};
		\node [style=filled] (21) at (12, -1) {};
		\node [style=filled] (22) at (12, -4) {};
		\node [style=filled] (23) at (9, 4) {};
		\node [style=filled] (24) at (9, 6) {};
		\node [style=tiny rectangle] (25) at (7, -1) {};
		\node [style=tiny rectangle] (26) at (5, -1) {};
		\node [style=tiny rectangle] (27) at (5, 1) {};
		\node [style=tiny rectangle] (28) at (5, 4) {};
		\node [style=tiny rectangle] (29) at (5, 6) {};
		\node [style=tiny rectangle] (30) at (5, -4) {};
		\node [style=tiny rectangle] (31) at (5, -6) {};
		\node [style=none] (32) at (-1, 7) {};
		\node [style=none] (33) at (1, 7) {};
		\node [style=none] (34) at (-1, -7.5) {};
		\node [style=none] (35) at (1, -7.5) {};
		\node [style=none] (36) at (0, 5.25) {$v_1$};
		\node [style=none] (37) at (0, 2.25) {$v_2$};
		\node [style=none] (38) at (0, -0.75) {$v_3$};
		\node [style=none] (39) at (0, -3.75) {$v_4$};
		\node [style=none] (40) at (0, -6.75) {$v_5$};
		\node [style=none] (46) at (0, -8.5) {subdivided $K_5$};
		\node [style=filled] (51) at (12, -7.5) {};
		\node [style=none] (52) at (11, -7.5) {=};
		\node [style=filled] (53) at (10, -6.5) {};
		\node [style=filled] (54) at (10, -8.5) {};
		\node [style=filled] (55) at (12, -6.5) {};
		\node [style=filled] (56) at (12, -8.5) {};
		\node [style=filled] (57) at (10, -8.5) {};
		\node [style=tiny rectangle] (58) at (-12, -7.5) {};
		\node [style=none] (59) at (-11, -7.5) {=};
		\node [style=none] (60) at (-7.5, -7.5) {adjacent to root $\hat{r}$};
		\node [style=none] (61) at (-8.75, -1.75) {$u_{1,1}$};
		\node [style=none] (62) at (-8.75, -4.75) {$u_{1,2}$};
		\node [style=none] (63) at (-10.25, 4) {$a_{1,1}$};
		\node [style=none] (64) at (-10.25, 6) {$a_{1,2}$};
		\node [style=none] (65) at (-7, -1.75) {$a_{1,3}$};
		\node [style=none] (66) at (-5, -1.75) {$b_{1,1}$};
		\node [style=none] (67) at (-5, 1.75) {$b_{1,0}$};
		\node [style=none] (68) at (-5, 3.25) {$c_{1,1}$};
		\node [style=none] (69) at (-5, 6.5) {$c_{1,0}$};
		\node [style=none] (70) at (-5, -6.75) {$d_{1,0}$};
		\node [style=none] (71) at (-5.5, -3.25) {$d_{1,1}$};
		\node [style=none] (72) at (9, -1.75) {$u_{2,1}$};
		\node [style=none] (73) at (9, -4.75) {$u_{2,2}$};
		\node [style=none] (74) at (10.25, 4) {$a_{2,1}$};
		\node [style=none] (75) at (10.25, 6) {$a_{2,2}$};
		\node [style=none] (76) at (7, -1.75) {$a_{2,3}$};
		\node [style=none] (77) at (5, -1.75) {$b_{2,1}$};
		\node [style=none] (78) at (5, 1.75) {$b_{2,0}$};
		\node [style=none] (79) at (5.5, 3.25) {$c_{2,1}$};
		\node [style=none] (80) at (5, 6.5) {$c_{2,0}$};
		\node [style=none] (81) at (5, -6.75) {$d_{2,0}$};
		\node [style=none] (82) at (5, -3.25) {$d_{2,1}$};
	\end{pgfonlayer}
	\begin{pgfonlayer}{edgelayer}
		\draw (9) to (7);
		\draw (7) to (8);
		\draw (8) to (6);
		\draw (6) to (9);
		\draw (11) to (10);
		\draw (10) to (6);
		\draw (10) to (15);
		\draw (7) to (17);
		\draw [style=packing] (6) to (12);
		\draw [style=packing] (12) to (13);
		\draw [style=packing] (13) to (14);
		\draw [style=packing] (14) to (6);
		\draw [style=packing] (16) to (15);
		\draw [style=packing] (17) to (18);
		\draw (10) to (14);
		\draw (22) to (20);
		\draw (20) to (21);
		\draw (21) to (19);
		\draw (19) to (22);
		\draw (24) to (23);
		\draw (23) to (19);
		\draw (23) to (28);
		\draw (20) to (30);
		\draw [style=packing] (19) to (25);
		\draw [style=packing] (25) to (26);
		\draw [style=packing] (26) to (27);
		\draw [style=packing] (27) to (19);
		\draw [style=packing] (29) to (28);
		\draw [style=packing] (30) to (31);
		\draw (23) to (27);
		\draw [style=dashed cyan thick] (35.center) to (33.center);
		\draw [style=dashed cyan thick] (33.center) to (32.center);
		\draw [style=dashed cyan thick] (32.center) to (34.center);
		\draw [style=dashed cyan thick] (34.center) to (35.center);
		\draw [style=packing] (15) to (0);
		\draw [style=packing] (15) to (1);
		\draw [style=packing] (16) to (2);
		\draw [style=packing] (16) to (3);
		\draw [style=packing] (16) to (4);
		\draw [style=packing] (14) to (3);
		\draw [style=packing] (14) to (4);
		\draw [style=packing] (13) to (2);
		\draw [style=packing] (13) to (1);
		\draw [style=packing] (0) to (13);
		\draw [style=packing] (17) to (3);
		\draw [style=packing] (18) to (4);
		\draw [style=packing] (17) to (0);
		\draw [style=packing] (1) to (18);
		\draw [style=packing] (18) to (2);
		\draw [style=packing] (0) to (29);
		\draw [style=packing] (29) to (1);
		\draw [style=packing] (29) to (3);
		\draw [style=packing] (28) to (4);
		\draw [style=packing] (2) to (28);
		\draw [style=packing] (27) to (0);
		\draw [style=packing] (27) to (3);
		\draw [style=packing] (26) to (1);
		\draw [style=packing] (2) to (26);
		\draw [style=packing] (26) to (4);
		\draw [style=packing] (4) to (30);
		\draw [style=packing] (30) to (3);
		\draw [style=packing] (31) to (2);
		\draw [style=packing] (1) to (31);
		\draw [style=packing] (31) to (0);
		\draw [style=packing] (53) to (57);
		\draw (55) to (51);
		\draw (51) to (56);
	\end{pgfonlayer}
\end{tikzpicture}
  \caption{The path gadget for \CDS parameterized by linear clique-width. Vertices denoted by squares are adjacent to the root vertex $\rvertex$ and bold red edges represent subdivided edges. The edges of the subdivided 5-clique formed by the clique vertices $v_\ell$, $\ell \in [5]$, are not depicted. Besides the connections to $\rvertex$, only the join vertices $u_{i,1}, u_{i,2}$, $i \in [2]$, and clique vertices $v_\ell$, $\ell \in [5]$, have connections to outside this path gadget.}
\end{figure}

\subparagraph*{Behavior of a Single Path Gadget.} For the upcoming lemmas, we assume that $G$ is a graph that contains $P + \rvertex$ as an induced subgraph and that only the join vertices $u_{i,1}, u_{i,2}, i \in [2]$, and clique vertices $v_\ell, \ell \in [5]$, have neighbors outside this copy of $P + \rvertex$. Furthermore, let $X$ be a connected dominating set of $G$ with $\rvertex \in X$. We study the behavior of such connected dominating sets on $P$; we will abuse notation and write $X \cap P$ instead of $X \cap V(P)$. The assumption on how $P$ connects to the remaining graph implies that any vertex $v \in V(P)$ with $\statemap_{X \cap P}(v) = \one_0$ has to be root-connected in $X$ through some path that leaves $P + \rvertex$ via one of the join vertices $u_{i,j}$, $i,j \in [2]$. Note that any path leaving $P + \rvertex$ through some clique vertex $v_\ell$, $\ell \in [5]$, immediately yields a path to $\rvertex$ in $P + \rvertex$ as $\{v_\ell \sep \ell \in [5]\} \subseteq N(\rvertex)$.

We begin by showing a lower bound for $|X \cap P|$ via a vertex-disjoint packing of subgraphs.

\begin{lem}\label{thm:cds_cw_path_gadget_lb}
  Any connected dominating set $X$ with $\rvertex \in X$ satisfies $|X \cap P| \geq 14 = 2 \cdot 5 + 4$ and more specifically, $a_{i,1} \in X$, $|X \cap \{b_{i,0}, b_{i,1}, a_{i,3}, u_{i,1}\}| \geq 2$, $X \cap \{c_{i,0}, c_{i,1}\} \neq \emptyset$, $X \cap \{d_{i,0}, d_{i,1}\} \neq \emptyset$ for all $i \in [2]$ and $|X \cap \{v_1, \ldots, v_5\}| \geq 4$.
\end{lem}

\begin{proof}
  First, observe that the closed neighborhoods of each of the following vertices are disjoint: $a_{i,2}$, $w_{\{b_{i,0}, b_{i,1}\}}$, $w_{\{a_{i,3}, u_{i,1}\}}$, $w_{\{c_{i,0}, c_{i,1}\}}$, $w_{\{d_{i,0}, d_{i,1}\}}$ for all $i \in [2]$. Since $X$ is a dominating set, $X$ contains at least one vertex of each of these neighborhoods. Since $X$ is connected, it follows that $a_{i,1} \in X$ and for each of the subdivided edges that one of the endpoints must be contained in $X$. 
  
  Next, we turn to the subdivided 5-clique. Again, $X$ has to contain at least one endpoint of each subdivided edge present in this clique. If there were two clique vertices $v_i, v_j \notin X$, $i \neq j$, then the subdivided edge between them is not resolved by $X$. Therefore, $X$ can avoid taking at most one of the clique vertices. This concludes the proof.
\end{proof}

For any partial solution $Y \subseteq V(P)$, we formalize the states at the boundary of $P$ with the 4-tuple $\statemap(Y) = (\statemap_Y(u_{1,1}), \statemap_Y(u_{1,2}), \statemap_Y(u_{2,1}), \statemap_Y(u_{2,2}))$. The following lemma shows that the states communicated to the boundary depend on the state of the central clique in the desired way.

\begin{lem}\label{thm:cds_cw_path_gadget_ub}
  Any connected dominating set $X$ with $\rvertex \in X$ and $|X \cap P| \leq 14$ satisfies $a_{i,1} \in X$, $X \cap \{b_{i,0}, b_{i,1}, a_{i,3}, u_{i,1}\} \in \{\{b_{i,0}, a_{i,3}\}, \{b_{i,1}, u_{i,1}\}\}$, $|X \cap \{c_{i,0}, c_{i,1}\}| = 1$, $|X \cap \{d_{i,0}, d_{i,1}\}| = 1$ for all $i \in [2]$ and there exists a unique $\ell \in [5]$ such that $v_\ell \notin X$. Furthermore, we have that $\statemap(X \cap P) = \state^\ell$.
\end{lem}

\begin{proof}
  All the inequalities of \cref{thm:cds_cw_path_gadget_lb} have to be tight in this case which proves everything of the first part except the part regarding $\{b_{i,0}, b_{i,1}, a_{i,3}, u_{i,1}\}$. We know that $X$ contains exactly two of these vertices, but if $X$ contains two that are connected by a subdivided edge, then the subdivided edge between the other two vertices is not resolved. Therefore, $X \cap \{b_{i,0}, b_{i,1}, a_{i,3}, u_{i,1}\} \in \{\{b_{i,0}, a_{i,3}\}, \{b_{i,1}, u_{i,1}\}\}$.
  
  It remains to prove the property regarding the states of the join vertices. Since $v_\ell \notin X$, the other endpoints of the incident subdivided edges have to be contained in $X$. By construction, these are the vertices $b_{i,\sol(\state^\ell_{2i-1})}$, $c_{i,\conn(\state^\ell_{2i-1})}$, and $d_{i,\dom(\state^\ell_{2i})}$ for $i \in [2]$. Fix $i \in [2]$, and note that by the budget allocation inside $P$, it follows that $u_{i,2} \notin X$. Therefore, $u_{i,2}$ can only be dominated by $d_{i,1}$ inside $P$ and $\statemap_{X \cap P}(u_{i,2}) = \zero_1 \iff d_{i,1} \in X \iff \dom(\state_{2i}^\ell) = 1 \iff \state_{2i}^\ell = \zero_1$. We have that $u_{i,1} \in X \iff b_{i,1} \in X \iff \sol(\state_{2i-1}^\ell) = 1 \iff \state_{2i-1}^\ell \in \{\one_0, \one_1\}$ and  $u_{i,1}$ is always  dominated by $a_{i,1}$, hence $\statemap_{X \cap P}(u_{i,1}) \neq \zero_0$ in all cases. If $u_{i,1} \in X$, then its only neighbor inside $X \cap P$ is $a_{i,1}$ which brings root-connectivity to $u_{i,1}$ if and only if $c_{i,1} \in X$, therefore $\statemap_{X \cap P}(u_{i,1}) = \one_1 \iff c_{i,1} \in X \iff \conn(\state^\ell_{2i-1}) = 1 \iff \state^\ell_{2i-1} = \one_1$. This concludes the proof.
\end{proof}

Next, we establish that for any $\state^\ell \in \states$, a partial solution attaining $\state^\ell$ actually exists and that these partial solutions are root-connected and dominate everything inside $P$, if this holds at the join vertices $u_{i,j}$, $i,j \in [2]$.

\begin{lem}\label{thm:cds_cw_state_exists}
  For every $\ell \in [5]$, there exists a set $X_P^\ell$ of $P$ such that $|X_P^\ell| = 14$, $X_P^\ell \cap \{v_1, \ldots, v_5\} = \{v_1, \ldots, v_5\} \setminus \{v_\ell\}$, $\statemap(X_P^\ell) = \state^\ell$, and 
  \begin{equation*}
    \statemap_{X_P^\ell}(V(P) \setminus \{a_{1,1}, a_{2,1}, u_{1,1}, u_{1,2}, u_{2,1}, u_{2,2}\}) \subseteq \{\zero_1, \one_1\}. 
  \end{equation*}
  If $X$ is a vertex subset of $G$ with $\rvertex \in X$ and $X \cap P = X_P^\ell$ and $\statemap_{X }(\{u_{i,1}, u_{1,2}\}) \subseteq \{\zero_1, \one_1\}$ for every $i \in [2]$, then $\statemap_{X}(V(P)) \subseteq \{\zero_1, \one_1\}$.
\end{lem}

\begin{proof}
  We define 
  \begin{align*}
    X^\ell_P & = (\{v_1, \ldots, v_5\} \setminus \{v_\ell\}) \cup \{a_{i,1} \sep i \in [2]\} \\
    & \cup \{u_{i,1} \sep i \in [2] \text{ and } \sol(\state^\ell_{2i-1}) = 1\} \cup \{a_{i,3} \sep i \in [2] \text{ and } \sol(\state^\ell_{2i-1}) = 0\} \\
    & \cup \{b_{i, \sol(\state^\ell_{2i-1})} \sep i \in [2]\} \cup \{c_{i, \conn(\state^\ell_{2i-1})} \sep i \in [2]\} \cup \{d_{i, \dom(\state^\ell_{2i})} \sep i \in [2]\}
  \end{align*} 
  and claim that $X^\ell_P$ has the desired properties. We clearly have that $|X^\ell_P| = 14$ and $v_\ell \notin X^\ell_P$. We proceed by showing that every vertex of $P$ with the 6 exceptions $a_{1,1}, a_{2,1}, u_{1,1}, u_{1,2}, u_{2,1}, u_{2,2}$ is dominated by $X^\ell_P$ or root-connected in $X^\ell_P + \rvertex$. 
  
  First, observe that for any vertex $v \in N(\rvertex) \cap V(P)$, we have that $\statemap_{X^\ell_P}(v) \in \{\zero_1, \one_1\}$. Since $a_{i,1} \in X^\ell_P$ and $a_{i,2} \notin X^\ell_P$, we have that $\statemap_{X^\ell_P}(a_{i,2}) = \zero_1$ for every $i \in [2]$. It remains to handle the subdividing vertices $w_{\{x,y\}}$. Since $X^\ell_P$ contains one vertex from every pair of indicator vertices, every subdividing edge between such a pair is resolved. By construction $X^\ell_P$ either contains the pair $\{b_{i,0}, a_{i,3}\}$ or the pair $\{b_{i,1}, u_{i,1}\}$ for every $i \in [2]$ and in either case the subdivided edges incident to $a_{i,3}$ and $u_{i,1}$ are resolved. The subdivided edges between the clique vertices are resolved, because $X^\ell_P$ contains 4 out of 5 clique vertices. Finally, the subdivided edges between $v_\ell$ and the indicator vertices are resolved, since by construction of $P$, the subdivided edges incident to $v_\ell$ lead to $b_{i,\sol(\state^\ell_{2i-1})}$, $c_{i,\conn(\state^\ell_{2i-1})}$, and $d_{i,\dom(\state^\ell_{2i})}$ for all $i \in [2]$ which are precisely the indicator vertices contained in $X^\ell_P$. 
  
  The claim $\statemap(X^\ell_P) = \state^\ell$ follows by very similar arguments as in the proof of \cref{thm:cds_cw_path_gadget_ub}.
  
  We proceed with the second part. By the first part, it remains to handle the vertices $a_{1,1}, a_{2,1}, u_{1,1}, u_{1,2}, u_{2,1}, u_{2,2}$. By assumption, the join-vertices that are not contained in $X$ are dominated and those that are contained in $X$ are root-connected. By definition of $X^\ell_P$, we see that $a_{1,1}, a_{2,1} \in X$, hence it remains to establish the root-connectivity of $a_{1,1}$ and $a_{2,1}$.

  We show that $a_{i,1}$ is root-connected by considering the different cases for $\statemap_{X_P^\ell}(u_{i,1})$. Note that $\statemap_{X_P^\ell}(u_{i,1}) \neq \zero_0$, since $a_{i,1} \in X_P^\ell$. If $\statemap_{X_P^\ell}(u_{i,1}) = \zero_1$, then $a_{i,1}$ is root-connected via the path $a_{i,1}, b_{i,0}, \rvertex$ in $X$. If $\statemap_{X_P^\ell}(u_{i,1}) = \one_0$, then $u_{i,1}$ is root-connected in $X$ via some path that leaves $P + \rvertex$ which we can extend to $a_{i,1}$. Finally, if $\statemap_{X_P^\ell}(u_{i,1}) = \one_1$, then the path $a_{i,1}, c_{i,1}, \rvertex$ exists in $G[X]$. This concludes the proof.
\end{proof}

\subparagraph*{State Transitions.} In the lower bound construction, we again create long paths by repeatedly concatenating the path gadgets $P$. To study how the state can change between two consecutive path gadgets, suppose that we have two copies $P^1$ and $P^2$ of $P$ such that the vertices $u_{2,1}$ and $u_{2,2}$ in $P^1$ are joined to the vertices $u_{1,1}$ and $u_{1,2}$ in $P^2$. We denote the vertices of $P^1$ with a superscript $1$ and the vertices of $P^2$ with a superscript $2$, e.g., $u^1_{2,1}$ refers to the vertex $u_{2,1}$ of $P^1$. Again, suppose that $P^1$ and $P^2$ are embedded as induced subgraphs in a larger graph $G$ with a root vertex $\rvertex$ and that only the vertices $u^1_{1,1}, u^1_{1,2}, u^2_{2,1}, u^2_{2,2}$ and the clique vertices $v^1_\ell, v^2_\ell$, $\ell \in [5]$, have neighbors outside of $P^1 + P^2 + \rvertex$. Furthermore, $X$ denotes a connected dominating set of $G$ with $\rvertex \in X$. 

The previous lemmas will now allow us to conclude that the gadget state can indeed only transition according to the transition order and that the state can remain stable.

\begin{figure}[h]
  \centering
  \begin{tikzpicture}
	\begin{pgfonlayer}{nodelayer}
		\node [style=filled W] (0) at (6, 6.5) {0};
		\node [style=filled W] (1) at (6, 5.5) {1};
		\node [style=filled 10pt] (2) at (6, 4.5) {0};
		\node [style=filled 10pt] (3) at (6, 3.5) {1};
		\node [style=none] (4) at (7, 6.5) {=};
		\node [style=none] (5) at (7, 5.5) {=};
		\node [style=none] (6) at (7, 4.5) {=};
		\node [style=none] (7) at (7, 3.5) {=};
		\node [style=none] (8) at (8, 6.5) {$\mathbf{0}_0$};
		\node [style=none] (9) at (8, 5.5) {$\mathbf{0}_1$};
		\node [style=none] (10) at (8, 4.5) {$\mathbf{1}_0$};
		\node [style=none] (11) at (8, 3.5) {$\mathbf{1}_1$};
		\node [style=filled W] (73) at (-4.5, 4.25) {};
		\node [style=filled W] (74) at (-4.5, 3.5) {1};
		\node [style=filled 10pt] (75) at (-4.5, 4.25) {1};
		\node [style=filled W] (76) at (-4.5, 0.25) {};
		\node [style=filled W] (77) at (-4.5, -0.5) {0};
		\node [style=filled 10pt] (78) at (-4.5, 0.25) {1};
		\node [style=filled W] (79) at (-4.5, -1.75) {};
		\node [style=filled W] (80) at (-4.5, -2.5) {0};
		\node [style=filled 10pt] (81) at (-4.5, -1.75) {0};
		\node [style=filled W] (82) at (-4.5, 1.5) {1};
		\node [style=filled W] (83) at (-4.5, 2.25) {1};
		\node [style=filled W] (84) at (-4.5, -3.75) {1};
		\node [style=filled W] (85) at (-4.5, -4.5) {0};
		\node [style=filled W] (86) at (3.25, 6.5) {};
		\node [style=filled W] (87) at (3.25, 5.75) {1};
		\node [style=filled 10pt] (88) at (3.25, 6.5) {1};
		\node [style=filled W] (89) at (1.75, 6.5) {};
		\node [style=filled W] (90) at (1.75, 5.75) {0};
		\node [style=filled 10pt] (91) at (1.75, 6.5) {1};
		\node [style=filled W] (92) at (0.25, 6.5) {};
		\node [style=filled W] (93) at (0.25, 5.75) {0};
		\node [style=filled 10pt] (94) at (0.25, 6.5) {0};
		\node [style=filled W] (95) at (-1.25, 6.5) {1};
		\node [style=filled W] (96) at (-1.25, 5.75) {1};
		\node [style=filled W] (97) at (-2.75, 6.5) {1};
		\node [style=filled W] (98) at (-2.75, 5.75) {0};
		\node [style=none] (99) at (-3.75, 6.75) {};
		\node [style=none] (100) at (-3.75, -5) {};
		\node [style=none] (101) at (-6, 5) {};
		\node [style=none] (102) at (4, 5) {};
		\node [style=none] (103) at (3.25, 4) {$\times$};
		\node [style=none] (104) at (1.75, 4) {$\times$};
		\node [style=none] (105) at (0.25, 4) {$\times$};
		\node [style=none] (106) at (-1.25, 4) {$\times$};
		\node [style=none] (107) at (-2.75, 4) {$\times$};
		\node [style=none] (108) at (3.25, 0) {$\times$};
		\node [style=none] (109) at (1.75, 0) {$\times$};
		\node [style=none] (110) at (0.25, 0) {$\times$};
		\node [style=none] (111) at (3.25, -2) {$\times$};
		\node [style=none] (112) at (3.25, 2) {$\times$};
		\node [style=none] (113) at (3.25, -4) {$\times$};
		\node [style=none] (114) at (-1.25, 2) {$\times$};
		\node [style=none] (115) at (1.75, -2) {$\times$};
		\node [style=none] (118) at (-3.75, 5) {};
		\node [style=none] (120) at (-5.5, 6.5) {};
		\node [style=none] (121) at (-5.25, 5.5) {$P^1$};
		\node [style=none] (122) at (-4.25, 6.25) {$P^2$};
		\node [style=none] (123) at (7.5, 1.5) {at $P^1$: $(\mathbf{s}^{\ell_1}_3, \mathbf{s}^{\ell_1}_4)$};
		\node [style=none] (124) at (7.5, 0.5) {at $P^2$: $(\mathbf{s}^{\ell_2}_1, \mathbf{s}^{\ell_2}_2)$};
	\end{pgfonlayer}
	\begin{pgfonlayer}{edgelayer}
		\draw [style=thick] (99.center) to (100.center);
		\draw [style=thick] (101.center) to (102.center);
		\draw [style=thick] (120.center) to (118.center);
	\end{pgfonlayer}
\end{tikzpicture}
  \caption{A matrix depicting the possible state transitions between two consecutive path gadgets. The rows are labeled with $(\state_3^{\ell_1}, \state_4^{\ell_1})$, $\ell_1 \in [5]$, and the columns are labeled with $(\state_1^{\ell_2}, \state_2^{\ell_2})$, $\ell_2 \in [5]$. An $\times$ marks possible state transitions in a connected dominating set.}
  \label{fig:cds_cw_transition_table}
\end{figure}

\begin{lem}\label{thm:cds_cw_path_transition}
  If $X$ satisfies $|X \cap P^1| \leq 14$ and $|X \cap P^2| \leq 14$, then $\statemap(X \cap P^1) = \state^{\ell_1}$ and $\statemap(X \cap P^2) = \state^{\ell_2}$ with $\ell_1 \leq \ell_2$. 
  
  Additionally, for each $\ell \in [5]$, the set $X^\ell = X_{P^1}^\ell \cup X_{P^2}^\ell$ dominates or root-connects all inner join vertices, i.e., $\statemap_{X^\ell}(\{u^1_{2,1}, u^1_{2,2}, u^2_{1,1}, u^2_{1,2}\}) \subseteq \{\zero_1, \one_1\}$.
\end{lem}

\begin{proof}
  By \cref{thm:cds_cw_path_gadget_ub}, there are $\ell_1, \ell_2 \in [5]$ such that $\statemap(X \cap P^1) = \state^{\ell_1}$ and $\statemap(X \cap P^2) = \state^{\ell_2}$. It remains to verify that $\ell_1 \leq \ell_2$. The main idea is to consider the states at the inner join vertices $u^1_{2,1}, u^1_{2,2}$, and $u^2_{1,1}$, $u^2_{1,2}$ and notice that at least one of these four vertices is not dominated or not root-connected whenever $\ell_1 > \ell_2$. Recall that these inner join vertices are only adjacent to vertices inside $P^1 + P^2$ and $N(\{u^1_{2,1}, u^1_{2,2}\}) \cap V(P^2) = \{u^2_{1,1}, u^2_{1,2}\}$ and $N(\{u^2_{1,1}, u^2_{1,2}\}) \cap V(P^1) = \{u^1_{2,1}, u^1_{2,2}\}$. \cref{fig:cds_cw_transition_table} depicts the possible state transitions.
  
  First, if $\ell_1 \in \{3,4,5\}$ and $\ell_2 \in [2]$, then $\statemap_{X \cap P^1}(u^1_{2,2}) = \zero_0$ and $\statemap_{X \cap P^2}(\{u^2_{1,1}, u^2_{1,2}\})\subseteq \{\zero_0, \zero_1\}$. Therefore, the vertex $u^1_{2,2}$ cannot be dominated by $X$ in this case.
  
  Secondly, if $\ell_1 \in \{4,5\}$ and $\ell_2 = 3$, then $\one_1 \notin \statemap_{X \cap P^1}(\{u^1_{2,1}, u^1_{2,2}\})$ and $\statemap_{X \cap P^2}(u^2_{1,1}) = \one_0$. Therefore, the vertex $u^2_{1,1}$ cannot be root-connected in $X$ in this case.
  
  Lastly, if $(\ell_1, \ell_2) \in \{(2,1),(5,4)\}$, then we have $\statemap_{X \cap P^1}(\{u^1_{2,1}, u^1_{2,2}\}) \subseteq \{\zero_0, \zero_1\}$ and $\statemap_{X \cap P^2}(u^2_{1,2}) = \zero_0$. Therefore, the vertex $u^2_{1,2}$ cannot be dominated by $X$ in this case.
  
  For the second part of the lemma, fix some $\ell \in [5]$. By \cref{thm:cds_cw_state_exists}, we have that $\statemap(X^\ell \cap P^1) = \statemap(X^\ell \cap P^2) = \state^{\ell}$. We distinguish based on $\ell \in [5]$; see also \cref{fig:cds_cw_equal_transition}:
  \begin{itemize}
    \item $\ell = 1$: $u^2_{1,2}$ is dominated by $u^1_{2,1}$.
    \item $\ell = 2$: all four vertices are already dominated by \cref{thm:cds_cw_state_exists}.
    \item $\ell = 3$: $u^1_{2,1}$ root-connects $u^2_{1,1}$ and dominates $u^2_{1,2}$; $u^2_{1,1}$ dominates $u^1_{2,2}$.
    \item $\ell = 4$: reverse situation of $\ell = 3$.
    \item $\ell = 5$: reverse situation of $\ell = 1$.
  \end{itemize}
  This finishes the proof of the second part.
\end{proof}

\begin{figure}[h]
  \centering
  \begin{tikzpicture}
	\begin{pgfonlayer}{nodelayer}
		\node [style=filled W] (0) at (-10, -3.5) {0};
		\node [style=filled W] (1) at (-4, -3.5) {1};
		\node [style=filled 10pt] (2) at (2, -3.5) {0};
		\node [style=filled 10pt] (3) at (8, -3.5) {1};
		\node [style=none] (4) at (-9, -3.5) {=};
		\node [style=none] (5) at (-3, -3.5) {=};
		\node [style=none] (6) at (3, -3.5) {=};
		\node [style=none] (7) at (9, -3.5) {=};
		\node [style=none] (8) at (-8, -3.5) {$\mathbf{0}_0$};
		\node [style=none] (9) at (-2, -3.5) {$\mathbf{0}_1$};
		\node [style=none] (10) at (4, -3.5) {$\mathbf{1}_0$};
		\node [style=none] (11) at (10, -3.5) {$\mathbf{1}_1$};
		\node [style=filled W] (12) at (11.25, 0.75) {1};
		\node [style=filled W] (13) at (11.25, -0.75) {0};
		\node [style=filled W] (14) at (12.75, 0.75) {};
		\node [style=filled W] (15) at (12.75, -0.75) {1};
		\node [style=none] (16) at (10, 0.75) {$u^1_{2,1}$};
		\node [style=none] (17) at (10, -0.75) {$u^1_{2,2}$};
		\node [style=none] (18) at (14, 0.75) {$u^2_{1,1}$};
		\node [style=none] (19) at (14, -0.75) {$u^2_{1,2}$};
		\node [style=none] (20) at (-14, 2.5) {Case 1:};
		\node [style=filled W] (21) at (5.25, 0.75) {};
		\node [style=filled W] (22) at (5.25, -0.75) {0};
		\node [style=filled W] (23) at (6.75, 0.75) {};
		\node [style=filled W] (24) at (6.75, -0.75) {0};
		\node [style=none] (25) at (4, 0.75) {$u^1_{2,1}$};
		\node [style=none] (26) at (4, -0.75) {$u^1_{2,2}$};
		\node [style=none] (27) at (8, 0.75) {$u^2_{1,1}$};
		\node [style=none] (28) at (8, -0.75) {$u^2_{1,2}$};
		\node [style=none] (29) at (-8, 2.5) {Case 2:};
		\node [style=filled W] (30) at (-0.75, 0.75) {};
		\node [style=filled W] (31) at (-0.75, -0.75) {0};
		\node [style=filled W] (32) at (0.75, 0.75) {};
		\node [style=filled W] (33) at (0.75, -0.75) {0};
		\node [style=none] (34) at (-2, 0.75) {$u^1_{2,1}$};
		\node [style=none] (35) at (-2, -0.75) {$u^1_{2,2}$};
		\node [style=none] (36) at (2, 0.75) {$u^2_{1,1}$};
		\node [style=none] (37) at (2, -0.75) {$u^2_{1,2}$};
		\node [style=none] (38) at (-2, 2.5) {Case 3:};
		\node [style=filled W] (39) at (-6.75, 0.75) {1};
		\node [style=filled W] (40) at (-6.75, -0.75) {1};
		\node [style=filled W] (41) at (-5.25, 0.75) {1};
		\node [style=filled W] (42) at (-5.25, -0.75) {1};
		\node [style=none] (43) at (-8, 0.75) {$u^1_{2,1}$};
		\node [style=none] (44) at (-8, -0.75) {$u^1_{2,2}$};
		\node [style=none] (45) at (-4, 0.75) {$u^2_{1,1}$};
		\node [style=none] (46) at (-4, -0.75) {$u^2_{1,2}$};
		\node [style=none] (47) at (4, 2.5) {Case 4:};
		\node [style=filled W] (48) at (-12.75, 0.75) {};
		\node [style=filled W] (49) at (-12.75, -0.75) {1};
		\node [style=filled W] (50) at (-11.25, 0.75) {1};
		\node [style=filled W] (51) at (-11.25, -0.75) {0};
		\node [style=none] (52) at (-14, 0.75) {$u^1_{2,1}$};
		\node [style=none] (53) at (-14, -0.75) {$u^1_{2,2}$};
		\node [style=none] (54) at (-10, 0.75) {$u^2_{1,1}$};
		\node [style=none] (55) at (-10, -0.75) {$u^2_{1,2}$};
		\node [style=none] (56) at (10, 2.5) {Case 5:};
		\node [style=filled 10pt] (57) at (5.25, 0.75) {0};
		\node [style=filled 10pt] (58) at (-0.75, 0.75) {1};
		\node [style=filled 10pt] (59) at (-12.75, 0.75) {1};
		\node [style=filled 10pt] (60) at (12.75, 0.75) {1};
		\node [style=filled 10pt] (61) at (6.75, 0.75) {1};
		\node [style=filled 10pt] (62) at (0.75, 0.75) {0};
		\node [style=none] (63) at (12, 2) {};
		\node [style=none] (64) at (12, -2) {};
		\node [style=none] (65) at (6, 2) {};
		\node [style=none] (66) at (6, -2) {};
		\node [style=none] (67) at (0, 2) {};
		\node [style=none] (68) at (0, -2) {};
		\node [style=none] (69) at (-6, 2) {};
		\node [style=none] (70) at (-6, -2) {};
		\node [style=none] (71) at (-12, 2) {};
		\node [style=none] (72) at (-12, -2) {};
	\end{pgfonlayer}
	\begin{pgfonlayer}{edgelayer}
		\draw (12) to (14);
		\draw (14) to (13);
		\draw (13) to (15);
		\draw (15) to (12);
		\draw (21) to (23);
		\draw (23) to (22);
		\draw (22) to (24);
		\draw (24) to (21);
		\draw (30) to (32);
		\draw (32) to (31);
		\draw (31) to (33);
		\draw (33) to (30);
		\draw (39) to (41);
		\draw (41) to (40);
		\draw (40) to (42);
		\draw (42) to (39);
		\draw (48) to (50);
		\draw (50) to (49);
		\draw (49) to (51);
		\draw (51) to (48);
		\draw [style=dotted] (63.center) to (64.center);
		\draw [style=dotted] (65.center) to (66.center);
		\draw [style=dotted] (67.center) to (68.center);
		\draw [style=dotted] (69.center) to (70.center);
		\draw [style=dotted] (71.center) to (72.center);
	\end{pgfonlayer}
\end{tikzpicture}
  \caption{Case distinction to show $\statemap_{X^\ell}(\{u^1_{2,1}, u^1_{2,2}, u^2_{1,1}, u^2_{1,2}\}) \subseteq \{\zero_1, \one_1\}$.}
  \label{fig:cds_cw_equal_transition}
\end{figure}

We say that a \emph{cheat} occurs when $\ell_1 < \ell_2$. When creating arbitrary long concatenation of the path gadgets, \cref{thm:cds_cw_path_transition} tells us that at most $|\states| - 1 = 4 = \Oh(1)$ cheats may occur on such a path, which allows us to find a \emph{cheat-free region} as outlined before.

\subsubsection{Complete Construction}

\renewcommand{\nregions}{{4\ngrps\grpsize + 1}}
\renewcommand{\ncolumns}{{\nclss(\nregions)}}
\renewcommand{\sequence}{\mathbf{h}}

\subparagraph*{Setup.} 
Assume that \CDS can be solved in time $\Oh^*((5 - \eps)^{\nlabels})$ for some $\eps > 0$ when given a linear $\nlabels$-expression. Given a \SAT-instance $\formula$ with $\nvars$ variables and $\nclss$ clauses, we construct an equivalent \CDS instance with linear clique-width approximately $\nvars \log_5(2)$ so that the existence of such an algorithm for \CDS would imply that the \SETH is false. 

We pick an integer $\vgrpsize$ only depending on $\eps$; the precise choice of $\vgrpsize$ will be discussed at a later point. The variables of $\formula$ are partitioned into groups of size at most $\vgrpsize$, resulting in $\ngrps = \lceil \nvars / \vgrpsize \rceil$ groups. Furthermore, we pick the smallest integer $\grpsize$ that satisfies $5^\grpsize \geq 2^\vgrpsize$, i.e., $\grpsize = \lceil \log_5(2^\vgrpsize) \rceil$. We now begin with the construction of the \CDS instance $(G = G(\formula, \vgrpsize), \budget)$.

We create the root vertex $\rvertex$ and attach a leaf $\rvertex'$ which forces $\rvertex$ into any connected dominating set.
For every group $i \in [\ngrps]$, we create $\grpsize$ long path-like gadgets $P^{i, j}$, $j \in [\grpsize]$, where each $P^{i, j}$ consists of $\ncolumns$ copies $P^{i, j, \ell}$, $\ell \in [\ncolumns]$, of the path gadget $P$ and consecutive copies are connected by a join. More precisely, the vertices in some $P^{i, j, \ell}$ inherit their names from the generic path gadget $P$ and the superscript of $P^{i, j, \ell}$ and for every $i \in [\ngrps]$, $j \in [\grpsize]$, $\ell \in [\ncolumns - 1]$, the vertices $\{u^{i,j,\ell}_{2,1}, u^{i,j,\ell}_{2,2}\}$ are joined to the vertices $\{u^{i,j,\ell + 1}_{1,1}, u^{i,j,\ell + 1}_{1,2}\}$. The ends of each path $P^{i, j}$, namely the vertices $u^{i, j, 1}_{1,1}, u^{i, j, 1}_{1,2}, u^{i, j, \ncolumns}_{2,1}, u^{i, j, \ncolumns}_{2,2}$ are made adjacent to the root $\rvertex$. 

\begin{figure}[h]
  \centering
  \begin{tikzpicture}
	\begin{pgfonlayer}{nodelayer}
		\node [style=tiny rectangle] (0) at (4, 8) {};
		\node [style=tiny rectangle] (1) at (4, 4) {};
		\node [style=tiny rectangle] (2) at (4, -0.75) {};
		\node [style=none] (3) at (4, 2.25) {$\vdots$};
		\node [style=filled] (4) at (0, 8) {};
		\node [style=filled] (5) at (0, 4) {};
		\node [style=filled] (6) at (0, -0.75) {};
		\node [style=filled] (7) at (1, 0.25) {};
		\node [style=filled] (8) at (1, 5) {};
		\node [style=filled] (9) at (1, 9) {};
		\node [style=filled] (10) at (10, 5) {};
		\node [style=filled] (11) at (12.25, 5) {};
		\node [style=none] (18) at (10.25, 5.75) {$z^{i,\ell}$};
		\node [style=none] (19) at (12.25, 5.75) {$\bar{z}^{i, \ell}$};
		\node [style=none] (20) at (0, -1.5) {$x^{i, \ell}_{(5,5,\ldots)}$};
		\node [style=none] (21) at (1.25, 1) {$\bar{x}^{i, \ell}_{(5,5,\ldots)}$};
		\node [style=none] (22) at (0, 3.25) {$x^{i, \ell}_{(2,1,\ldots)}$};
		\node [style=none] (23) at (1.25, 5.75) {$\bar{x}^{i, \ell}_{(2,1,\ldots)}$};
		\node [style=none] (24) at (-0.25, 7.25) {$x^{i, \ell}_{(1,1,\ldots)}$};
		\node [style=none] (25) at (1, 9.75) {$\bar{x}^{i, \ell}_{(1,1,\ldots)}$};
		\node [style=none] (26) at (4, 8.75) {$y^{i, \ell}_{(1,1,\ldots)}$};
		\node [style=none] (27) at (4, 4.75) {$y^{i, \ell}_{(2,1,\ldots)}$};
		\node [style=none] (28) at (4, 0) {$y^{i, \ell}_{(5,5,\ldots)}$};
		\node [style=tiny rectangle] (29) at (-8, 9) {};
		\node [style=tiny rectangle] (30) at (-8, 8) {};
		\node [style=tiny rectangle] (31) at (-8, 7) {};
		\node [style=tiny rectangle] (32) at (-8, 6) {};
		\node [style=tiny rectangle] (33) at (-8, 5) {};
		\node [style=none] (34) at (-10, 10) {};
		\node [style=none] (35) at (-7, 10) {};
		\node [style=none] (36) at (-10, 4) {};
		\node [style=none] (37) at (-7, 4) {};
		\node [style=none] (38) at (-9, 9) {$v_1^{i,1,\ell}$};
		\node [style=none] (39) at (-9, 8) {$v_2^{i,1,\ell}$};
		\node [style=none] (40) at (-9, 7) {$v_3^{i,1,\ell}$};
		\node [style=none] (41) at (-9, 6) {$v_4^{i,1,\ell}$};
		\node [style=none] (42) at (-9, 5) {$v_5^{i,1,\ell}$};
		\node [style=tiny rectangle] (43) at (-8, 2) {};
		\node [style=tiny rectangle] (44) at (-8, 1) {};
		\node [style=tiny rectangle] (45) at (-8, 0) {};
		\node [style=tiny rectangle] (46) at (-8, -1) {};
		\node [style=tiny rectangle] (47) at (-8, -2) {};
		\node [style=none] (48) at (-10, 3) {};
		\node [style=none] (49) at (-7, 3) {};
		\node [style=none] (50) at (-10, -3) {};
		\node [style=none] (51) at (-7, -3) {};
		\node [style=none] (52) at (-9, 2) {$v_1^{i,2,\ell}$};
		\node [style=none] (53) at (-9, 1) {$v_2^{i,2,\ell}$};
		\node [style=none] (54) at (-9, 0) {$v_3^{i,2,\ell}$};
		\node [style=none] (55) at (-9, -1) {$v_4^{i,2,\ell}$};
		\node [style=none] (56) at (-9, -2) {$v_5^{i,2,\ell}$};
		\node [style=none] (57) at (-8.5, -4) {$\vdots$};
		\node [style=none] (58) at (0, 2.25) {$\vdots$};
		\node [style=filled] (59) at (8, -2.75) {};
		\node [style=filled] (60) at (10, -2.75) {};
		\node [style=none] (61) at (8.5, -2.25) {$o^\ell$};
		\node [style=none] (62) at (10.5, -2.25) {$\bar{o}^\ell$};
		\node [style=none] (63) at (7.5, -2) {};
		\node [style=none] (64) at (7.25, -2.25) {};
		\node [style=tiny rectangle] (65) at (8, 9) {};
		\node [style=none] (66) at (9, 9) {$=$};
		\node [style=none] (67) at (11.5, 9) {adjacent to $\hat{r}$};
		\node [style=none] (68) at (-8.5, 10.5) {subdivided $K_5$};
	\end{pgfonlayer}
	\begin{pgfonlayer}{edgelayer}
		\draw [style=thin] (9) to (4);
		\draw [style=thin] (8) to (5);
		\draw [style=thin] (7) to (6);
		\draw [style=thin] (6) to (2);
		\draw [style=thin] (5) to (1);
		\draw [style=thin] (4) to (0);
		\draw [style=thin] (11) to (10);
		\draw [style=thin, in=0, out=150, looseness=0.50] (10) to (0);
		\draw [style=thin, in=0, out=-180, looseness=0.50] (10) to (1);
		\draw [style=thin, in=0, out=-90, looseness=0.50] (10) to (2);
		\draw [style=dashed cyan thick] (37.center) to (35.center);
		\draw [style=dashed cyan thick] (35.center) to (34.center);
		\draw [style=dashed cyan thick] (34.center) to (36.center);
		\draw [style=dashed cyan thick] (36.center) to (37.center);
		\draw [style=dashed cyan thick] (51.center) to (49.center);
		\draw [style=dashed cyan thick] (49.center) to (48.center);
		\draw [style=dashed cyan thick] (48.center) to (50.center);
		\draw [style=dashed cyan thick] (50.center) to (51.center);
		\draw [style=thin, in=180, out=0] (29) to (4);
		\draw [style=thin, in=0, out=-180, looseness=0.75] (4) to (43);
		\draw [style=thin, in=0, out=-180, looseness=0.50] (5) to (30);
		\draw [style=thin, in=0, out=-180] (5) to (43);
		\draw [style=thin, in=0, out=-180, looseness=0.50] (6) to (33);
		\draw [style=thin] (6) to (47);
		\draw [style=thin, in=90, out=-30, looseness=0.50] (1) to (59);
		\draw [style=thin] (59) to (60);
		\draw [style=thin dashed] (63.center) to (59);
		\draw [style=thin dashed] (64.center) to (59);
	\end{pgfonlayer}
\end{tikzpicture}
  \caption{Decoding and clause gadget for \CDS.}
  \label{fig:cds_cw_decoding_gadget}
\end{figure}

For every group $i \in [\ngrps]$ and column $\ell \in [\ncolumns]$, we create a \emph{decoding gadget} $D^{i,\ell}$ in the same style as Cygan et al.~\cite{CyganNPPRW11} for \CVC parameterized by pathwidth. Every variable group $i$ has at most $2^\vgrpsize$ possible truth assignments and by choice of $\grpsize$ we have that $5^\grpsize \geq 2^\vgrpsize$, so we can find an injective mapping $\embedding \colon \{0,1\}^\vgrpsize \rightarrow [5]^\grpsize$ which assigns to each truth assignment $\tassign \in \{0,1\}^\vgrpsize$ a sequence $\embedding(\tassign) \in [5]^\grpsize$. For each sequence $\sequence = (h_1, \ldots, h_\grpsize) \in [5]^\grpsize$, we create vertices $x^{i, \ell}_\sequence$, $\bar{x}^{i, \ell}_\sequence$, $y^{i, \ell}_\sequence$ and edges $\{x^{i, \ell}_\sequence, \bar{x}^{i, \ell}_\sequence\}$, $\{x^{i, \ell}_\sequence, y^{i, \ell}_\sequence\}$, $\{y^{i, \ell}_\sequence, \rvertex\}$. Furthermore, we add the edge $\{x^{i, \ell}_\sequence, v^{i, j, \ell}_{h_j}\}$ for all $\sequence = (h_1, \ldots, h_\grpsize) \in [5]^\grpsize$ and $j \in [\grpsize]$. Finally, we create two adjacent vertices $z^{i, \ell}$ and $\bar{z}^{i, \ell}$ and edges $\{z^{i, \ell}, y^{i, \ell}_\sequence\}$ for all $\sequence \in [5]^\grpsize$. Each decoding gadget $D^{i, \ell}$ together with the adjacent path gadgets $P^{i,j,\ell}$, $j \in [\grpsize]$, forms the \emph{block} $B^{i,\ell}$, see \cref{fig:cds_cw_schematic} for a high-level depiction of the connections between different blocks.

Lastly, we construct the \emph{clause gadgets}.
We number the clauses of $\formula$ by $C_0, \ldots, C_{\nclss - 1}$. For every column $\ell \in [\ncolumns]$, we create an adjacent pair of vertices $o^{\ell}$ and $\bar{o}^{\ell}$. Let $\ell' \in [0, \nclss - 1]$ be the remainder of $(\ell - 1)$ modulo $\nclss$; for all $i \in [\ngrps]$, $\sequence \in \embedding(\{0,1\}^\vgrpsize)$ such that $\embedding^{-1}(\sequence)$ is a truth assignment for variable group $i$ satisfying clause $C_{\ell'}$, we add the edge $\{o^{\ell}, y^{i, \ell}_\sequence\}$. See \cref{fig:cds_cw_decoding_gadget} for a depiction of the decoding and clause gadget.

\begin{figure}[h]
  \centering
  \scalebox{.9}{\input{pictures/cvc_modtw_schematic.tikz}}
  \caption{High-level construction for \CDS. Every edge between two blocks $B^{i,\ell}$ and $B^{i,\ell+1}$ represents a join. The regions $R^\gamma$ are important for the proof of \cref{thm:cds_cw_sol_to_sat}.}
  \label{fig:cds_cw_schematic}
\end{figure}

\begin{lem}\label{thm:cds_cw_sat_to_sol}
  If $\formula$ is satisfiable, then there exists a connected dominating set $X$ of $G = G(\formula, \vgrpsize)$ of size $|X| \leq (14\ngrps\grpsize + (5^\grpsize + 2)\ngrps + 1)\ncolumns + 1 = \budget$.
\end{lem}

\begin{proof}
  Let $\tassign$ be a satisfying truth assignment of $\formula$ and let $\tassign^i$ denote the restriction of $\tassign$ to the $i$-th variable group for every $i \in [\ngrps]$ and let $\embedding(\tassign^i) = \sequence^i = (h^i_1, \ldots, h^i_\grpsize) \in [5]^{\grpsize}$ be the corresponding sequence. 
  The connected dominating set is given by 
  \begin{equation*}
    X = \{\rvertex\} \cup \bigcup_{\ell \in [\ncolumns]} \left(\{o^\ell\} \cup \bigcup_{i \in [\ngrps]} \left(\{y^{i, \ell}_{\sequence^i}, z^{i,\ell}\} \cup \bigcup_{\sequence \in [5]^\grpsize} \{x^{i,\ell}_\sequence\} \cup \bigcup_{j \in [\grpsize]} X_{P^{i,j,\ell}}^{h^i_j} \right) \right),
  \end{equation*}
  where $X_{P^{i,j,\ell}}^{h^i_j}$ refers to the sets given by \cref{thm:cds_cw_state_exists}. 
  
  Clearly, $|X| = \budget$ as $|X_{P^{i,j,\ell}}^{h^i_j}| = 14$ for all $i,j, \ell$, so it remains to prove that $X$ is a connected dominating set. We begin by considering the path gadgets. First, notice that we have
  \begin{equation*}
    \bigcup_{i \in [\ngrps]} \bigcup_{j \in [\grpsize]} \{u^{i,j,1}_{1,1}, u^{i,j,1}_{1,2}, u^{i,j,\ncolumns}_{2,1}, u^{i,j,\ncolumns}_{2,2}\} \subseteq N(\rvertex),
  \end{equation*}
  hence the vertices at the ends of the paths $P^{i,j}$ are dominated or root-connected in $X$. Next, we invoke \cref{thm:cds_cw_state_exists} and the second part of \cref{thm:cds_cw_path_transition} to see that all vertices on the path gadgets are root-connected in $X$ or dominated.
  
  Fix $i \in [\ngrps]$, $\ell \in [\ncolumns]$, and consider the corresponding decoding gadget. Since $z^{i, \ell} \in X$ and $x^{i,\ell}_\sequence \in X$ for all $\sequence \in [5]^\grpsize$, all vertices in the decoding gadget are dominated by $X$. Furthermore, since $o^{\ell} \in X$, all vertices inside the clause gadget are dominated by $X$. Hence, $X$ has to be a dominating set of $G$.
  
  It remains to prove that the vertices in the decoding and clause gadgets that belong to $X$ are also root-connected. Again, fix $i \in [\ngrps]$, $\ell \in [\ncolumns]$, and $\sequence = (h_1, \ldots, h_\grpsize) \in [5]^\grpsize \setminus \{\sequence^i\}$. Since $\sequence \neq \sequence^i$, there is some $j \in [\grpsize]$ such that $h_j \neq h_j^i$ and hence $v^{i,j,\ell}_{h_j} \in X$ by \cref{thm:cds_cw_state_exists} which connects $x_\sequence^{i, \ell}$ to the root $\rvertex$. The vertices $x_{\sequence^i}^{i, \ell}$ and $z^{i, \ell}$ are root-connected via $y^{i, \ell}_{\sequence^i} \in X$.
  
  We conclude by showing that $o^\ell$ is root-connected for all $\ell \in [\ncolumns]$. Since $\tassign$ is a satisfying truth assignment of $\formula$, there is some variable group $i \in [\ngrps]$ such that $\tassign^i$ already satisfies clause $C_{\ell'}$, where $\ell'$ is the remainder of $(\ell - 1)$ modulo $\nclss$. By construction of $G$ and $X$, the vertex $y^{i, \ell}_{\sequence^i} \in X$ is adjacent to $o^\ell$, since $\embedding(\tassign^i) = \sequence^i$, and connects $o^\ell$ to the root $\rvertex$. This shows that all vertices of $X$ are root-connected, so $G[X]$ has to be connected.
\end{proof}

\begin{lem}\label{thm:cds_cw_sol_to_sat}
  If there exists a connected dominating set $X$ of $G = G(\formula, \vgrpsize)$ of size $|X| \leq (14\ngrps\grpsize + (5^\grpsize + 2)\ngrps + 1)\ncolumns + 1 = \budget$, then $\formula$ is satisfiable.
\end{lem}

\begin{proof}
  We begin by arguing that $X$ has to satisfy $|X| = \budget$. First, we must have that $\rvertex \in X$, because $\rvertex$ has a neighbor of degree 1. By \cref{thm:cds_cw_path_gadget_lb}, we have that $|X \cap P^{i,j,\ell}| \geq 14$ for all $i \in [\ngrps]$, $j \in [\grpsize]$, $\ell \in [\ncolumns]$. In every decoding gadget, i.e., one for every $i \in [\ngrps]$ and $\ell \in [\ncolumns]$, the set $\{z^{i,\ell}\} \cup \bigcup_{\sequence \in [5]^\grpsize} x_\sequence^{i, \ell}$ has to be contained in $X$, since every vertex in this set has a neighbor of degree 1. Furthermore, to connect $z^{i,\ell}$ to $\rvertex$, at least one of the vertices $y^{i,\ell}_\sequence$, $\sequence \in [5]^\grpsize$, has to be contained in $X$. Hence, $X$ must contain at least $5^\grpsize + 2$ vertices per decoding gadget. Lastly, $o^\ell \in X$ for all $\ell \in [\ncolumns]$, since $o^\ell$ has a neighbor of degree 1. Since we have only considered disjoint vertex sets, this shows that $|X| = \budget$ and all of the previous inequalities have to be tight.
  
  By \cref{thm:cds_cw_path_gadget_ub}, we know that $X$ assumes one of the five possible states on each $P^{i,j,\ell}$. Fix some $P^{i,j} = \bigcup_{\ell \in [\ncolumns]} P^{i,j,\ell}$ and note that due to \cref{thm:cds_cw_path_transition} the state can change at most four times along $P^{i,j}$. Such a state change is called a \emph{cheat}. Let $\gamma \in [0, 4 \ngrps \grpsize]$ and define the $\gamma$-th \emph{region} $R^\gamma = \bigcup_{i \in [\ngrps]} \bigcup_{j \in [\grpsize]} \bigcup_{\ell = \gamma \nclss + 1}^{(\gamma + 1) \nclss} P^{i,j,\ell}$. Since there are $\nregions$ regions, there is at least one region $R^\gamma$ such that no cheat occurs in $R^\gamma$. We will consider this region for the remainder of the proof and read off a satisfying truth assignment from this region.
  
  For $i \in [\ngrps]$, define $\sequence^{i} = (h^i_1, \ldots, h^i_\grpsize) \in [5]^\grpsize$ such that $v^{i,j,\gamma \nclss + 1}_{h^i_j} \notin X$ for all $j \in [\grpsize]$; this is well-defined by \cref{thm:cds_cw_path_gadget_ub}. Since $R^\gamma$ does not contain any cheats, the definition of $\sequence^i$ is independent of which column $\ell \in [\gamma \nclss + 1, (\gamma + 1)\nclss]$ we consider. For every $i \in [\ngrps]$ and $\ell \in [\gamma \nclss + 1, (\gamma + 1)\nclss]$, we claim that $y^{i, \ell}_\sequence \in X$ if and only if $\sequence = \sequence^i$. We have already established that for every $i$ and $\ell$, there is exactly one $\sequence$ such that $y^{i, \ell}_\sequence \in X$. Consider the vertex $x^{i, \ell}_{\sequence^i} \in X$, its neighbors in $G$ are $v^{i, 1, \ell}_{h^i_1}, v^{i, 2, \ell}_{h^i_2}, \ldots, v^{i, \grpsize, \ell}_{h^i_\grpsize}$, $\bar{x}^{i,\ell}_{\sequence^i}$, and $y^{i,\ell}_{\sequence^i}$. By construction of $\sequence^i$ and the tight allocation of the budget, we have $(N_G(x^{i, \ell}_{\sequence^i}) \setminus \{y^{i, \ell}_{\sequence^i}\}) \cap X = \emptyset$. Therefore, $X$ has to include $y^{i, \ell}_{\sequence^i}$ to connect $x^{i, \ell}_{\sequence^i}$ to the root $\rvertex$. This shows the claim.
  
  For $i \in [\ngrps]$, we define the truth assignment $\tassign^i$ for variable group $i$ by taking an arbitrary truth assignment if $\sequence^i \notin \embedding(\{0,1\}^\vgrpsize)$ and setting $\tassign^i = \embedding^{-1}(\sequence^i)$ otherwise. By setting $\tassign = \bigcup_{i \in [\ngrps]} \tassign^i$ we obtain a truth assignment for all variables and we claim that $\tassign$ satisfies $\formula$. Consider some clause $C_{\ell'}$, $\ell' \in [0, \nclss - 1]$, and let $\ell = \gamma \nclss + \ell' + 1$. We have already argued that $o^\ell \in X$ and to connect $o^\ell$ to the root $\rvertex$, there has to be some $y^{i, \ell}_{\sequence} \in N_G(o^\ell) \cap X$. By the previous claim, $\sequence = \sequence^i$ for some $i \in [\ngrps]$ and therefore $\tassign^i$, and also $\tassign$, satisfy clause $C_{\ell'}$ due to the construction of $G$. Because the choice of $C_{\ell'}$ was arbitrary, $\tassign$ has to be a satisfying assignment of $\formula$.
\end{proof}

\begin{lem}\label{thm:cds_cw_bound}
  The constructed graph $G = G(\formula, \vgrpsize)$ has $\lcw(G) \leq \ngrps \grpsize + 3 \cdot 5^\grpsize + \Oh(1)$ and a linear clique-expression of this width can be constructed in polynomial time.
\end{lem}

\begin{proof}
  We will describe how to construct a linear clique-expression for $G$ of width $\ngrps \grpsize + 3 \cdot 5^\grpsize + \Oh(1)$. The clique-expression will use one path label $(i,j)$ for every long path $P^{i,j}$, $i \in [\ngrps]$, $j \in [\grpsize]$, temporary decoding labels for every vertex of a decoding gadget, i.e.\ $3 \cdot 5^\grpsize + 2$ many, temporary gadget labels for every vertex of a path gadget, two temporary clause labels for the clause gadget, one label for the root vertex and a trash label.
  
  \begin{algorithm}
    Introduce $\rvertex$ with root label and $\rvertex'$ with trash label and add the edge $\{\rvertex, \rvertex'\}$\;
    \For{$\ell \in [\ncolumns]$}
    {
      Introduce $o^\ell$ and $\bar{o}^\ell$ and the edge $\{o^\ell, \bar{o}^\ell\}$ with the clause labels\;
      \For{$i \in [\ngrps]$}
      {
        Build $D^{i,\ell}$ using decoding labels and add edges to $o^\ell$\;
        \For{$j \in [\grpsize]$}
        {
          Build $P^{i,j,\ell}$ using gadget labels and add edges to $D^{i,j}$\;
          Join $u^{i,j,\ell}_{1,1}$ and $u^{i,j,\ell}_{1,2}$ to path label $(i,j)$ if $\ell > 1$ and to root label otherwise\;
          Relabel path label $(i,j)$ to trash label\;
          Relabel $u^{i,j,\ell}_{2,1}$ and $u^{i,j,\ell}_{2,2}$ to path label $(i,j)$\;
          Relabel all other gadget labels to the trash label\;
        }
        Relabel all decoding labels to the trash label\;
      }
      Relabel all clause labels to the trash label\;
    }
    \caption{Constructing a linear clique-expression for $G$.}
    \label{algo:cds_cw_expression}
  \end{algorithm}
  The construction of the clique-expression is described in \cref{algo:cds_cw_expression} and the idea is to proceed column by column and group by group in each column. By reusing the temporary labels, we keep the total number of labels small. The maximum number of labels used simultaneously occurs in line 7 and is $\ngrps \grpsize + (3 \cdot 5^\grpsize + 2) + \Oh(1)$. This concludes the proof.
\end{proof}

\begin{thm}
  There is no algorithm that solves \CDS, given a linear $k$-expression, in time $\Oh^*((5 - \eps)^k)$ for some $\eps > 0$, unless \CNFSETH fails.
\end{thm}

\begin{proof}
  Assume that there exists an algorithm $\algo$ that solves \CDS in time $\Oh^*((5 - \eps)^k)$ for some $\eps > 0$ given a linear $k$-expression. Given $\vgrpsize$, we define $\delta_1 < 1$ such that $(5 - \eps)^{\log_5(2)} = 2^{\delta_1}$ and $\delta_2$ such that $(5 - \eps)^{1 / \vgrpsize} = 2^{\delta_2}$. By picking $\vgrpsize$ large enough, we can ensure that $\delta = \delta_1 + \delta_2 < 1$. We will show how to solve \SAT  using $\algo$ in time $\Oh^*(2^{\delta \nvars})$, where $\nvars$ is the number of variables, thus contradicting \CNFSETH.
  
  Given a \SAT instance $\formula$, we construct $G = G(\formula, \vgrpsize)$ and the linear clique-expression from \cref{thm:cds_cw_bound} in polynomial time, note that we have $\vgrpsize = \Oh(1)$ and hence $\grpsize = \Oh(1)$; recall $\grpsize = \lceil \log_5(2^\vgrpsize) \rceil$. We then run $\algo$ on $G$ and return its answer. This is correct by \cref{thm:cds_cw_sat_to_sol} and \cref{thm:cds_cw_sol_to_sat}. Due to \cref{thm:cds_cw_bound}, the running time is
  \begin{alignat*}{6}
    \phantom{\leq} \quad & \Oh^*\left( (5 - \eps)^{\ngrps \grpsize + 3 \cdot 5^\grpsize + \Oh(1)} \right)
    & \,\,\leq\,\, & \Oh^*\left( (5 - \eps)^{\ngrps \grpsize} \right) 
    & \,\,\leq\,\, & \Oh^*\left( (5 - \eps)^{\lceil \frac{\nvars}{\vgrpsize} \rceil \grpsize} \right) \\
    \leq \quad & \Oh^*\left( (5 - \eps)^{\frac{\nvars}{\vgrpsize} \grpsize} \right) 
    & \,\,\leq\,\, & \Oh^*\left( (5 - \eps)^{\frac{\nvars}{\vgrpsize} \lceil \log_5(2^\vgrpsize) \rceil} \right) 
    & \,\,\leq\,\, & \Oh^*\left( (5 - \eps)^{\frac{\nvars}{\vgrpsize} \log_5(2^\vgrpsize)} (5 - \eps)^{\frac{\nvars}{\vgrpsize}} \right)\\
    \leq \quad & \Oh^*\left( 2^{\delta_1 \vgrpsize \frac{\nvars}{\vgrpsize}} 2^{\delta_2 \nvars} \right)
    & \,\,\leq\,\, & \Oh^*\left( 2^{(\delta_1 + \delta_2) \nvars} \right)
    & \,\,\leq\,\, & \Oh^*\left( 2^{\delta \nvars} \right),
  \end{alignat*}
  hence completing the proof.  
\end{proof}

\section{Conclusion and Open Problems}

We have provided the first tight results under SETH for connectivity problems parameterized by clique-width, namely the problems \CVC and \CDS. For several important benchmark problems such as \ST, \COCT, and \FVS, we are not able to achieve tight results with the current techniques. For \ST, our algorithmic techniques readily yield an $\Oh^*(4^{\cw})$-time algorithm, but the compatibility matrix for the lower bound only contains a triangular submatrix of size $3 \times 3$, hence we are not able to prove a larger lower bound than for treewidth. Similarly for \COCT, the techniques for \CVC yield an $\Oh^*(14^{\cw})$-time algorithm and a larger lower bound can be proven by adapting the gadgets for \CVC and adding a gadget to detect the used color at join-vertices, however there is again no large enough triangular submatrix that would allow us to show that $\Oh^*(14^{\cw})$ is optimal. For \FVS, a problem with a negative connectivity constraint in the form of acyclicness, the usual cut-and-count approach involves counting the edges induced by a partial solution, but this immediately leads to an XP-algorithm parameterized by clique-width as already noted by Bergougnoux and Kant\'e~\cite{BergougnouxK19a, BergougnouxK21}. Hence, a different approach is required to obtain plausible running times for tight results. 

Beyond benchmark problems, it is crucial for our understanding of the impact of low clique-width on problem complexity to also study large problem classes. Currently, the only such fine-grained result is the study of homomorphism problems by Ganian et al.~\cite{GanianHKOS22}. Very recently, Focke et al.~\cite{FockeMMNSSW23} have shown tight complexity results for a large subset of \textsc{$(\sigma, \rho)$-Dominating Set} problems parameterized by treewidth. By results of Bui-Xuan et al.~\cite{Bui-XuanTV13}, we know that these problems can also be solved in single-exponential time parameterized by clique-width. However, besides a few exceptions such as \IS and \DS the optimal running times parameterized by clique-width are not known. Obtaining results similar to Focke et al.~\cite{FockeMMNSSW23} for parameterization by clique-width would hence be quite interesting.

A big caveat in applying algorithms parameterized by clique-width is that we are lacking good algorithms for computing clique-expressions. The currently best algorithms rely on approximating clique-width via rankwidth, see Oum and Seymour~\cite{OumS06} for the first such algorithm and Fomin and Korhonen~\cite{FominK22} for the most recent one. However, the approximation via rankwidth introduces an exponential error, therefore all single-exponential algorithms parameterized by clique-width become double-exponential algorithms unless we are given a clique-expression by other means. A first step towards better approximation algorithms for clique-width could be a fixed-parameter tractable algorithm with subexponential error.

\bibliography{cliquewidth}

\appendix
\section{Fast Convolution Algorithms}
\label{sec:fast_convolutions}

\subsection{Trimmed Subset Convolution}

\newcommand{\upclos}[1]{{\uparrow\!\!(#1)}}

In this section we describe the details necessary to quickly compute the \emph{cover product} which occurs at the union node for the \CVC algorithm. The states of the \CVC algorithm are given by $\states = \powerset{\bstates} \setminus \{\emptyset, \bstates\}$, where $\bstates = \{\zero, \one_L, \one_R\}$. Essentially, we are given two tables $A, B \colon \states^{[k]} \rightarrow \ZZ_2$ and want to compute the following \emph{cover product} $A \otimes_c B \colon \states^{[k]} \rightarrow \ZZ_2$ in time $\Oh^*(6^k) = \Oh^*(|\states|^k)$:
\begin{equation*}
	(A \otimes_c B)(f) = \sum_{\substack{f_1, f_2 \in \states^{[k]} \colon \\ f_1 \cup f_2 = f }} A(f_1) B(f_2),
\end{equation*}
where the union is \emph{componentwise}.
If not for the exclusion of $\emptyset$ and $\bstates$, a standard application of the fast Zeta and Möbius transform~\cite{CyganFKLMPPS15} would be sufficient. However, this results only in a running time of $\Oh^*(8^k)$. To obtain the improved running time, we trim the computations of the fast Zeta and Möbius transform from above and below. By ordering the subsets of $\bstates$ by their size and performing the computation along this ordering, we can start only with the relevant subsets, i.e., exclude $\emptyset$, and can stop the computation before we reach $\bstates$ itself. All these ideas were used by Björklund et al.~\cite{BjorklundHKK10}, but their presentation is not suited to our setting, instead we provide a suitable self-contained presentation.

We first switch to a more standard notation and consider tables defined over some set family $\family \subseteq \powerset{\universe}$ and finite universe $U$ instead of over functions from $[k]$ to $\states$. A given function $f \colon [k] \rightarrow \states$ is transformed to the set $S_f = \{(i, \state) \sep i \in [k], \state \in f(i)\} = \bigcup_{i \in [k]} \{i\} \times f(i) \subseteq \bigcup_{i \in [k]} \{i\} \times \bstates =: \universe$ and the resulting set family is given by $\family = \{S \subseteq \universe \sep 1 \leq |S \cap (\{i\} \times \bstates)| \leq 2 \,\,\forall i \in [k]\}$. We will see that this set family can be written in a special form that allows us to provide trimmed algorithms for the Zeta and Möbius transform. 

For now, we return to the general setting. Let $\universe$ be some fixed universe and let $\family \subseteq \powerset{\universe}$ be some set family over $\universe$. The \emph{upward closure} $\upclos{\family} \subseteq \powerset{\universe}$ of $\family$ is given by $\upclos{\family} = \{S \subseteq \universe \sep \text{there is a $T \in \family$ s.t. } T \subseteq S\}$; a set family $\family$ is closed under supersets if and only if $\family = \upclos{\family}$. We say that a set family $\family$ is a \emph{closure difference} if there are set families $\family_+$ and $\family_-$ such that $\family = \upclos{\family_+} \setminus \upclos{\family_-}$. Observe that $\upclos{\emptyset} = \emptyset$ and hence $\upclos{\family_+}$ is also a closure difference. 
Closure differences are precisely the set families that contain no holes in the sense of the following lemma, which will be important to prove properties of the various transforms.

\begin{lem}\label{thm:interval_property}
	A set family $\family \subseteq \powerset{\universe}$ is a closure difference if and only if it satisfies the following \emph{interval property}: for all sets $W, T, S \subseteq \universe$ with $W \subseteq T \subseteq S$ and $W, S \in \family$, we must also have $T \in \family$.
\end{lem}

\begin{proof}
  We first show that the interval property holds when $\family$ is a closure difference. So, let $\family = \upclos{\family_+} \setminus \upclos{\family_-}$ and $W \subseteq T \subseteq S$ with $W \in \family$ and $S \in \family$. There exists a set $X \in \family_+$ with $X \subseteq W$, since $W \in \upclos{\family_+}$, and for all sets $Y \in \family_-$ we have $Y \not\subseteq S$, since $S \notin \upclos{\family_-}$. Hence, we have $T \in \upclos{\family_+}$ due to $X \subseteq W \subseteq T$ and no $Y \in \family_-$ can satisfy $Y \subseteq T$, else we would also have $Y \subseteq S$. Therefore, $T \in \upclos{\family_+} \setminus \upclos{\family_-} = \family$.
  
  For the other direction, we first show that $\upclos{\family} \setminus \family$ is an upward closure, i.e., closed under taking supersets. Suppose, for sake of contradiction, that $T \in \upclos{\family} \setminus \family$ and $T \subseteq S$, but $S \notin \upclos{\family} \setminus \family$. Since $T \in \upclos{\family}$, there exists some $W \in \family$ with $W \subseteq T \subseteq S$, hence also $S \in \upclos{\family}$. From $S \notin \upclos{\family} \setminus \family$, it then follows that $S \in \family$. Finally, since $\family$ satisfies the interval property, we must then have $T \in \family$, thus contradicting that $T \in \upclos{\family} \setminus \family$. So, $\upclos{\family} \setminus \family = \upclos{\family_-}$ for some family $\family_-$ and we obtain $\family = \upclos{\family} \setminus (\upclos{\family} \setminus \family) = \upclos{\family} \setminus \upclos{\family_-}$.
\end{proof}

\newcommand{\ring}{R}

Given a table $A \colon \family \rightarrow \ring$, where $\ring$ is some commutative ring with unit, the \emph{Zeta transform} $\zeta A \colon \family \rightarrow \ring$ is given by 
\begin{equation*}
	(\zeta A) (S) = \sum_{\mathclap{\substack{T \in \family \colon \\ T \subseteq S}}} A(T),
\end{equation*}
the \emph{Möbius transform} $\mu A \colon \family \rightarrow \ring$ is given by 
\begin{equation*}
	(\mu A) (S) = \sum_{\mathclap{\substack{T \in \family \colon \\ T \subseteq S}}} (-1)^{|S \setminus T|} A(T),
\end{equation*}
and the \emph{odd-negation transform} $\sigma A \colon \family \rightarrow \ring$ is given by 
\begin{equation*}
	(\sigma A) (S) = (-1)^{|S|} A(S).
\end{equation*}
All of these transforms can be viewed as operators on the space of functions from $\family$ to $\ring$. Given two tables $A, B \colon \family \rightarrow \ring$, their \emph{cover product} $A \otimes_c B \colon \family \rightarrow \ring$ is given by 
\begin{equation*}
	(A \otimes_c B) (S) = \sum_{\mathclap{\substack{T_1, T_2 \in \family \colon \\ T_1 \cup T_2 = S}}} A(T_1) B(T_2).
\end{equation*}
Moreover, we let $A \cdot B \colon \family \rightarrow \ring$ denote the pointwise multiplication of the two tables $A$ and $B$, i.e., $(A \cdot B) (S) = A(S)B(S)$.
We will now prove several properties of these transforms that allow us to design a fast algorithm for computing the cover product.
\begin{lem}\label{thm:transform_properties}
	Let $\family$ be a closure difference and $\ring$ a commutative ring with unit. The following statements are true:
	\begin{enumerate}
		\item $\sigma \zeta \sigma = \mu$, $\sigma \mu \sigma = \zeta$,
		\item $\mu \zeta = \zeta \mu = \mathrm{id}$, where $\mathrm{id}$ is the identity transform,
		\item $\zeta (A \otimes_c B) = (\zeta A) \cdot (\zeta B)$, for any two tables $A,B \colon \family \rightarrow \ring$.
	\end{enumerate}
\end{lem}

\begin{proof}
  We adapt the proofs of Cygan et al.~\cite{CyganFKLMPPS15} to our setting and make note of the proof steps where we use that $\family$ is a closure difference.
	\begin{enumerate}
		\item We compute for every table $A \colon \family \rightarrow \ring$ and set $S \in \family$:
		\begin{align*}
			(\sigma \zeta \sigma A) (S) & = (-1)^{|S|} \sum_{\mathclap{\substack{T \in \family \colon \\ T \subseteq S}}} (-1)^{|T|} A(T) = \sum_{\mathclap{\substack{T \in \family \colon \\ T \subseteq S}}} (-1)^{|S| + |T|} A(T) = \sum_{\mathclap{\substack{T \in \family \colon \\ T \subseteq S}}} (-1)^{|S \setminus T|} A(T) \\
			& = (\mu A) (S),
		\end{align*}
		where the third equality follows from $|S| + |T| = 2|S \cap T| + |S \setminus T|$.
		Since $\sigma \sigma = \mathrm{id}$, it also follows that $\sigma \mu \sigma = \sigma \sigma \zeta \sigma \sigma = \zeta$.
		\item From the previous statement it follows that $\mu \zeta = \sigma \zeta \sigma \zeta$. Again, we compute for every table $A \colon \family \rightarrow \ring$ and set $S \in \family$:
		\begin{align*}
		 (\sigma \zeta \sigma \zeta A) (S) & = (-1)^{|S|} \sum_{\mathclap{\substack{T \in \family \colon \\ T \subseteq S}}} (\sigma \zeta A)(T) = (-1)^{|S|} \sum_{\mathclap{\substack{T \in \family \colon \\ T \subseteq S}}} (-1)^{|T|} \sum_{\mathclap{\substack{W \in \family \colon \\ W \subseteq T}}} A(W) \\
		  & = (-1)^{|S|} \sum_{\mathclap{\substack{W \in \family \colon \\ W \subseteq S}}} A(W) \sum_{\mathclap{\substack{T \in \family \colon \\ W \subseteq T \subseteq S}}} (-1)^{|T|} = A(S) + (-1)^{|S|} \sum_{\mathclap{\substack{W \in \family \colon \\ W \subsetneq S}}} A(W) \sum_{\mathclap{\substack{T \in \family \colon \\ W \subseteq T \subseteq S}}} (-1)^{|T|} \\
		  & = A(S),
		\end{align*}
		where the last equality follows from the fact that $\sum_{T \in \family \colon W \subseteq T \subseteq S} (-1)^{|T|} = 0$ whenever $W \subsetneq S$, since we can pick some $x \in S \setminus W$ and pair every $T$ with $T \symdiff \{x\}$ which yields a fixpoint-free sign-reversing involution. Note that this step relies on \cref{thm:interval_property} (applied to $W \subseteq T \symdiff \{x\} \subseteq S$). For an arbitrary set family $\family'$, it does not necessarily hold that $T \symdiff \{x\} \in \family'$.
		
		The remaining equality follows from $\zeta \mu = \zeta \sigma \zeta \sigma = \sigma \mu \sigma \sigma \zeta \sigma = \sigma \mu \zeta \sigma = \sigma \sigma = \mathrm{id}$.
		\item We compute for every two tables $A, B \colon \family \rightarrow \ring$ and set $S \in \family$:
		\begin{align*}
			(\zeta (A \otimes_c B))(S) & = \sum_{\substack{W \in \family \colon \\ W \subseteq S}} \sum_{\substack{T_1, T_2 \in \family \colon \\ T_1 \cup T_2 = W}} A(T_1) B(T_2)  = \sum_{\mathclap{\substack{T_1, T_2 \in \family \colon \\ T_1 \cup T_2 \subseteq S}}} A(T_1) B(T_2) \\
			& = \sum_{\substack{T_1 \in \family \colon \\ T_1 \subseteq S}} \sum_{\substack{T_2 \in \family \colon \\ T_2 \subseteq S}} A(T_1) B(T_2) = \left( \sum_{\substack{T \in \family \colon \\ T \subseteq S}} A(T) \right) \left( \sum_{\substack{T \in \family \colon \\ T \subseteq S}} B(T) \right) \\
			& = ((\zeta A) \cdot (\zeta B))(S),
		\end{align*} 
		where the second equality again relies on \cref{thm:interval_property} (applied to $T_1 \subseteq T_1 \cup T_2 \subseteq S$). For an arbitrary set family $\family'$, we might have $T_1 \cup T_2 \notin \family'$ which would therefore not be summed over on the left-hand side. \qedhere
	\end{enumerate}
\end{proof}

Equipped with the relations between the various transforms from \cref{thm:transform_properties}, we can now proceed with the algorithmic part. However, we need that the set family $\family$ satisfies some algorithmic requirements. We say that $\family \subseteq \powerset{\universe}$ is \emph{efficiently listable} if there is an algorithm that outputs all members of $\family$ in time $\Oh^*(|\family|) = |\family| |\universe|^{\Oh(1)}$. Each member of $\family$ is represented as a bitstring and we can assume that the output of the listing algorithm is in sorted order. Hence, after listing $\family$, we can decide for a set $S \subseteq \universe$ whether $S \in \family$ or not by binary search in time $\Oh(\log |\family|) = \Oh(|\universe|)$. A table $A \colon \family \rightarrow \ring$ is given by listing the values $A(S)$, $S \in \family$, in the same order as the listing of $\family$. 

In the following running time analysis, we separate the impact of the ring operations from the rest of the algorithms, i.e., we only count the number of performed ring operations and bound the time spent on the remainder of the algorithm. 

\begin{cor}\label{thm:transform_implies_product}
 Let $\family$ be an efficiently listable closure difference and $\ring$ be a commutative ring with unit. If the Zeta transform $\zeta A$ and the Möbius transform $\mu A$ of a table $A \colon \family \rightarrow \ring$ can be computed in $\Oh(|\family||\universe|)$ ring operations and additional time $\Oh^*(|\family|)$, then the cover product $A \otimes_c B$ can be computed in $\Oh(|\family||\universe|)$ ring operations and additional time $\Oh^*(|\family|)$.
\end{cor}

\begin{proof}
  We make use of \cref{thm:transform_properties}. Given tables $A$ and $B$, we first compute $\zeta A$ and $\zeta B$ and then the pointwise multiplication $(\zeta A) \cdot (\zeta B)$. By definition, the pointwise multiplication can be computed in $\Oh(|\family|)$ ring operations and additional time $\Oh^*(|\family|)$ for listing all members of $\family$. Finally, we compute the Möbius transform of the pointwise multiplication. By \cref{thm:transform_properties}, we have $\mu ((\zeta A) \cdot (\zeta B)) = \mu \zeta (A \otimes_c B) = A \otimes_c B$. By assumption, every step takes $\Oh(|\family||\universe|)$ ring operations and additional time $\Oh^*(|\family|)$, hence the statement follows.
\end{proof}

\newcommand{\setlist}{{\mathcal{L}}}

\begin{thm}\label{thm:fast_transform}
  Let $\family = \upclos{\family_+} \setminus \upclos{\family_-}$ be an efficiently listable closure difference and $\ring$ be a commutative ring with unit. Given a table $A \colon \family \rightarrow \ring$, the Zeta transform $\zeta A$ and the Möbius transform $\mu A$ can be computed in $\Oh(|\family||\universe|)$ ring operations and additional time $\Oh^*(|\family|)$.
\end{thm}

\begin{proof}
  We modify the algorithm of Björklund et al.~\cite[Algorithm Z]{BjorklundHKK10} for the Zeta transform to make it suitable for our setting. Let $n = |\universe|$ be the size of the universe and without loss of generality assume $\universe = [n]$. The algorithm maintains $n+1$ set families $\setlist_0, \ldots, \setlist_n \subseteq \family$, where each $\setlist_i$ only contains sets of size $i$. For all $j \in [0, n]$ and $S \in \family$, we will compute auxiliary values
  \begin{equation*}
    A_j(S) = \sum_{\substack{T \in \family \colon T \cap [j] \subseteq S \cap [j], \\ T \cap [j+1,n] = S \cap [j+1,n]}} A(T).
  \end{equation*}
  Clearly, $A_0(S) = A(S)$ and $A_n(S) = \zeta A(S)$. Also, one can verify that the recurrence $A_j(S) = [(j \in S) \wedge (S\setminus\{j\} \in \family)]A_{j-1}(S\setminus\{j\}) + A_{j-1}(S)$ holds for all $j \in [n]$ and $S \in \family$. We highlight the case that $j \in S$ but $S \setminus \{j\} \notin \family$; here every $T \in \family$ with $T \subseteq S$ must satisfy $j \in T$, otherwise we would obtain $S \setminus \{j\} \in \family$ by \cref{thm:interval_property}, hence $A_j(S) = A_{j-1}(S)$.

  \begin{algorithm}[h]
    \lFor{$S \in \family$}{insert $S$ into $\setlist_{|S|}$}
    \For{$r=0,1,\ldots,n$}{
      \While{$\setlist_r$ is not empty}{
        Select any $S \in \setlist_r$ and remove it from $\setlist_r$\;
        $A_0(S) := A(S)$\;
        \For{$j=1,\ldots,n$}{
          $A_j(S) := [(j \in S) \wedge (S\setminus\{j\} \in \family)]A_{j-1}(S\setminus\{j\}) + A_{j-1}(S)$\;
        }
      }
    }
    \Return $A_n$\;
    \caption{Fast Zeta transform}
    \label{algo:fastzeta}
  \end{algorithm}
  
  Clearly, \cref{algo:fastzeta} considers every set $S \in \family$ exactly once and due to the ordering by cardinality only accesses already computed values. \cref{algo:fastzeta} correctly computes the Zeta transform $\zeta A = A_n$ by the previous considerations. Ring operations are only performed in line 7, namely at most two per execution of line 7. Line 7 is executed exactly $n$ times for every element of $\family$, hence \cref{algo:fastzeta} performs in total $\Oh(|\family|n)$ ring operations. The additional time is dominated by the listing of $\family$ in line 1 and the $\Oh(|\family|n)$ membership queries incurred by line 7, thus leading to additional time $\Oh^*(|\family|)$.
  
  To compute the Möbius transform $\mu A$, we use $\mu = \sigma \zeta \sigma$ from \cref{thm:transform_properties}. So, we first compute $\sigma A$, then apply \cref{algo:fastzeta} to $\sigma A$ to obtain $\zeta \sigma A$, and finally we apply $\sigma$ again. Since $(-1)^k = 1$ if $k$ is even and $(-1)^k = -1$ if $k$ is odd, we can apply $\sigma$ in $\Oh(|\family|)$ ring operations and additional time $\Oh^*(|\family|)$ due to listing, therefore the statement follows.
\end{proof}

\newcommand{\proj}{\pi}

\subparagraph*{Componentwise Union to Set Union.} Fix a natural number $k$ and universe $\universe$. We define the \emph{projection} $\proj_\universe \colon [k] \times \universe \rightarrow U$, $(i,u) \mapsto u$, and for any $i \in [k]$ and set $S \subseteq \universe$, we define $\proj_\universe^i(S) := \proj_\universe(S \cap (\{i\} \times \universe))$. Given a set family $\family \subseteq \powerset{\universe}$, we construct the set family $k\family := \{S \subseteq [k] \times \universe \sep \proj_\universe^i(S) \in \family \text{ for all } i \in [k]\}$. Given a function $f \colon [k] \rightarrow \family$, we construct the set $k\family \ni S_f := \bigcup_{i \in [k]} \{i\} \times f(i) = \{(i,x) \in [k] \times \universe \sep x \in f(i)\}$. This is a bijection and the construction turns componentwise union of functions $[k] \rightarrow \family$ into union of sets in $k\family$, i.e., for functions $f, f_1, f_2 \colon [k] \rightarrow \family$, we have $f(i) = f_1(i) \cup f_2(i)$ for all $i \in [k]$ if and only if $S_f = S_{f_1} \cup S_{f_2}$. Hence, given two tables $A,B\colon \family^{[k]} \rightarrow \ring$, the \emph{componentwise cover product} $A \otimes_{comp} B \colon \family^{[k]} \rightarrow \ring$ can be reduced to a standard cover product, i.e.,
\begin{equation*}
	(A \otimes_{comp} B)(f) := \sum_{\mathclap{\substack{f_1, f_2 \in \family^{[k]} \colon \\ f_1(i) \cup f_2(i) = f(i) \,\forall i \in [k]}}} A(f_1) B(f_2) = \sum_{\substack{S_{f_1}, S_{f_2} \in k\family \colon \\ S_{f_1} \cup S_{f_2} = S_f }} A(f_1) B(f_2) = (A' \otimes_c B')(S_f),
\end{equation*}
where $A', B'\colon k\family \rightarrow \ring$ with $A'(S_f) := A(f)$ and $B'(S_f) := B(f)$, which is well-defined since $f \mapsto S_f$ is a bijection. Furthermore, the following lemma shows that we can apply the fast convolution algorithms to $k\family$ if $\family$ is a closure difference.

\begin{lem}\label{thm:disjoint_union_closure_difference}
	Let $\family \subseteq \powerset{\universe}$ be a set family and $k$ be some natural number. If $\family$ is a closure difference, then also the set family $k\family \subseteq \powerset{[k] \times \universe}$ is a closure difference.
\end{lem}

\begin{proof}
	We first argue that $k \upclos{\family} = \upclos{k\family}$ for all set families $\family$ and natural numbers $k$. We have $S \in k \upclos{\family}$ if and only if there exist $T_i \in \family$ such that $\{i\} \times T_i \subseteq S\cap(\{i\} \times \universe)$ for all $i \in [k]$. Setting $T = \bigcup_{i \in [k]} \{i\} \times T_i \in k\family$, we see that this is equivalent to the existence of some $T \subseteq S$ which holds if and only if $S \in \upclos{k\family}$.
	
	Let $\family = \upclos{\family_+} \setminus \upclos{\family_-}$ be a closure difference. We compute
	\begin{align*}
		k\family & = \{S \subseteq [k] \times \universe \sep \proj_\universe^i(S) \in \family \tfa i \in [k]\} \\
		& = \{S \subseteq [k] \times \universe \sep \proj_\universe^i(S) \in (\upclos{\family_+} \setminus \upclos{\family_-}) \tfa i \in [k]\} \\
		& = k\upclos{\family_+} \setminus \{S \subseteq [k] \times \universe \sep \text{ there is an $i \in [k]$ s.t. } \proj_\universe^i(S) \in \upclos{\family_-}\} \\
		& = \upclos{k\family_+} \setminus \{S \subseteq [k] \times \universe \sep \text{ there is an $i \in [k]$ s.t. } \proj_\universe^i(S) \in \upclos{\family_-}\}
	\end{align*}
	and notice that the second set in the last line is clearly closed under taking supersets and hence $k\family$ is a closure difference, too.
\end{proof}

\begin{thm}\label{thm:fast_compwise_cover_product}
	Let $\family \subseteq \powerset{\universe}$ be a fixed efficiently listable closure difference and $\ring$ be a commutative ring with unit. Given tables $A, B \colon \family^{[k]} \rightarrow \ring$, their componentwise cover product $A \otimes_{comp} B$ can be computed in $\Oh(k|\family|^k |\universe|)$ ring operations and additional time $\Oh^*(|\family|^k)$.
\end{thm}

\begin{proof}
	We construct the set family $k \family \subseteq \powerset{[k] \times \universe}$ which is a closure difference by \cref{thm:disjoint_union_closure_difference}. Since $\family$ is efficiently listable, also $k \family$ can be efficiently listed in lexicographic order. We invoke \cref{thm:fast_transform} and \cref{thm:transform_implies_product} with $k \family$ to compute the cover product over $k \family$ in $\Oh(k|\family|^k |U|)$ ring operations and additional time $\Oh^*(|\family|^k)$. By the preceding discussion, the cover product over $k \family$ yields the componentwise cover product $A \otimes_{comp} B$ over $\family$.
\end{proof}


\subsection{Lattice-based Convolution}

A \emph{poset} is a pair $(\lattice, \preceq)$ consisting of a set $\lattice$ and a binary relation $\preceq$ on $\lattice$ that is reflexive, transitive, and anti-symmetric. A \emph{lattice} is a poset $(\lattice, \preceq)$ such that every pair $a, b \in \lattice$ has a greatest lower bound (\emph{meet}) $a \wedge b \in \lattice$ and a least upper bound (\emph{join}) $a \vee b \in \lattice$. Any finite lattice $(\lattice, \preceq)$ contains a $\preceq$-minimum element $\hat{0} \in \lattice$, which is obtained by taking the meet of all elements in $\lattice$ and satisfies $\hat{0} \vee a = a \vee \hat{0} = a$ for all $a \in \lattice$.

In the algorithm for \CDS, the set of possible states forms a lattice and at union-nodes in the clique-expression, we must compute a convolution-like product. To obtain an efficient algorithm, we will observe some lattice-theoretic properties. 

The product that we are interested in can be formulated in the lattice setting as follows. Given a lattice $(\lattice, \preceq)$ and tables $A, B\colon \lattice \rightarrow \FF$, where $\FF$ is some field, the \emph{$\vee$-product} $A \otimes_{\lattice} B\colon \lattice \rightarrow \FF$ is given by
$(A \otimes_{\lattice} B)(x) = \sum_{y, z \in \lattice \colon x = y \vee z} A(y)B(z)$
for every $x \in \lattice$.

Björklund et al.~\cite{BjorklundHKKNP16} develop an efficient algorithm for the $\vee$-product for specific lattices by designing small arithmetic circuits for the Zeta and Möbius transform, whose precise definitions we do not need here. The relevant concept is as follows; we say that an element $x \in \lattice$ of a lattice $(\lattice, \preceq)$ is \emph{join-irreducible} if $x = a \vee b$ implies $x = a$ or $x = b$ for all $a, b \in \lattice$, otherwise $x$ is called \emph{join-reducible}. We denote the set of join-irreducible elements in $\lattice$ by $\lattice_\vee$. Observe that $\hat{0}$ is always join-irreducible, as otherwise $\hat{0}$ would not be the $\preceq$-minimum. We even have the stronger property that $\hat{0} = a \vee b$ implies $a = b = \hat{0}$.

We assume that a finite lattice $(\lattice, \preceq)$ is algorithmically given to us in the \emph{join representation}~\cite{BjorklundHKKNP16}; we are given the set $\lattice$, where the elements of $\lattice$ represented as $\Oh(\log |\lattice|)$-bit strings, the set of join-irreducible elements $\lattice_\vee \subseteq \lattice$, and an algorithm $\algo_\lattice$ that computes the join $a \vee x$ given an element $a \in \lattice$ and a join-irreducible element $x \in \lattice_\vee$. 

\begin{thm}[\cite{BjorklundHKKNP16}]\label{thm:few_irreducibles_fast}
	Let $(\lattice, \preceq)$ be a finite lattice given in join-representation and $A, B\colon \lattice \rightarrow \FF$ be two tables, where $\FF$ is some field. The $\vee$-product $A \otimes_\lattice B$ can be computed in $\Oh(|\lattice| |\lattice_\vee|)$ field operations and calls to algorithm $\algo_\lattice$ and further time $\Oh(|\lattice| |\lattice_\vee|^2)$.
\end{thm}

Next, we analyze the lattice that occurs in the algorithm for \CDS and derive a bound on the number of join-irreducible elements in this lattice, so that we can apply \cref{thm:few_irreducibles_fast}. The relevant lattice can be written as a power of a smaller lattice and we give a general bound for such lattices. We proceed with the relevant definitions. 

Given finitely many lattices $(\lattice_i, \preceq_i)$, $i \in \idxset$, their \emph{direct product} $(\prod_{i \in \idxset} \lattice_i, \preceq)$, with $(a_i)_{i \in \idxset} \preceq (b_i)_{i \in \idxset}$ if and only if $a_i \preceq_i b_i$ for all $i \in \idxset$, is again a lattice; the join- and meet-operations in the direct product lattice are given by componentwise application of the corresponding operation in the constituent lattices. Given a lattice $(\lattice, \preceq)$ and a natural number $k$, the \emph{$k$th-power} $(\lattice^k, \preceq^k)$ of $\lattice$ is the direct product of $k$ copies of $\lattice$.


\begin{lem}\label{thm:power_irreducibles}
	Let $(\lattice, \preceq)$ be a finite lattice and $k$ be a natural number. In the $k$-th power $(\lattice^k, \preceq^k)$ of $\lattice$, an element $x \in \lattice^k$ is join-irreducible if and only if $x = (\hat{0}, \ldots, \hat{0})$ or there is exactly one component that is not $\hat{0}$ and this component is join-irreducible in $\lattice$. In particular, there are exactly $1 + k(|\lattice_\vee| - 1) \leq k |\lattice|$ join-irreducible elements in $\lattice^k$.
\end{lem}

\begin{proof}
	First, we prove that every join-irreducible element $x$ of $\lattice^k$ must have the stated form. Suppose to the contrary that $x = (x_1, \ldots, x_k)$ has at least two components that are not $\hat{0}$; without loss of generality we can assume that $x_1 \neq \hat{0}$ and $x_2 \neq \hat{0}$. We compute that $x = (x_1, \ldots, x_k) = (\hat{0}, x_2, x_3, \ldots, x_k) \vee (x_1, \hat{0}, x_3, \ldots, x_k)$ and note that $x \neq (\hat{0}, x_2, x_3, \ldots, x_k)$ and $x \neq (x_1, \hat{0}, x_3, \ldots, x_k)$ by assumption. Therefore, $x$ is join-reducible in this case. Furthermore, if $x$ contains a join-reducible component, say $x_1 = a_1 \vee b_1$ with $a_1 \neq x_1 \neq b_1$, then $x$ is join-reducible, since $x = (a_1, x_2, \ldots, x_k) \vee (b_1, x_2, \ldots, x_k)$.
	
	For the other direction, we have $x = (\hat{0}, \ldots, \hat{0})$ which is the $\preceq^k$-minimum element in $\lattice^k$ and hence join-irreducible. Now, let $x = (x_1, \ldots, x_k) \in \lattice^k$ have exactly one component that is not $\hat{0}$ and let this component be join-irreducible, say $x_1 \in \lattice_\vee \setminus \{\hat{0}\}$ and $x_i = \hat{0}$ for all $i \in [2,k]$. Suppose that $x = y \vee z$ with $y,z \in \lattice^k$, then we have $\hat{0} = x_i = y_i \vee z_i$ for all $i \in [2,k]$ which implies $x_i = y_i = z_i = \hat{0}$ for all $i \in [2,k]$. For $i = 1$, we see that $x_1 = y_1 \vee z_1$ implies $x_1 = y_1$ or $x_1 = z_1$ by irreducibility of $x_1$, and hence $x = y$ or $x = z$.
\end{proof}

\begin{cor}\label{thm:power_join_rep}
	Let $(\lattice, \preceq)$ be a finite lattice given in join-representation and $k$ be a natural number. There is an algorithm $\algo_{\lattice^k}$ that computes the join $a \vee x$, where $a \in \lattice^k$ and $x \in (\lattice^k)_\vee$, using one call to $\algo_\lattice$ and further time $\Oh(k \log |\lattice|)$.
\end{cor}

\begin{proof}
	Every element of $\lattice$ is represented by a $\Oh(\log |\lattice|)$-bit string, hence every element of $\lattice^k$ can be represented by a $\Oh(k \log |\lattice|)$-bit string. By \cref{thm:power_irreducibles}, we know that the join-irreducible element $x \in \lattice^k$ has at most one component that is not $\hat{0}$. We search for this component and the corresponding component in $a$ in time $\Oh(k \log |\lattice|)$, afterwards we call $\algo_\vee$ to obtain the result for this component and all other components of $a$ remain unchanged due to $a_i \vee \hat{0} = a_i$ for all $i \in [k]$.
\end{proof}

\begin{cor}\label{thm:power_lattice_fast_product}
	Let $(\lattice, \preceq)$ be a finite lattice given in join-representation and $k$ be a natural number. Given two tables $A, B\colon \lattice^k \rightarrow \ZZ_2$, the $\vee$-product $A \otimes_{\lattice^k} B$ in $\lattice^k$ can be computed in time $\Oh(k^2 |\lattice|^{k+2})$ and $\Oh(k|\lattice|^{k+1})$ calls to algorithm $\algo_\lattice$.
\end{cor}

\begin{proof}
	We pipeline \cref{thm:few_irreducibles_fast} with \cref{thm:power_irreducibles} and \cref{thm:power_join_rep}. First, observe that the field operations in $\ZZ_2$ can be performed in constant time. Secondly, every call to $\algo_{\lattice^k}$ can be simulated by one call to $\algo_\lattice$ and further time $\Oh(k \log |\lattice|)$ and hence these calls lead to in total $\Oh(|\lattice^k| |(\lattice^k)_\vee| (k \log |\lattice|)) = \Oh(k^2 |\lattice|^{k+1} \log |\lattice|)$ further running time, which will be dominated by the rest of the algorithm. Finally, the algorithm of \cref{thm:few_irreducibles_fast} needs further time $\Oh(|\lattice^k| |(\lattice^k)_\vee|^2) = \Oh(|\lattice|^k (k |\lattice|)^2) = \Oh(k^2 |\lattice|^{k+2})$ which dominates all other computations.
\end{proof}

\end{document}